\documentclass[a4paper,10pt,reqno]{amsart}
          \usepackage{amssymb}
	  \usepackage{amsmath}
          \usepackage{amsfonts}
          \usepackage[english]{babel}
          \usepackage[utf8]{inputenc}

\usepackage[margin=2.5cm]{geometry}

\usepackage{color}
 \usepackage[unicode,colorlinks,plainpages=false,hyperindex=true,bookmarksnumbered=true,bookmarksopen=true,pdfpagelabels]{hyperref}
 \hypersetup{urlcolor=cyan,linkcolor=blue,citecolor=red,colorlinks=true} 
\makeatletter
\renewcommand*{\p@subsection}{\S\,}
\renewcommand*{\p@subsubsection}{\S\,}
\makeatother

\newtheorem{thm}{Theorem}[section]

\newtheorem{lem}[thm]{Lemma}
\newtheorem{prop}[thm]{Proposition}

\newtheorem{conj}[thm]{Conjecture}
\theoremstyle{definition}

\theoremstyle{remark}
\newtheorem{rem}{Remark}[section]

\numberwithin{equation}{section}

\newcommand{\CC}{\ensuremath{\mathbb{C}}}
\newcommand{\N}{\ensuremath{\mathbb{N}}}

\newcommand{\Z}{\ensuremath{\mathbb{Z}}}

\newcommand{\tr}{\operatorname{tr}}
\newcommand{\diag}{\operatorname{diag}}

\newcommand{\Mat}{\operatorname{Mat}}
\newcommand{\Id}{\operatorname{Id}}
\newcommand{\Gl}{\operatorname{GL}}

\newcommand{\LL}{L}

\newcommand{\Rep}{\operatorname{Rep}}
\newcommand\dgal[1]{  \left\{\!\!\left\{#1\right\}\!\!\right\} }
\newcommand\br[1]{\{ #1 \}} 

\def\aalpha{{\bar{\alpha}}}

\def\h{\mathfrak{h}}
\def\hreg{\mathfrak{h}_{\mathrm{reg}}}

\def\VV{\mathcal V}

\def\Cnqd{\mathcal{C}_{n, d, q}}
\def\Cnqdt{\mathcal{C}_{n, d, q}^{\times}}

\newcommand{\aaa}{a}
\newcommand{\ccc}{b}

\newcommand{\As}{A}
\newcommand{\Bs}{B}
\newcommand{\MM}{\ensuremath{\mathcal{M}}_{n, d, q}}
\newcommand{\MMt}{\ensuremath{\mathcal{M}}^{\times}_{n, d, q}}

\newcommand{\Qal}{\ensuremath{\mathcal{Q}}}
\newcommand{\Qd}{\ensuremath{\mathcal{Q}_d}}

\newcommand{\Aal}{\ensuremath{\mathcal{A}}}
\newcommand{\Aalt}{\ensuremath{\mathcal{A}^{\times}}}

\begin{document}

\title[Hamiltonian formulation of the spin Ruijsenaars--Schneider system]  
{On the Hamiltonian formulation of the trigonometric spin Ruijsenaars--Schneider system}   

\author{Oleg Chalykh}
 \address[Oleg Chalykh]{School of Mathematics\\
         University of Leeds\\
         Leeds, LS2 9JT, UK}
 \email{o.chalykh@leeds.ac.uk}

\author{Maxime Fairon}
 \address[Maxime Fairon]{School of Mathematics and Statistics, University of Glasgow\\
University Place, G12 8QQ Glasgow, UK}
 \email{maxime.fairon@glasgow.ac.uk}


\begin{abstract}
We suggest a Hamiltonian formulation for the spin Ruijsenaars-Schneider system in the trigonometric case. 
Within this interpretation, the phase space is obtained by a quasi-Hamiltonian reduction performed on (the cotangent bundle to) a representation space of a framed Jordan quiver. For arbitrary quivers, analogous varieties were introduced by Crawley-Boevey and Shaw, and their interpretation as quasi-Hamiltonian quotients was given by Van den Bergh. Using Van den Bergh's formalism, we construct commuting Hamiltonian functions on the phase space and identify one of the flows with the spin Ruijsenaars--Schneider system. We then calculate all the Poisson brackets between local coordinates, thus answering an old question of Arutyunov and Frolov. We also construct a complete set of commuting Hamiltonians and integrate all the flows explicitly.
\end{abstract}

\maketitle

 \setcounter{tocdepth}{2}

\section{Introduction}  \label{intro}

The spin Ruijsenaars--Schneider model (RS model) has been introduced by Krichever and Zabrodin \cite{KrZ} as a 
generalisation of the well-known non-spin model \cite{RS86}. It is a system of $n$ particles on the line, with positions $q_i$ 
and spin variables $\aaa_i^\alpha, \ccc_i^\alpha$, where $i=1, \ldots, n$ labels the particles and $\alpha=1, \ldots, d$ labels the internal degrees of freedom.   
Its equations of motion have the following form (cf. \cite{AF}):
\begin{subequations}
 \begin{align}
  \dot{q}_i=&{f}_{ii} \,,\label{AFq} \\ 
\dot{\aaa}^\alpha_i=&
\sum_{k\neq i} V(q_{ik}){f}_{ik} (\aaa_k^\alpha-\aaa_i^\alpha) \,,\label{AFa} \\
\dot{\ccc}_i^\alpha=&\sum_{k\neq i}
\left(\,
V(q_{ik}){f}_{ik}\ccc_i^\alpha -V(q_{ki}){f}_{ki}\ccc_k^\alpha  
\,\right)
\,. \label{AFc}
 \end{align}
\end{subequations}
Here $f_{ij}=\sum_\alpha \aaa_i^\alpha \ccc_j^\alpha$ and $V(z)=\zeta(z)-\zeta(z+\gamma)$ where $\zeta(z)$ is the  Weierstrass $\zeta$-function, $\gamma$ is a coupling parameter, and $q_{ij}:=q_i-q_j$. In the equations \eqref{AFq}--\eqref{AFc} it is assumed that 
$\sum_\alpha \aaa_i^\alpha=1$ for all $i$: it is easy to see that such constraints are preserved by the flow. This reduces the dimension of the phase space to $2nd$. 
The rational and trigonometric (or, rather, hyperbolic) versions are obtained by setting $V(z)=z^{-1}-(z+\gamma)^{-1}$ and $V(z)=\coth(z)-\coth(z+\gamma)$, respectively. Note that in \cite{KrZ} the equations are written in $2n+2nd$ coordinates, and the above form is obtained after a reduction, see \cite{KrZ, AF} for the details (in the notation of \cite{AF}, our variables $\aaa_i^\alpha, \ccc_i^\alpha$ correspond to $\mathbf{a}_i^\alpha, \mathbf{c}_i^\alpha$).

In \cite{KrZ} the above dynamical model was derived by studying solutions of the non-abelian 2D Toda equation. As shown in \cite{KrZ}, the model admits a Lax representation, and its general solution can be expressed in Riemann theta functions. Similar results for the Calogero--Moser model and the KP equation were previously obtained in \cite{GH, Kr80, KBBT}. It is therefore natural to ask for a Hamiltonian formulation of the spin RS model. In the rational case, the answer was given by Arutyunov and Frolov in \cite{AF}. To formulate their result, let us
define the antisymmetric bracket (bivector) on the phase space of the system \eqref{AFq}--\eqref{AFc} by the formulas 
\begin{subequations}
 \begin{align}
  \br{q_i,q_j} = &\, 0\,,\quad \br{q_i, \aaa_j^\alpha}=0\,,\quad 
\br{q_i, \ccc_j^\alpha}=\delta_{ij}\ccc_j^\alpha \,, \label{AF1} \\
\br{\aaa_i^\alpha,\aaa_j^\beta} = &\, \delta_{(i\neq j)}\frac{1}{q_{ij}}
(\aaa_i^\alpha \aaa_j^\beta+\aaa_i^\beta\aaa_j^\alpha -\aaa_i^\alpha \aaa_i^\beta - 
\aaa_j^\alpha \aaa_j^\beta) \,, \label{AF2} \\
\br{\aaa_i^\alpha, \ccc_j^\beta} = & \,\aaa_i^\alpha \LL_{ij} -\delta_{\alpha\beta}\LL_{ij}-
 \delta_{(i\neq j)}\frac{1}{q_{ij}}(\aaa_i^\alpha-\aaa_j^\alpha) \ccc_j^\beta \,, \label{AF3} \\
\br{\ccc_i^\alpha, \ccc_j^\beta}= &\,
 \delta_{(i\neq j)}\frac{1}{q_{ij}} (\ccc_i^\alpha \ccc_j^\beta + \ccc_i^\beta\ccc_j^\alpha) 
-\ccc_i^\alpha \LL_{ij} + \ccc_j^\beta \LL_{ji}  \,. \label{AF4}
 \end{align}
\end{subequations}
Here $\LL_{ij}=\frac{{f}_{ij}}{q_{ij}+\gamma}$, and $\delta_{(i\neq j)}:=1-\delta_{ij}$. The above bracket can be viewed on the space of dimension $2nd+n$ with coordinates $q_i, \aaa_i^\alpha, \ccc_j^\beta$; it is then easy to check that the ideal generated by the $n$ functions $\sum_\alpha \aaa_i^\alpha-1$ is also an ideal with respect to this bracket, and so the bracket restricts onto the phase space of dimension $2nd$.  With these definitions, it is shown in \cite{AF} that the formulas \eqref{AF1}--\eqref{AF4} define a Poisson bracket on the phase space of dimension $2nd$, and the Hamiltonian flow with the Hamiltonian $h=\sum_i f_{ii}$ takes the form \eqref{AFq}--\eqref{AFc} with $V(z)=z^{-1}-(z+\gamma)^{-1}$. 

Note that the above $L=(\LL_{ij})$ is the Lax matrix of the rational RS model, 
and $h=\gamma\tr L$. Arutyunov and Frolov observed that the brackets between $L_{ij}$ admit an $r$-matrix formulation with the same $r$-matrix as for the non-spin RS model. 
This led them to conjecture a similar result in the trigonometric case. However, they were unable to prove that conjecture and to find a Hamiltonian interpretation for that case, i.e. a trigonometric analogue of \eqref{AF1}--\eqref{AF4} remained unknown.  
One of the main results of this paper is an explicit description of the appropriate brackets between the variables $(q_i,\aaa_i^\alpha,\ccc_i^\alpha)$ and a proof of the Arutyunov--Frolov's conjecture as a corollary.

We should mention that a Hamiltonian interpretation of the model \eqref{AFq}--\eqref{AFc} was found by Krichever in \cite{Kr} within his elegant geometric approach to the systems of Calogero--Moser type. Nonetheless, his formula for the symplectic form (see \cite[(3.21)]{Kr}) is implicit and, while he constructs the action-angle coordinates, the Poisson brackets between the original variables are not determined. (In the special case of $n=2$, the symplectic form has been calculated explicitly \cite{S09}.) Another issue is that the symplectic structure in \cite{Kr} is constructed on a reduced phase space which is obtained by gauging away extra $d(d-1)$ degrees of freedom. However, such reduction is only valid for $d\le n$. In comparison, our Hamiltonian formulation is given in the original $2nd$ coordinates, it is completely explicit and remains valid for any $d$. Another advantage of our approach is that it provides a \emph{completion} of the original phase space of the trigonometric spin RS model, in the same way as the Calogero--Moser space in \cite{W1} allows coalescence of particles. To this end, we demonstrate the (degenerate) integrability of the system and explicitly integrate all the flows, from which the completeness of the flows is obvious.

Let us say a few words about our methods. In our approach, the phase space of the system is obtained by quasi-Hamiltonian reduction from a representation space of a framed Jordan quiver, using the framework developed in \cite{CBShaw, VdB1}. This is very natural from the viewpoint of the existing results for the Calogero--Moser \cite{W1, W2, BP, T15, CS} and the (non-spin) RS systems \cite{FockRosly, Oblomkov,FK12,FKlu, CF}. 
Once a geometric model is correctly identified, the remaining task is to confirm that by calculating the Poisson brackets and flows in suitable local coordinates. This then becomes a natural extension of the methods and results of our previous work \cite{CF}. Note that we work in the holomorphic setting, and so while we refer to the system under consideration as to the trigonometric RS model, there is no actual difference between the trigonometric and hyperbolic versions.
The quasi-Hamiltonian reduction framework also offers a nice perspective on the integrability of the system. First, it allows us to establish its degenerate integrability in a very natural way. To further extend it to a Liouville integrable system is a non-trivial problem, which in general does not have a canonical answer. Our solution to this problem can be viewed as an analogue of the Gelfand--Tsetlin integrable system. 
It would be interesting to find a quantum version of this integrable system. There are two possible approaches to this, either by using the double affine Hecke algebras as in \cite{U} or by using quantized multiplicative quiver varieties due to D.~Jordan \cite{Jo}. 

The paper is organised as follows. Section \ref{main results} outlines our main results. In Section \ref{SqHpicture}, we describe all the necessary ingredients for performing quasi-Hamiltonian reduction in the special case of a framed Jordan quiver, including the corresponding representation spaces \cite{CBShaw} and (double) quasi-Poisson brackets \cite{VdB1}. In Section \ref{Slocal} we describe local coordinates on the constructed quasi-Hamiltonian quotients and express the Poisson brackets in the local coordinates. As a corollary, this provides the Hamiltonian formulation for the spin RS model and proves the conjecture from \cite{AF}. Section \ref{First} is devoted to the integrability of the model. We show that the trigonometric spin RS system is {degenerately} integrable, and then extend it to a completely integrable system and explicitly integrate all the flows. We also discuss the relationship with the results of Krichever and Zabrodin \cite{KrZ}. Note that in Sections \ref{Slocal}, \ref{First} a modification of the spin RS model is also considered, and similar results are obtained for that case. The paper finishes with the appendix  containing calculations with the brackets, used in the proofs of our main results. 

 {\it Note added in proof:} Alternative treatments of the complex trigonometric spin RS model can be found in \cite{ArR,AO}.

{\bf Acknowledgement.} 
The authors thank L. Feh{\'e}r, I. Marshall, S. Ruijsenaars and P. Vanhaecke for useful discussions. 
 The work of the first author (O.~C.) was partially supported by EPSRC under grant EP/K004999/1.  Some of the results in this paper appear in the University of Leeds PhD thesis of the second author (M.~F.), supported by a University of Leeds 110 Anniversary Research Scholarship.

\subsection{Notations} \label{SNot}
The sets $\N,\Z,\CC$ contain the zero element and we write $\N^\times,\Z^\times,\CC^\times$ when we omit it. By an algebra we always mean an associative algebra over $\CC$. Vector spaces, matrices, varieties are also viewed over complex numbers.  
We write $\delta_{ij}$ or $\delta_{(i,j)}$ for Kronecker delta function. We extend this definition for a general proposition $P$ by setting $\delta_P=+1$ if $P$ is true and $\delta_P=0$ if $P$ is false.
Throughout the paper, the Greek letters placed as indices range through $1,\ldots,d$, for some fixed integer $d\geq 1$. The \emph{ordering function} on $d$ elements $\{1,\ldots,d\}$ is a skew-symmetric symbol defined by $o(\alpha,\beta)=0$ if $\alpha=\beta$, 
$o(\alpha,\beta)=+1$ if $\alpha<\beta$, and $o(\alpha,\beta)=-1$ if $\alpha>\beta$.


\section{Main results} \label{main results}

In the trigonometric case it is more convenient to work with $x_i=e^{2 q_i}$ and $q=e^{-2\gamma}$, rewriting \eqref{AFq}--\eqref{AFc} in the form 
\begin{subequations}
 \begin{align}
  \dot{x}_i & =  2{f}_{ii}\,x_i \,,\label{Trigq} \\ 
\dot{\aaa}^\alpha_i &=
\sum_{k\neq i} V_{ik}{f}_{ik} (\aaa_k^\alpha-\aaa_i^\alpha) \,,\label{Triga} \\
\dot{\ccc}_i^\alpha &=\sum_{k\neq i}
\left(\,V_{ik}{f}_{ik}\ccc_i^\alpha -V_{ki}{f}_{ki}\ccc_k^\alpha  \,\right)\,, \label{Trigc}
 \end{align}
\end{subequations}
with
\begin{equation} \label{EqPotV}
 V_{ik}
=\frac{x_i+x_k}{x_i-x_k}-\frac{x_i+qx_k}{x_i-qx_k}\,.
\end{equation}
As before, we view the system under the constraint $\sum_\alpha \aaa_i^\alpha=1$ for all $i=1,\ldots,n$, hence the phase space has dimension $2nd$. 
Equations \eqref{Trigq}--\eqref{Trigc} admit a Lax formulation \cite{KrZ,AF} with the Lax matrix 
\begin{equation}
 L=(L_{ij})_{i, j= 1, \dots, n}\,,\qquad L_{ij}=\frac{2x_i f_{ij}}{x_i-q x_j}  \,. \label{EqLax}
\end{equation}
Arutyunov and Frolov conjectured in \cite{AF} that the brackets between $L_{ij}$ should satisfy the $r$-matrix formulation of the non-spin trigonometric case \cite{AvR,Su,AFM}. Their conjecture can be formulated as follows. 

\begin{conj}[\cite{AF}]  \label{conjAF} 
The phase space with coordinates $(x_i,a_i^\alpha,b_i^\alpha)$ satisfying $\sum_\alpha a_i^\alpha=1$ admits a Poisson bracket such that 
 \begin{subequations}
       \begin{align}
\br{x_i,x_k} &= 0\,, \quad
\br{x_i, {f}_{jk}}=\delta_{ik}x_i {f}_{jk}\,, \label{Eq:qAF}\\
\br{{f}_{ij},{f}_{kl}} &=\frac12 
f_{ij}f_{kl} \left[\delta_{(i\neq k)}\frac{x_i+x_k}{x_i-x_k} + 
\delta_{(j\neq l)}\frac{x_j+x_l}{x_j-x_l}+\delta_{(k\neq j)}\frac{x_k+x_j}{x_k-x_j} + 
\delta_{(l\neq i)}\frac{x_l+x_i}{x_l-x_i}\right] \nonumber \\
&+\frac12 f_{il}f_{kj} \left[\delta_{(i\neq k)}\frac{x_i+x_k}{x_i-x_k} + 
\delta_{(j\neq l)}\frac{x_j+x_l}{x_j-x_l}+\frac{x_k+qx_j}{x_k-qx_j} - 
\frac{x_i+qx_l}{x_i-qx_l}\right]  \nonumber \\
&+\frac12 f_{ij}f_{il} \left[\delta_{(i\neq k)}\frac{x_k+x_i}{x_k-x_i} 
+\frac{x_i+qx_l}{x_i-qx_l} \right]+\frac12 f_{ij}f_{jl} 
\left[\delta_{(j\neq k)}\frac{x_j+x_k}{x_j-x_k} -\frac{x_j+qx_l}{x_j-qx_l} \right]  \nonumber \\
&+\frac12 f_{kj}f_{kl} \left[\delta_{(i\neq k)}\frac{x_k+x_i}{x_k-x_i} 
-\frac{x_k+qx_j}{x_k-qx_j} \right]+ \frac12 f_{lj}f_{kl} 
\left[\delta_{(i\neq l)}\frac{x_i+x_l}{x_i-x_l} +\frac{x_l+qx_j}{x_l-qx_j} \right] \,,\label{Eq:ffAF}
       \end{align}
  \end{subequations}
and such that the Hamiltonian vector field associated with the function $h=(1-q) \tr L$ coincides with \eqref{Trigq}--\eqref{Trigc}.
\end{conj}
\begin{rem}
 The above formulas are obtained in \cite{AF} from the assumption that the Poisson brackets in the spin and non-spin cases are governed by the same $r$-matrix. 
Namely, following \cite{AF} define  
\begin{equation*}
  \begin{aligned}
    r=&\sum_{ij} E_{ij} \otimes E_{ji} + \sum_{i\neq j} \frac{x_i+x_j}{x_i-x_j}  E_{ii} \otimes E_{jj} 
+ \sum_{i\neq j} \frac{x_i+x_j}{x_i-x_j} E_{ij} \otimes E_{ji} \nonumber \\
&\quad - \sum_{i \neq j} \frac{2 x_i}{x_i- x_j} E_{ij} \otimes E_{jj} + \sum_{i \neq j} \frac{2 x_i}{x_i- x_j} E_{jj} \otimes E_{ij}\,, \nonumber \\
\bar{r}=&-\sum_{i} E_{ii} \otimes E_{ii} + \sum_{i\neq j} \frac{x_i+x_j}{x_i-x_j}  E_{ii} \otimes E_{jj}  
- \sum_{i \neq j} \frac{2 x_i}{x_i- x_j} E_{ij} \otimes E_{jj}\,, \nonumber \\
\hat{r}=&-\sum_{ij} E_{ij} \otimes E_{ji}+ \sum_{i\neq j} \frac{x_i+x_j}{x_i-x_j}  (E_{ii} \otimes E_{jj} - E_{ij} \otimes E_{ji} )\,, \nonumber 
  \end{aligned}
\end{equation*}
 where $E_{ij}$ denote the elementary $n \times n$ matrices with $(E_{ij})_{kl}=\delta_{ik}\delta_{jl}$. 
Then, assuming \eqref{Eq:qAF}, the relations \eqref{Eq:ffAF} are equivalent to 
\begin{equation} \label{Zrmatrix}
  \br{L_1,L_2}=\,\frac12\left(r_{12}L_1 L_2 + L_1 L_2 \hat{r}_{12} + L_1 \bar{r}_{21} L_2 - L_2 \bar{r}_{12} L_1\right)\,,
\end{equation}
where $L$ is the Lax matrix \eqref{EqLax}.  See \cite[Section 3]{AF} for the details. 
\end{rem}

Let us now describe the space on which the quasi-Hamiltonian reduction will be performed, postponing a more detailed account to Section \ref{SqHpicture}. Consider the space $\mathcal M$ whose elements are the matrix data $X, Z, \{V_\alpha, W_\alpha\}_{\alpha=1,\ldots,d}$, where $X,Z\in \Mat_{n\times n}$, $V_\alpha\in \Mat_{1\times n}$,  $W_\alpha \in \Mat_{n\times 1}$. Clearly, $\mathcal M$ is an affine space of dimension $2n^2+2nd$. Let $\MMt\subset \mathcal M$ denote a subvariety defined by
\begin{equation} \label{mm}
 XZX^{-1}Z^{-1}(\Id_n+W_1 V_1)^{-1}\dots (\Id_n+W_d V_d)^{-1}=q\Id_n\,.
\end{equation}
Here $q$ is a nonzero parameter, and all the factors are assumed invertible. Throughout the paper it will be assumed that $q$ is not a root of unity. 
The group $\Gl_n$ acts on $\mathcal M$ and $\MMt$ by
\begin{equation}
g. (X,Z,V_\alpha,W_\alpha)=(gXg^{-1},gZg^{-1}, V_\alpha g^{-1}, gW_\alpha)\,,\quad g\in\Gl_n\,.
\end{equation}
For $q$ not a root of unity, the action on $\MMt$ is free and the GIT quotient $\MMt/\!/\Gl_n$ is a smooth affine variety of dimension $2nd$, whose coordinate ring is $\CC[\MMt]^{\Gl_n}$, i.e. the ring of $\Gl_n$-invariant functions on $\MMt$. The variety $\MMt/\!/\Gl_n$ is an example of  a multiplicative quiver variety. For general quivers, such varieties were introduced by Crawley--Boevey and Shaw in the context of multiplicative preprojective algebras \cite{CBShaw} (see also \cite{Y} and \cite[Appendix]{BEF}). Van den Bergh \cite{VdB1, VdB2} interpreted them as quasi-Hamiltonian quotients, so by his general result, $\MMt/\!/\Gl_n$ is a Poisson manifold with Poisson bracket induced from a quasi-Poisson bracket on $\mathcal M$. Van den Bergh's bracket on $\mathcal M$ is an anti-symmetric bi-derivation (bivector) defined in coordinates by     
\begin{subequations}
       \begin{align}
\br{X_{ij},X_{kl}} &= \frac12 \left( \delta_{il}(X^2)_{kj}  - \delta_{kj} (X^2)_{il} \right)\,, \quad 
\br{Z_{ij},Z_{kl}}=\frac12 \left(\delta_{kj} (Z^2)_{il} -\delta_{il}(Z^2)_{kj} \right)\,,  \label{qbr1}\\
\br{X_{ij},Z_{kl}} &= \frac12 \left((ZX)_{kj} \delta_{il} + \delta_{kj} (XZ)_{il} + Z_{kj} X_{il} - X_{kj} Z_{il} \right)\,, \\
\br{X_{ij}, W_{\alpha,k}} &= \frac12 \left(\delta_{kj} (X W_{\alpha})_i - X_{kj} W_{\alpha,i}\right)\,,\quad 
\br{X_{ij}, V_{\alpha,l}}= \frac12 \left( (V_\alpha X)_{j}\delta_{il}- V_{\alpha,j}X_{il}\right)\,, \\
\br{Z_{ij}, W_{\alpha,k}} &= \frac12 \left(\delta_{kj} (Z W_{\alpha})_i - Z_{kj} W_{\alpha,i}\right)\,,\quad 
\br{Z_{ij}, V_{\alpha,l}}= \frac12 \left( (V_\alpha Z)_{j}\delta_{il}- V_{\alpha,j}Z_{il}\right)\,,\\ 
\br{V_{\alpha,j},V_{\beta,l}} &= \frac12 \,o(\beta, \alpha) 
 \left(V_{\beta,j} V_{\alpha,l} + V_{\alpha,j} V_{\beta,l} \right)\,, \\
\br{W_{\alpha,i},W_{\beta,k}} &= \frac12 \,o(\beta, \alpha) 
\left(W_{\beta,k} W_{\alpha,i} + W_{\alpha,k} W_{\beta,i} \right)\,, \\
\br{V_{\alpha,j},W_{\beta,k}} &= \delta_{\alpha \beta}\left(  \delta_{kj}
+\frac12 W_{\alpha,k} V_{\alpha,j} + \frac12 \delta_{kj}  (V_\alpha W_\alpha) \right) \nonumber \\
\,& 
+ \frac12 \,o(\alpha,\beta) 
\left(\delta_{kj} (V_\alpha W_\beta) + W_{\beta,k} V_{\alpha,j}  \right)\,. \label{qbrLast}
	\end{align}
\end{subequations}
In these formulas, $o(\alpha, \beta)$ denotes the ordering function defined in \ref{SNot}. Note that the bracket does not satisfy the Jacobi identity, but the induced bracket on $\MMt/\!/\Gl_n$ does. This is because the space $\MMt$ with the $\Gl_n$-action fits into the framework of the quasi-Hamiltonian reduction \cite{AMM, AKSM}, with the left-hand side of \eqref{mm} playing the role of a multiplicative moment map. Thus, the Poisson variety $\MMt/\!/\Gl_n$ is an example of a quasi-Hamiltonian quotient. Our main result can then be stated as follows.

\begin{thm}\label{mainthm} The Poisson manifold $\MMt/\!/\Gl_n$ admits local coordinates $x_i, a_i^\alpha, b_i^\alpha$, with $\sum_\alpha a_i^\alpha=1$, so that the Hamiltonian vector field associated with the function $h=2(q^{-1}-1)\tr Z$ has the form \eqref{Trigq}--\eqref{Trigc}. Moreover, the Poisson bracket on $\MMt/\!/\Gl_n$ admits an explicit description in local coordinates, and the brackets between $x_i$ and $f_{ij}=\sum_\alpha \aaa_i^\alpha \ccc_j^\alpha$ agree with the formulas \eqref{Eq:qAF}--\eqref{Eq:ffAF}, thus confirming 
Conjecture \ref{conjAF}.  
This also implies the validity of the $r$-matrix formulation \eqref{Zrmatrix}.
\end{thm}

We can explain how the above local coordinates are constructed. They appear as a parametrisation of a local slice for the $\Gl_n$-action on $\MMt$. Namely, given $x_i, a_i^\alpha, b_i^\alpha$ with $\sum_\alpha a_i^\alpha=1$, let us introduce 
\begin{equation}\label{zlax}
 f_{ij}=\sum_\alpha \aaa_i^\alpha \ccc_j^\alpha\,,\quad X_{ij}=\delta_{ij}x_i\,, \quad Z_{ij}=\frac{qx_j  f_{ij} }{x_i-qx_j}\,.
\end{equation}
 We also set $(W_\alpha)_i=a_i^\alpha$. Finally, introduce $\Bs_\alpha\in\Mat_{1\times n}$ with $(\Bs_\alpha)_i:=b_i^\alpha$ and define $V_\beta\in\Mat_{1\times n}$ inductively by
 \begin{equation*}
V_1=\Bs_1Z^{-1}\,,\qquad V_\beta = \Bs_\beta Z^{-1} (\Id_n + W_1 V_1)^{-1} \ldots (\Id_n+W_{\beta-1}V_{\beta-1})^{-1}\,.
\end{equation*}
Then it is easy to check that the constructed matrix data $X, Z, V_\alpha, W_\alpha$ satisfy the moment map equation \eqref{mm}. (To be precise, we need to assume that $x_i, a_i^\alpha, b_i^\alpha$ are generic, so that $Z=(Z_{ij})$ is well-defined and invertible, and $\Id_n+W_\beta V_\beta$ is  invertible for any $\beta$.) This gives a subvariety in $\MMt$, with coordinates $x_i, a_i^\alpha, b_i^\alpha$, and it is easy to see that the $\Gl_n$-action on $\MMt$ is locally transversal to it. Thus, we obtain a local parametrisation of the space of orbits in $\MMt$, i.e. local coordinates on $\MMt/\!/\Gl_n$.      

The explicit description of the Poisson bracket in the local coordinates is given in the following proposition. 
\begin{prop}\label{brbr}
 The Poisson brackets between $(x_i,a_i^\alpha,b_i^\alpha)$ are given by 
 \begin{subequations}
 \begin{align}
  \br{x_i,x_j}=&0\,,\quad \br{x_i, \aaa_j^\alpha}=0\,,\quad 
\br{x_i, \ccc_j^\alpha}=\delta_{ij} x_i \ccc_j^\alpha \,, \label{Eqh1} \\
\br{\aaa_i^\alpha, \aaa_j^\beta}=&\frac12 \delta_{(i\neq j)}\frac{x_i+x_j}{x_i-x_j}
(\aaa_i^\alpha \aaa_j^\beta +\aaa_j^\alpha \aaa_i^\beta -\aaa_j^\alpha \aaa_j^\beta - 
\aaa_i^\alpha \aaa_i^\beta ) +\frac12 o(\beta, \alpha)
(\aaa_i^\alpha\aaa_j^\beta +\aaa_j^\alpha\aaa_i^\beta ) \nonumber \\
&+\frac12 \sum_{\gamma=1}^d o(\alpha,\gamma)\aaa_j^\beta 
(\aaa_i^\alpha\aaa_j^\gamma +\aaa_j^\alpha\aaa_i^\gamma )
-\frac12 \sum_{\gamma=1}^d o(\beta,\gamma)\aaa_i^\alpha 
(\aaa_j^\beta \aaa_i^\gamma+\aaa_i^\beta \aaa_j^\gamma)
\,, \label{Eqh2} \\
\br{\aaa_i^\alpha, \ccc_j^\beta}=&\aaa_i^\alpha Z_{ij}-\delta_{\alpha\beta}Z_{ij}-
\frac12 \delta_{(i\neq j)}\frac{x_i+x_j}{x_i-x_j} (\aaa_i^\alpha-\aaa_j^\alpha)\ccc_j^\beta
+\delta_{(\alpha<\beta)}\aaa_i^\alpha \ccc_j^\beta \nonumber \\
&+\aaa_i^\alpha \sum_{\gamma=1}^{\beta-1}\aaa_i^\gamma (\ccc_j^\gamma-\ccc_j^\beta) 
-\delta_{\alpha\beta} \sum_{\gamma=1}^{\beta-1} \aaa_i^\gamma \ccc_j^\gamma 
-\frac12 \sum_{\gamma=1}^d o(\alpha,\gamma)\ccc_j^\beta 
(\aaa_i^\alpha\aaa_j^\gamma +\aaa_j^\alpha\aaa_i^\gamma ) \,, \label{Eqh3} \\
\br{\ccc_i^\alpha, \ccc_j^\beta}=&
\frac12 \delta_{(i\neq j)}\frac{x_i+x_j}{x_i-x_j} (\ccc_i^\alpha\ccc_j^\beta  + \ccc_j^\alpha\ccc_i^\beta) 
-\ccc_i^\alpha Z_{ij} + \ccc_j^\beta Z_{ji} +\frac12 o(\beta,\alpha)
(\ccc_i^\alpha\ccc_j^\beta-\ccc_j^\alpha \ccc_i^\beta) \nonumber \\
&-\ccc_i^\alpha \sum_{\gamma=1}^{\beta-1}\aaa_i^\gamma (\ccc_j^\gamma-\ccc_j^\beta)
+\ccc_j^\beta \sum_{\gamma=1}^{\alpha-1}\aaa_j^\gamma (\ccc_i^\gamma-\ccc_i^\alpha)
 \,, \label{Eqh4}
 \end{align}
\end{subequations}
where $o(-,-)$ is the skew-symmetric pairing defined in \ref{SNot} and $Z_{ij}$ is defined in \eqref{zlax}.
\end{prop}
The formulas \eqref{Eqh1}--\eqref{Eqh4} are considerably more complicated than  \eqref{AF1}--\eqref{AF4}, which is probably why they have not been guessed earlier. The proof of this proposition relies on some fairly long computations performed in Appendix \ref{Ann:tadpole}. Note that the fact that the bracket defined by the formulas \eqref{Eqh1}--\eqref{Eqh4} is Poisson is not immediately obvious but follows from the reduction procedure.
The Arutyunov--Frolov's conjecture \eqref{Eq:qAF}--\eqref{Eq:ffAF} is then a direct corollary of Proposition \ref{brbr} and the definition of $f_{ij}$.

On the variety $\MMt/\!/\Gl_n$ we have $n$ algebraically independent functions $h_k=\tr Z^k$ ($k=1, \ldots, n$), which Poisson commute as a consequence of \eqref{qbr1}, see Lemma \ref{invyz}. 
The Hamiltonian flow for each of $h_k$ is complete and can be explicitly integrated, see \ref{ssFlowDIS} below (cf. \cite{RaS}). Thus, one may view $\MMt/\!/\Gl_n$ as a {completed phase space} for the trigonometric spin RS system. 
The following theorem is another main result of this paper. 
\begin{thm} \label{mainthm2}
 The Hamiltonian system defined on $\MMt/\!/\Gl_n$ by the Poisson commuting Hamiltonians $h_1,\ldots,h_n$ is degenerately integrable. Namely, there exists a Poisson subalgebra $\Qal \subset \CC[\MMt/\!/\Gl_n]$ of (Krull) dimension $2nd-n$, whose centre contains $h_1, \dots, h_n$.  
\end{thm}
This theorem is proved in Section \ref{STadIS}. The algebra $\Qal$ is described as follows. For any $k\in \N$, $\alpha,\beta=1,\ldots, d$, define functions $t_{\alpha \beta}^k=\tr(W_\alpha V_\beta Z^k)=V_\beta Z^k W_\alpha$. Note that $t_{\alpha \beta}^k\in \CC[\MMt]^{\Gl_n}=\CC[\MMt/\!/\Gl_n]$. Then $\Qal$ is the subalgebra of $\CC[\MMt]^{\Gl_n}$ generated by all $t_{\alpha \beta}^k$. 

We should mention that for the rational spin RS model,  degenerate integrability was established by Reshetikhin \cite{Re}. His approach applies to a wider family of spin models related to simple Lie algebras, and can be adapted to the trigonometric case as well \cite[Section 6]{Res2}. However, the reduction is performed on the Heisenberg double of $G$ (with $G=\mathrm{SL}_n$ for the type $A$ model), which is of different dimension compared to our space $\mathcal M$. Another important difference is that we consider complexified dynamics, for which the variety $\MMt/\!/\Gl_n$ provides a \emph{completed} phase space.

In Sections \ref{Slocal}, \ref{First} we also discuss another integrable system defined by the Hamiltonians $h_k=\tr Y^k$, for $Y=Z-X^{-1}$.  This system can be viewed as a modification of the trigonometric spin RS model, and it has similar properties: the Hamiltonian flows are complete and can be explicitly integrated, and we also have degenerate integrability. 

Degenerate integrability is known to be a stronger property than Liouville integrability. In the real smooth setting, it implies that the phase space can be fibred into invariant tori (or more general non-compact fibers) of smaller dimension, see \cite{N, J08}. Therefore, it is natural to expect that the above Hamiltonians $h_i$ can be extended to a full set of $nd$ commuting Hamiltonians on the variety $\MMt/\!/\Gl_n$. We establish this fact in Section \ref{Sd2}. Note that in general such an extension is not canonical. The completely integrable extension that we construct can be viewed as an analogue of the Gelfand--Tsetlin system.


\section{The quasi-Hamiltonian picture}   \label{SqHpicture} 

In this section we describe the quasi-Hamiltonian reduction procedure for obtaining a completed phase space for the spin trigonometric RS model. The reduction is performed on a representation space of a framed Jordan quiver, and is an application of the general theory developed by Van den Bergh \cite{VdB1, VdB2}. Following his approach, we first introduce a suitable noncommutative quasi-Hamiltonian algebra; the corresponding geometric objects will arise after passing to representation spaces for this algebra.

\subsection{The quasi-Hamiltonian algebra} \label{ss:qHtadpole} 

According to \cite{VdB1}, a quasi-Hamiltonian algebra is a triple consisting of an algebra $\Aal$, a double bracket $\dgal{-, -}$ on it, and a multiplicative moment map $\Phi\in\Aal$. These must satisfy certain properties which should be regarded as noncommutative analogues of the properties of quasi-Hamiltonian spaces \cite{AMM, AKSM}. In \cite[Section 6.7]{VdB1}, it is explained how to associate a quasi-Hamiltonian algebra to any quiver. We will not present Van den Bergh's theory in full detail (see \cite[Section 2]{CF} for a brief account sufficient for the purposes of this paper), and simply describe below a particular choice of $\Aal$, $\dgal{-,-}$ and $\Phi$ that we make.

\subsubsection{The algebra $\Aal$} \label{qui}
Consider a framed Jordan quiver $Q$ which has two vertices, $0$ and $\infty$, and arrows, $x: 0 \to 0$,  $v_1,\ldots,v_d:\infty \to 0$. By $\bar{Q}$ we denote the doubled quiver, which has additional arrows $y: 0 \to 0$ and  $w_1,\ldots,w_d:0 \to \infty$. 
Let $\CC \bar{Q}$ be the path algebra of the doubled quiver; it is generated by the idempotents $e_0, e_\infty$ (representing the zero paths) and the arrows $x,y,v_\alpha,w_\alpha$, with the multiplication given by concatenation of paths. We will be writing paths from left to right: e.g., $xw_\alpha$ represents a path that starts at $0$ and ends at $\infty$. The element $e_0+e_\infty$ will be identified with $1$. 
Introduce an algebra $\Aal$, obtained from $\CC \bar{Q}$ by formally inverting the elements 
$1+xy, 1+yx, 1+w_\alpha v_\alpha$ and $1+v_\alpha w_\alpha$. Below we will also use a further localisation of $\Aal$, obtained by inverting $x$; we denote the resulting algebra as $\Aalt$.

\subsubsection{The double bracket}  \label{ss:dbr}

By definition \cite{VdB1}, a double bracket on an algebra $A$ is a map $A\times A\to A\otimes A$, $(a,b)\mapsto \dgal{a, b}$ which is linear in both arguments and satisfies two properties,
\begin{equation}\label{2prop}
\dgal{a, b} =-\dgal{b, a}^\circ \quad\text{and}\quad  \dgal{a, bc}=\dgal{a, b}c+b\dgal{a, c}\,. 
\end{equation}
Here $\circ$ denotes a linear map $A\otimes A\to A\otimes A$ defined by $(u\otimes v)^\circ=v\otimes u$, so the first formula replaces the usual antisymmetry. The second formula means that the bracket is a derivation in the second argument, with $A\otimes A$ viewed as an $A$-bimodule in the usual way, i.e. with $a(u\otimes v)b=au\otimes vb$. In the quasi-Hamiltonian setting, this bracket is also required to be quasi-Poisson (we omit the definition, see \cite[Section 5]{VdB1} or \cite[Section 2.2]{CF}).   

By \cite[Section 6.7]{VdB1}, the path algebra $\CC\bar{Q}$ of any doubled quiver admits a quasi--Poisson double bracket. In our situation, this bracket takes the following form:     
 \begin{subequations}
       \begin{align}
\dgal{x,x}\,=\,&\frac{1}{2}\left( x^2\otimes e_{0}- e_{0}\otimes x^2 \right)\,,\quad
\dgal{y,y}\,=\,\frac{1}{2}\left(e_0\otimes y^2 - y^2\otimes e_0 \right)\,,\label{tadida}\\
\dgal{x,y}\,=\,&e_{0}\otimes e_{0}
+\frac{1}{2} (yx\otimes e_{0} + e_{0}\otimes xy+y\otimes x-x\otimes y)\,, \label{tadidb}\\
\dgal{x, w_\alpha}\,=\,& \frac12 e_{0}\otimes xw_\alpha-\frac12 x\otimes w_\alpha\,,\quad 
\dgal{x, v_\alpha}= \frac12 v_\alpha x\otimes e_0-\frac12 v_\alpha\otimes x\,,\label{tadidd}\\
\dgal{y, w_\alpha}\,=\,& \frac12 e_{0}\otimes yw_\alpha-\frac12 y \otimes w_\alpha\,,\quad 
\dgal{y, v_\alpha}= \frac12 v_\alpha y\otimes e_0-\frac12 v_\alpha\otimes y\,,\label{tadide} \\
\dgal{v_\alpha,v_\beta}\,=\,&\frac12 \,o(\beta, \alpha) 
\left(v_\alpha \otimes v_\beta+v_\beta\otimes v_\alpha \right)\,, \label{tadidv}\\
\dgal{w_\alpha,w_\beta}=\,&\frac12 \,o(\beta, \alpha) 
\left(w_\alpha \otimes w_\beta+w_\beta\otimes w_\alpha \right)\,, \label{tadidw}\\
\dgal{v_\alpha,w_\beta}=\,& \delta_{\alpha \beta}\left(  e_0\otimes e_\infty
+\frac12 w_\alpha v_\alpha \otimes e_\infty + \frac12 e_0 \otimes v_\alpha w_\alpha \right) \nonumber \\
\,& + \frac12 \,o(\alpha,\beta) 
\left(e_0\otimes v_\alpha w_\beta + w_\beta v_\alpha \otimes e_\infty \right)\,.\label{tadidu}
	\end{align}
\end{subequations}
The double bracket depends on a total ordering on the set of arrows of $\bar{Q}$, and our choice corresponds to setting $x<y<v_1<w_1<v_2<\ldots<v_d<w_d$. 
It is assumed here that the bracket is linear over the subalgebra $\CC e_0 \oplus \CC e_\infty$, that is, $\dgal{e_0, a}=\dgal{e_\infty, a}=0$ for all $a$. The above formulas are obtained by using \cite[Proposition 2.6]{CF}; they completely determine a double bracket on $\CC\bar{Q}$ due to \eqref{2prop}. 

It is clear that this bracket uniquely extends to the localised algebras $\Aal$ and $\Aalt$, defined above. Considering $\Aalt$, we can introduce $z=y+x^{-1}$, and obtain the double brackets 
\begin{subequations}
\begin{align}
&\dgal{z,z}\,=\,\frac{1}{2}\left(e_0\otimes z^2-z^2\otimes e_0 \right)\label{tadidaZ}\,,\quad
\dgal{x,z}\,=\, \frac{1}{2} (zx\otimes e_{0} + e_{0}\otimes xz + z\otimes x-x\otimes z)\,.\\
&\dgal{z, w_\alpha}\,=\, \frac12 (e_{0}\otimes zw_\alpha-z \otimes w_\alpha)\,,\quad 
\dgal{z, v_\alpha}= \frac12( v_\alpha z\otimes e_0-v_\alpha\otimes z)\,.\label{tadideZ} 
\end{align}
\end{subequations} 
This follows from a direct calculation, or by using \cite[Section 2.5]{CF}.

\subsubsection{The multiplicative moment map}\label{ss:mmap}
For our choice of a quiver, the multiplicative moment map is the following element $\Phi= \Phi_0+\Phi_\infty$, where 
\begin{subequations}
    \begin{align}
\Phi_0 &= e_0(1+xy)(1+yx)^{-1}\, (1+w_1 v_1)^{-1}\dots(1+w_d v_d)^{-1}e_0\,, \label{t1} \\ 
\Phi_\infty &= e_\infty(1+v_1 w_1)\dots (1+w_d v_d)e_\infty\,. \label{t2}
    \end{align}
\end{subequations} 
Note that the definition of $\Phi$ in \cite[6.7]{VdB1} requires a total ordering on $\bar{Q}$, which we take as above. Considering $\Aalt$, we can use the element $z=y+x^{-1}$ and write    
\begin{subequations}
    \begin{align}
\Phi_0&=e_0\,xzx^{-1}z^{-1}\, (1+w_1 v_1)^{-1}\dots(1+w_d v_d)^{-1}e_0\,, \label{tt1} \\ 
\Phi_\infty&=e_\infty(1+v_1 w_1)\dots (1+w_d v_d)e_\infty\,. \label{tt2}
    \end{align}
\end{subequations} 
The defining property \cite[5.1.4]{VdB1} of the moment map $\Phi$ is that it satisfies 
\begin{equation} \label{Phim}
 \dgal{\Phi_i,a}=\frac12 (ae_i\otimes \Phi_i-e_i \otimes \Phi_i a +  a \Phi_i \otimes e_i-\Phi_i \otimes e_i a)\,,
\end{equation}
for $i=0,\infty$ and any $a\in \Aal$.

It will also be convenient to introduce $\phi= xzx^{-1}z^{-1}$, which can be viewed as the moment map for the quasi-Hamiltonian algebra associated to the subquiver 
$\bar{Q}_0$ of $\bar{Q}$, obtained by deleting the vertex $\infty$ and all the arrows passing through it. By the properties of the moment map and of the fusion procedure \cite[5.3.1]{VdB1} we have 
 \begin{equation} \label{Phim2}
 \dgal{\phi,a}=\frac12 (ae_0\otimes \phi-e_0 \otimes \phi a +  a \phi \otimes e_0- \phi\otimes e_0 a)\,,
\end{equation}
for any $a\in \CC\langle x^{\pm1},z^{\pm1}\rangle$. The above formula can also be verified directly. Another direct calculation using \eqref{tadidd} and \eqref{tadideZ} shows that 
\begin{equation}
  \dgal{\phi,v_\alpha}=\frac12 (v_\alpha \phi \otimes e_0 - v_\alpha \otimes \phi), \quad 
\dgal{\phi,w_\alpha}=\frac12 (e_0 \otimes \phi w_\alpha - \phi \otimes w_\beta)\,. \label{Phim3}
\end{equation}

\subsubsection{Spin elements} 

For later use, let us introduce the following elements in $\Aalt$ 
\begin{equation}\label{salpha0}
  s_\alpha = (1+w_\alpha v_\alpha)\dots (1+w_{1}v_1)z\,, \quad 1\leq \alpha \leq d\,.
\end{equation}
We can see that $s_d=(\Phi_0)^{-1}\phi z$, and we can obtain all the other elements inductively by noticing that 
$s_\alpha=u_{\alpha+1} s_{\alpha+1}$, where $u_{\alpha}=(1+w_{\alpha}v_{\alpha})^{-1}$. In this way, we can obtain the double brackets between $s_\alpha$ and the generators of $\Aalt$. 
\begin{lem} \label{Lem:s}
 We have: 
\begin{subequations}
  \begin{align}
\dgal{s_\alpha,z}=&\frac12 (s_\alpha \otimes z - zs_\alpha \otimes e_0 + e_0 \otimes s_\alpha z - z \otimes s_\alpha) \label{br:sz}\\
\dgal{s_\alpha,x}=&\frac12 (s_\alpha\otimes x - xs_\alpha \otimes e_0 - e_0 \otimes s_\alpha x - x \otimes s_\alpha) \label{br:sx} \\
\dgal{s_\alpha,v_\beta}=& -\frac12 (v_\beta s_\alpha \otimes e_0 + v_\beta \otimes s_\alpha),\,\, 
\dgal{s_\alpha,w_\beta}=\frac12 (e_0 \otimes s_\alpha w_\beta+s_\alpha \otimes w_\beta),\,\, 
\text{for }\alpha \geq \beta, \label{br:svw} \\
\dgal{s_\alpha,v_\beta}=& \frac12 (v_\beta s_\alpha \otimes e_0 - v_\beta \otimes s_\alpha),\,\, 
\dgal{s_\alpha,w_\beta}=\frac12 (e_0 \otimes s_\alpha w_\beta - s_\alpha \otimes w_\beta ),\,\, 
\text{for }\alpha < \beta, \label{br:svw2}
  \end{align}
\end{subequations}
\end{lem}
See \ref{Ann:br1} for the proof. We can also obtain the double brackets between the elements $s_\alpha$ themselves, and this is proved in \ref{Ann:brBis}. 
\begin{lem} \label{Lem:sBis}
  We have: 
\begin{equation*}
  \dgal{s_\alpha,s_\beta}= \frac12 (e_0 \otimes s_\alpha s_\beta - s_\beta s_\alpha \otimes e_0) 
+ \frac12 o(\alpha,\beta) (s_\beta \otimes s_\alpha - s_\alpha \otimes s_\beta )\,.
\end{equation*}
\end{lem}
Below it will be convenient to pass from $v_\alpha, w_\alpha$ to the following \emph{spin variables}:  
\begin{equation} \label{spinAlg}
 a_\alpha=w_\alpha\,, \qquad b_\alpha=v_\alpha(1+w_{\alpha-1}v_{\alpha-1})\ldots (1+w_{1}v_1)z\,.
\end{equation} 
If we set $s_0=z$, we can write $b_\alpha=v_\alpha s_{\alpha-1}$ for any $1\leq \alpha \leq d$. Moreover, 
\begin{equation}\label{salpha}
  s_\alpha =z+ a_1b_1+\dots + a_\alpha b_\alpha\,.
\end{equation}

The double bracket can be written in the spin variables, and the 
only brackets not already among \eqref{tadida}--\eqref{tadidu} or \eqref{tadidaZ}--\eqref{tadideZ} are gathered in the following lemma.
\begin{lem} \label{Lem:Talg1}
We have:
\begin{subequations}
 \begin{align}
\dgal{x, b_\alpha}\,=\,& \frac12 b_\alpha x\otimes e_0+\frac12 b_\alpha\otimes x\,,
\quad \dgal{z, b_\alpha}= \frac12 b_\alpha\otimes z-\frac12 b_\alpha z\otimes e_0\label{tadSpin1}\\
\dgal{a_\alpha, b_\beta}\,=\,& 
\frac12 \left(o(\alpha,\beta)-\delta_{\alpha \beta} \right) e_\infty \otimes a_\alpha b_\beta -\frac12 b_\beta a_\alpha \otimes e_0 - \delta_{\alpha \beta} e_\infty \otimes s_{\beta-1}\,,\label{tadSpin2} \\
\dgal{b_\alpha,b_\beta}\,=\,& \frac12 o(\alpha,\beta) 
\left(b_\beta \otimes b_\alpha - b_\alpha \otimes b_\beta \right)  \,.\label{tadSpin3} 
 \end{align}
\end{subequations}
\end{lem}
Noting that Lemmas \ref{Lem:s} and \ref{Lem:sBis} also hold for the case $\alpha=0$ where $s_0=z$, the proof is easily derived using these double brackets and is omitted.

%

\subsection{Representation spaces} \label{ss:Repr}
Both $\Aal$ and $\Aalt$, equipped with the above $\dgal{-, -}$ and $\Phi$ provide examples of quasi-Hamiltonian algebras in the framework of \cite{VdB1}. A noncommutative analogue of quasi-Hamiltonian reduction consists in taking the quotient algebras  
\begin{equation}\label{mprep}
\Lambda^q=\Aal/\langle\Phi-q\rangle\,,\qquad \Lambda^{q, \times}=\Aalt/\langle\Phi-q\rangle
\end{equation}
for a chosen $q=q_0 e_0 + q_\infty e_\infty$, with $q_0, q_\infty\in \CC^\times$. The algebra $\Lambda^q$ is an example of a multiplicative preprojective algebra \cite{CBShaw}, while $\Lambda^{q, \times}$ is its localisation. The commutative counterpart is obtained by taking the representation spaces of $\Lambda^q$ or $\Lambda^{q, \times}$, respectively.

Recall that for an algebra $A$ and any $N\in\N$, a representation space $\Rep(A, N)$ is the affine scheme that parametrises algebra homomorphisms $\varrho: A \to \Mat_{N\times N}$. The ring of functions on $\Rep(A, N)$ is generated by the functions $a_{ij}$ for $a\in A$, $i, j=1, \ldots, N$ defined by $a_{ij}(\varrho)=\varrho(a)_{ij}$ at any point $\varrho\in \Rep(A, N)$. The functions $a_{ij}$ are linear in $a$ and satisfy the relations $(ab)_{ij}=\sum_k a_{ik}b_{kj}$. 
On $\Rep(A, N)$ we have a natural action of $\Gl_N$, induced by conjugation on $\Mat_{N\times N}$. 

To a double bracket $\dgal{-, -}$ on $A$, one associates a bracket (antisymmetric biderivation) on $\Rep(A, N)$ as follows \cite{VdB1}:
\begin{equation}\label{derr}
 \{a_{ij},b_{kl}\}=\dgal{a,b}'_{kj}\, \dgal{a,b}''_{il}\,\, . 
\end{equation}
Here we are using Sweedler notation, abbreviating an element $a=\sum_i a_i'\otimes a_i''$ in ${A}\otimes {A}$ to $a'\otimes a''$, so that $\dgal{a,b}=\dgal{a, b}'\otimes \dgal{a,b}''$. We have the following important result.

\begin{thm}
\label{Thm:Poiss}  \cite[7.8, 7.13.2]{VdB1}
Assume that $(A, \dgal{-,-}, \Phi)$ is a quasi-Hamiltonian algebra. Then $\Rep(A, N)$ is a $\Gl_N$-space with a quasi-Poisson bracket $\{-,-\}$ determined from $\dgal{-,-}$ by \eqref{derr}.
The $\Gl_N$-valued function $(\Phi_{ij})$ 
associated with $\Phi\in A$ provides a (geometric) multiplicative moment map. Therefore, $\Rep(A, N)$ (if smooth) is a Hamiltonian quasi-Poisson manifold in the sense of \cite{AMM, AKSM}. 
\end{thm} 
With suitable modifications, this result can be applied to quivers, see \cite[Proposition 1.7]{VdB1}. An additional feature of that case is that representations are sums of vector spaces attached to the vertices, and the arrows are represented by linear maps between corresponding spaces.    
For example, for the quiver from \ref{qui}, a representation of $\CC\bar{Q}$ consists of a vector space $\mathcal V=\mathcal V_0\oplus\mathcal V_\infty$ together with linear maps $X, Y: \mathcal V_0\to\mathcal V_0$, $V_\alpha: \mathcal V_0\to \mathcal V_\infty$,  $W_\alpha: \mathcal V_\infty\to \mathcal V_0$ (the zero paths $e_0, e_\infty$ are represented by the identity maps on the corresponding spaces). 
The dimension of a representation is a tuple $(\dim\VV_0, \dim\VV_\infty)$. For $\aalpha\in\N^2$, we write $\Rep(\CC\bar{Q}, \aalpha)$ for the space of representations of dimension $\aalpha$. Our main interest will be in the case when $\dim\VV_\infty=1$, so let us   
consider the spaces  $\Rep(\CC\bar{Q}, \aalpha)$ where $\aalpha=(n, 1)$ with $n\ge 1$. By choosing bases in $\VV_0, \VV_\infty$, we identify points of $\Rep(\CC\bar{Q}, \aalpha)$ with collections of matrices $(X,Y,V_\alpha,W_\alpha)$, 
\begin{equation}\label{pts}
  X,Y\in \Mat_{n\times n},\quad V_\alpha\in \Mat_{1\times n},\quad W_\alpha\in \Mat_{n\times 1}\,, \quad \alpha=1, \dots, d\,.
\end{equation}
Isomorphic representations are related by a change of basis, 
\begin{equation}\label{gact0}
g. (X,Y,V_\alpha,W_\alpha)=(g_0Xg_0^{-1},g_0Yg_0^{-1}, g_\infty V_\alpha g_0^{-1}, g_0W_\alpha g_\infty^{-1})\,,\quad g=(g_0, g_\infty)\in\Gl_n\times \Gl_1\,. 
\end{equation}
Thus, $\Rep(\CC\bar{Q}, \aalpha)$ is isomorphic to an affine space of dimension ${2n^2+2nd}$ with the above action of $\Gl_n\times \Gl_1$. 
The double bracket on $\CC\bar{Q}$ induces a quasi-Poison bracket on the representation spaces $\Rep(\CC\bar{Q}, \aalpha)$. It can be calculated by applying the formula \eqref{derr}. In doing so one should think of the linear maps $X, Y, V_\alpha, W_\alpha$ as being represented by block matrices acting on $\VV=\VV_0\oplus \VV_\infty$, and omit trivial brackets that involve zero matrix entries. For example, applying this to \eqref{tadida} gives 
\begin{equation*}
 \br{X_{ij}, X_{kl}} = \frac12 \left( (X^2)_{kj} \delta_{il} - \delta_{kj} (X^2)_{il} \right)\,, \quad 
\br{Y_{ij}, Y_{kl}}=\frac12 \left(\delta_{kj} (Y^2)_{il} -(Y^2)_{kj} \delta_{il}\right)\,.
\end{equation*}
This equips $\Rep(\CC\bar{Q}, \aalpha)$ with a quasi-Poisson bracket. Setting $Z=Y+X^{-1}$, we also derive formulas in \eqref{qbr1}--\eqref{qbrLast} in a similar way. 

Fixing the value of the moment map $\Phi$ to $q_0e_0+q_\infty e_\infty$, we get from \eqref{t1}--\eqref{t2} the following equations:
 \begin{subequations}
       \begin{align}
&(\Id_{n}+XY)(\Id_{n}+YX)^{-1}(\Id_n+W_1 V_1)^{-1}\ldots (\Id_n+W_d V_d)^{-1}=q_0\Id_n\,, \label{Eq:CondTadInv} \\
& (1+V_1 W_1)\ldots (1+V_d W_d)=q_\infty\,. \label{Eq:CondRq}
       \end{align}
  \end{subequations} 
By taking determinants, we obtain $(q_0)^nq_\infty=1$, cf. \cite[Lemma 1.5]{CBShaw}. To simplify the notation, set $q_0=q$, $q_\infty=q^{-n}$, then \eqref{Eq:CondRq} is automatically implied by \eqref{Eq:CondTadInv} and so can be omitted. Thus the level set of the moment map is described by the equation
\begin{equation}\label{lq}
(\Id_{n}+XY)(\Id_{n}+YX)^{-1}=q(\Id_n+W_d V_d)\ldots (\Id_n+W_1 V_1)\,,
\end{equation}
where all the factors are assumed to be invertible. We denote this variety as $\MM$. Clearly, $\MM$ is a representation space for the algebra $\Lambda^q$ \eqref{mprep}, seen as a subvariety of $\Rep(\Aal, \aalpha)$.   
Similarly, representation spaces for $\Lambda^{q, \times}$ are given by the equation  
\begin{equation}\label{lqt}
XZX^{-1}Z^{-1}=q(\Id_n+W_d V_d)\ldots (\Id_n+W_1 V_1)\,,
\end{equation}
where $Z=Y+X^{-1}$. This is precisely the variety $\MMt$ from Section \ref{main results}. 

The group $\Gl_n\times \Gl_1$ still acts on $\MM$, $\MMt$ by \eqref{gact0}. Note that the subgroup $\CC^\times$ of scalar matrices of the form $(\lambda\Id_n, \lambda)$ acts trivially, and we can identify the action of
$(\Gl_n\times \Gl_1)\,/\,\CC^\times$ with the $\Gl_n$-action given by  
\begin{equation}\label{gact}
g. (X,Y,V_\alpha,W_\alpha)=(gXg^{-1},gYg^{-1}, V_\alpha g^{-1}, gW_\alpha)\,,\quad g\in\Gl_n\,. 
\end{equation}
We can now introduce the following \emph{Calogero--Moser spaces}: 
\begin{equation*}
\Cnqd=\MM/\!/\Gl_n\,,\quad \Cnqdt=\MMt/\!/\Gl_n\,.
\end{equation*}
These are spin versions of the spaces from \cite{FockRosly, Oblomkov, CF}, which in their turn are $q$-analogues of the Calogero-Moser spaces $\mathcal C_n$ from \cite{W1}. 
For $q$ not a root of unity, the $\Gl_n$-action on $\MM, \MMt$ is free and $\Cnqd, \Cnqdt$ are smooth varieties of dimension $2nd$, cf. \cite[Theorem 2.8, Proposition 2.9]{CF}, 

By Theorem \ref{Thm:Poiss}, the varieties $\Rep(\Aal, \aalpha)$ and $\Rep(\Aalt, \aalpha)$ are quasi-Hamiltonian spaces, therefore, $\Cnqd$, $\Cnqdt$ can be seen as quasi-Hamiltonian quotients and so they are Poisson manifolds. By \cite[Sections 8.2-8.3]{VdB2}, the resulting Poisson bracket is non-degenerate, thus $\Cnqd$, $\Cnqdt$ are, in fact, holomorphic symplectic manifolds. In the next section we will explain their link to the trigonometric RS system.

\medskip

We finish this section by a few useful facts about the quasi-Poisson brackets on $\Rep(\CC\bar{Q}, \aalpha)$.  First, according to \cite{VdB1} with any double bracket on an algebra $A$ one associates the bracket $\{-, -\}\,:\, A\times A\to A$ obtained by composing $\dgal{-,-}$ with the multiplication $m: A\otimes A\to A$, that is, 
\begin{equation}\label{sbra}
\{a,b\}=m\circ \dgal{a, b}=\dgal{a,b}'\dgal{a,b}''\,.  
\end{equation}
By \cite[Proposition 5.1.2]{VdB1}, if the double bracket is quasi-Poisson then the bracket \eqref{sbra} induces a Lie bracket on $A/[A, A]$. Furthermore, if $A$ admits a quasi-Hamiltonian structure, then the bracket \eqref{sbra} induces a Lie bracket on $\Lambda^q/[\Lambda^q, \Lambda^q]$, see \cite[Proposition 5.1.5]{VdB1}.  

Next, for any $a\in A$, define $\tr a=\sum_{i=1}^N a_{ii}$; this is a $\Gl_N$-invariant function on $\Rep(A, N)$. Then
\begin{subequations}
       \begin{align}
\{\tr a, b_{kl}\} &=\br{a, b}_{kl}\,,\label{relBrDyn} 
\\
\{\tr a, \tr b\} &=\tr \{a, b\}\,.\label{relInv}
       \end{align}
  \end{subequations} 
Here on the left we use the bracket \eqref{derr} on $\Rep (A, N)$, while $\{a,b\}$ on the right stands for the bracket \eqref{sbra} on $A$. Both formulas are easy corollaries of \eqref{derr}, cf. \cite[Proposition 7.7.3]{VdB1}.  

\begin{lem}\label{invyz}
We have $\br{y^k, y^l}=\br{z^k, z^l}=0$ and $\br{\tr Y^k, \tr Y^l}=\br{\tr Z^k, \tr Z^l}=0$ for all $k, l$.
\end{lem}
To prove this, note that $\dgal{y,y}=(e_0\otimes e_0)y^2-y^2(e_0\otimes e_0)$ by \eqref{tadida}. We can then use \cite[Lemma A.3]{CF} with $\mathcal E=\{e_0\otimes e_0\}$ to conclude that $\br{y^k, y^l}=0$. By \eqref{relInv}, this implies $\br{\tr Y^k, \tr Y^l}=0$. For $z$ the proof is the same. 
\qed


\section{Local structure} \label{Slocal}

We continue with the notation of Section \ref{SqHpicture}. The local coordinates on the Calogero--Moser space $\Cnqdt$ have already been introduced in Section \ref{main results}. Below we recall their definition and then calculate the Poisson bracket in these coordinates.   

\subsection{Local coordinates}   \label{ss:TadCoord} 

Let us rewrite the definition of the variety $\MMt$ using the spin variables \eqref{spinAlg}. They are represented by $\As\in \Mat_{n\times 1}$ and $\Bs\in \Mat_{1\times n}$ given by
  \begin{equation} \label{AsCs}
(\As_{\alpha})_i=\left[W_\alpha\right]_i \, ,\quad
(\Bs_\alpha)_j=\left[V_\alpha(\Id_n+W_{\alpha-1}V_{\alpha-1})\ldots (\Id_n+W_1V_1)Z\right]_j\,.
  \end{equation}
By \eqref{salpha0} and \eqref{salpha}, we have 
\begin{equation}\label{salpharep}
(\Id_n+W_\alpha V_\alpha)
\dots (\Id_n+W_{1}V_1)Z=Z+ A_1B_1+\dots + A_\alpha B_\alpha\,. 
\end{equation}
Thus, the equation \eqref{lqt} becomes 
\begin{equation}\label{EqZspTad}
q^{-1}XZX^{-1}=Z+\sum_{\alpha=1}^d A_\alpha B_\alpha\,.
\end{equation}
Therefore, the variety $\MMt$ is formed by the tuples $(X,Z,\As_\alpha,\Bs_\alpha)$ satisfying \eqref{EqZspTad} together with the requirement of invertibility of $X, Z$ and of the expressions \eqref{salpharep}. That is, we assume in \eqref{EqZspTad} that  
\begin{equation}\label{invert}
X, Z\in\Gl_n\,,\quad Z+ A_1B_1+\dots + A_\alpha B_\alpha\in\Gl_n\qquad (\alpha=1, \dots, d)\,.
\end{equation} 
The Calogero--Moser space $\Cnqdt=\MMt/\!/\Gl_n$ is obtained by taking quotient by the action 
\begin{equation}\label{gactsp}
g. (X,Z,\As_\alpha, \Bs_\alpha)=(gXg^{-1},gZg^{-1}, g\As_\alpha, \Bs_\alpha g^{-1})\,,\quad g\in\Gl_n\,. 
\end{equation}

Note that if we set $A_\alpha=B_\alpha=0$ for $\alpha>1$ then we are effectively in the non-spin case $d=1$. 
Similarly, by truncating $A_\alpha, B_\alpha$ for $\alpha>2$ gives $\mathcal M_{n, 2, q}^\times$, and so on. 
Thus, we have a chain of inclusions $\mathcal M_{n, 1, q}^\times\subset \mathcal M_{n, 2, q}^\times\subset \ldots \subset \mathcal M_{n, d, q}^\times$ and, after taking $\Gl_n$-quotients, $\mathcal C_{n, 1, q}^\times\subset \mathcal C_{n, 2, q}^\times\subset \ldots \subset \mathcal C_{n, d, q}^\times$. 
The variety $\mathcal C_{n, 1, q}^\times$ is well-known, see e.g. \cite{FockRosly,Oblomkov,CF}; in particular, it is shown in \cite{Oblomkov} that it 
is connected. For $d>1$ it is not known whether $\Cnqdt$ is connected (it is believed to be true). Let $\mathcal C^\times\subset \Cnqdt$ denote the unique irreducible component containing $\mathcal C_{n, 1, q}^\times$. To see its link with the spin trigonometric RS system, we are going to introduce local coordinates on $\mathcal C^\times$.

Let $\h$ denote the phase space of the trigonometric RS system: this is an affine space of dimension $2nd$ with coordinates $(x_i,\aaa^\alpha_i,\ccc^\alpha_i)$ subject to $\sum_{\alpha}\aaa^\alpha_i=1$ for $i=1, \dots, n$. Define a mapping
\begin{equation}
\xi\,:\, (x_i,\aaa^\alpha_i,\ccc^\alpha_i)\mapsto (X, Z, \As_\alpha,\Bs_\alpha)\,, 
\end{equation}
which associates to a point in $\h$ the matrices $X, Z\in \Mat_{n\times n}$, $A_\alpha\in \Mat_{n\times 1}$, $B_\alpha\in \Mat_{1\times n}$ given by    
\begin{equation} \label{Tadiffeo}
  X_{ij}=\delta_{ij}x_i\,, \quad Z_{ij}=\frac{q\,f_{ij}}{x_ix_j^{-1}-q}\,, \quad  (\As_\alpha)_i=\aaa^\alpha_i\,,\quad  (\Bs_{\alpha })_i=\ccc_i^\alpha \,,
\end{equation}
with $f_{ij}=\sum_\alpha \aaa^\alpha_i \ccc^\alpha_j$. 
Now define $\hreg\subset \h$ to be the open subset given by the conditions  $x_i\ne 0$, 
$x_i\ne qx_j$, $x_i\ne x_j$ for $i\ne j$, together with the invertibility conditions \eqref{invert}. Note that on $\hreg$ the flow \eqref{Trigq}--\eqref{Trigc} and the antisymmetric bracket \eqref{Eqh1}--\eqref{Eqh4} are both well defined. A simple calculation confirms that $X, Z, A_\alpha, B_\alpha$ satisfy the equation \eqref{EqZspTad}, thus we have a map $\xi: \hreg\to \MMt$. To show that $\hreg$ is non-empty, we can set $a^1_i=1$, $b^1_i=\sigma_i\ne 0$ and $a^\alpha_i=b^\alpha_i=0$ for all $i$ and $\alpha>1$. The variables $x_i, \sigma_i$ can be viewed as local coordinates on the non-spin variety $\mathcal C_{n, 1, q}$ sitting inside $\Cnqd$. In these local coordinates the matrix $Z$ is given by $Z_{ij}=\frac{qx_j\,\sigma_{j}}{x_i-qx_j}$. It is equivalent to the Lax matrix in the non-spin case \cite{R88}, and its determinant can be easily evaluated using Cauchy formula from which it follows that $\det Z\ne 0$. The moment map equation \eqref{EqZspTad} reduces to $q^{-1}XZX^{-1}=Z+A_1 B_1$, so the invertibility of $Z+A_1B_1$ is automatic. It follows that the conditions \eqref{invert} are satisfied in this case, and so $\hreg\ne\emptyset$. 

Note that on the space $\hreg$ there is a natural $S_n$-action given by $\tau . (x_i,\aaa^\alpha_i,\ccc^\alpha_i)=(x_{\tau^{-1}(i)},\aaa^\alpha_{\tau^{-1}(i)},\ccc^\alpha_{\tau^{-1}(i)})$ for $\tau \in S_n$. Under $\xi$, this corresponds to the action \eqref{gactsp} by the corresponding permutation matrix. Therefore, we have a well-defined map $\xi: \hreg/S_n\to \mathcal C^\times$, and it is easy to see that it is injective. Since $\hreg$ and $\mathcal C^\times$ are of the same dimension, $\xi$ has dense image, 
so this gives local coordinates on 
$\mathcal C^\times$.

Now, both $\mathcal C^\times$ and $\hreg/S_n$ are equipped with a bracket: indeed, $\Cnqdt$ is a Poisson manifold, while a bracket on $\hreg/S_n$ is induced by the $S_n$-invariant bracket  
\eqref{Eqh1}--\eqref{Eqh4}.  

\begin{prop} \label{Pr:TadPoi} The map $\xi: \hreg/S_n\to \mathcal C^\times$ intertwines the brackets on these two spaces, that is, $\xi^\ast\br{f, g}=\br{\xi^* f , \xi^* g}$ for any two functions on $\mathcal C^\times$. Hence the bracket \eqref{Eqh1}--\eqref{Eqh4}  on $\hreg/S_n$ is Poisson, and $\xi$ is a Poisson map.
\end{prop}

\begin{proof}
Consider the functions
\begin{equation}\label{xifg0}
  f_k:= \tr(X^k)\,, \quad g^{k}_{\alpha \beta}:=\tr(\As_\alpha  \Bs_\beta X^k)=\Bs_\beta X^k \As_\alpha\,, \quad k\in \N,\,\alpha,\beta=1,\ldots,d\,.
\end{equation}
Using \eqref{Tadiffeo}, we obtain  
\begin{equation} \label{xifg}
  \xi^\ast f_k:= \sum_i x_i^k\,, \quad \xi^\ast g^{k}_{\alpha \beta}=\sum_{i}\aaa_i^\alpha \ccc_i^\beta x_i^k
\,, \quad \sum_\alpha \xi^\ast g^{k}_{\alpha \beta}=\sum_{i} \ccc_i^\beta x_i^k\,.
\end{equation}
A local coordinate system near every point in $\hreg$ can be extracted from these functions. Thus, the proposition only needs to be checked  for the functions \eqref{xifg0}. This is done in \ref{Ann:C1}. 
\end{proof}

\begin{prop}\label{comp}
Let 
$\mathcal M^\times\subset \MMt$ denote the unique irreducible component containing  
$\mathcal M_{n, 1, q}^\times$. Then the maps $p_x, p_z: \mathcal M^\times\to \Mat_{n\times n}$ defined by $p_x(X, Z, V_\alpha, W_\alpha)=X$ and $p_z(X, Z, V_\alpha, W_\alpha)=Z$ have dense image.
As a corollary, at a generic point of $\mathcal M^\times$ both $X$ and $Z$ have simple spectrum.   
\end{prop}
To prove this, we may restrict $p_x, p_z$ on to $\mathcal M_{n, 1, q}^\times$. In our discussion of the map $\xi$ above we have seen that $X$ can be chosen generic diagonal, and so the restriction of $p_x$ has dense image. 
Switching the roles of $X, Z$ in the construction of the local coordinates on $\mathcal M_{n, 1, q}^\times$, we conclude that $Z$ also takes generic values (here we use that this space is connected). Hence, both $p_x, p_z$ have dense image when restricted onto $\mathcal M_{n, 1, q}^\times\subset \mathcal M^\times$. \qed

\medskip

 In \ref{Ann:C2} the following result is proved.

\begin{prop}
 \label{ggFinal} 
The elements $x_i, f_{ij}$ generate a Poisson subalgebra of $\CC[\h]$ described by \eqref{Eq:qAF}-\eqref{Eq:ffAF}. We also have  
 \begin{subequations}
       \begin{align}
\br{x_i,\tr Z}=&\frac{q}{(1-q)} x_i f_{ii}\,,\quad \label{trz1}\\
\br{\aaa_i^\gamma,\tr Z}=&-\frac{q}{2(1-q)} \sum_{k\neq i} 
V_{ik}(\aaa_i^\gamma-\aaa_k^\gamma)f_{ik} \,, \label{trz2}\\ 
\br{\ccc_j^\epsilon,\tr Z}=&\frac{q}{2(1-q)} \sum_{k\neq j} 
(V_{jk}\ccc_j^\epsilon f_{jk}-V_{kj}\ccc_k^\epsilon f_{kj}) \,,\label{trz3}
       \end{align}
  \end{subequations}
where $V_{ik}$ is defined by \eqref{EqPotV}.
\end{prop}
Theorem \ref{mainthm} is an immediate consequence. 

\begin{rem}
 Our approach does not cover the case of the system \eqref{AFq}-\eqref{AFc} with the potential $V(z)=\coth(z)$. That particular system was considered in \cite{BH} in relation to affine Toda field theory, see also \cite{Li, Fe,Fe19} for the geometric treatment. 
\end{rem}

\subsection{Modified spin RS system}
The variety $\MMt$ can be seen as an open subvariety of the variety $\MM$ \eqref{lq}, obtained by imposing invertibility of $X$; the same is true for $\Cnqdt$ and $\Cnqd$. 
According to Lemma \ref{invyz}, on the space $\Cnqd$ we have commuting Hamiltonians $\tr Y^k$, $k=1, \dots, n$. In the local coordinates introduced above, the first Hamiltonian looks as follows: 
\begin{equation*}
 H=\sum_{i=1}^n \left(\frac{q}{1-q}f_{ii}-\frac{1}{x_i} \right)\,,
\end{equation*}
so it is a modification of the Hamiltonian for the spin RS system. Now, if we use $\tr Y=\tr Z-\tr X^{-1}$ instead of $\tr Z$ in \eqref{trz1}--\eqref{trz3} together with \eqref{Eqh1}, we get the following system:   
\begin{subequations}
  \begin{align}
    \dot{x}_i=&\frac{q}{1-q} x_i f_{ii}\,,\quad \label{MotHr}\\
\dot{\aaa}_i^\gamma=&-\frac{q}{2(1-q)} \sum_{k\neq i} 
V_{ik}(\aaa_i^\gamma-\aaa_k^\gamma)f_{ik} \,, \label{MotHa}\\ 
\dot{\ccc}_j^\epsilon=&\frac{q}{2(1-q)} \sum_{k\neq j} 
(V_{jk}\ccc_j^\epsilon f_{jk}-V_{kj}\ccc_k^\epsilon f_{kj})\,-\, \frac{\ccc_j^\epsilon}{x_j}   \label{MotHc}\,.
  \end{align}
\end{subequations}
The difference between these and \eqref{AFq}--\eqref{AFc} is due to the additional term in the third equation. In the non-spin case $d=1$, this Hamiltonian system first appeared in \cite{Iliev00} in relation to the $q$-KP hierarchy and bispectrality, see also \cite{CF}. It is therefore natural to expect that the system \eqref{MotHr}--\eqref{MotHc} describes solutions to the multicomponent $q$-KP hierarchy. We intend to return to this question elsewhere. Let us also mention that the quantum version in the $d=1$ case appeared in \cite{BF}, cf. \cite[Remark 3.25]{BEF}.


\section{Integrability of the system} \label{First}

\subsection{Degenerate integrability} \label{STadIS}

In this section we prove that the trigonometric RS system is degenerately integrable, see Theorem \ref{mainthm2}. We will use freely the notation from the previous sections.


Recall that a completely integrable system on a real symplectic manifold $M^{2n}$ consists of $n$ independent   
functions $H_1, \dots, H_n$ in involution, i.e. with $\br{H_i, H_j}=0$ for all $i, j$. In this situation, according to Liouville--Arnold theorem, generic compact joint level sets of $H_i$ are $n$-dimensional tori, and the Hamiltonian dynamics on each of these tori is quasi-periodic and can be integrated in quadratures. The notion of \emph{degenerate integrability} \cite{N} generalises this to a situation when there are $1\le k\le n$ independent functions $H_1, \dots, H_k$ in involution, together with a Poisson subalgebra $\mathcal Q\subset C^\infty(M^{2n})$ of dimension $2n-k$, such that each of $H_i$ Poisson commutes with all of $\mathcal Q$. In that case, generic compact joint level sets of the functions in $\mathcal Q$ are $k$-dimensional tori, and the dynamics for each $H_i$ is quasi-periodic on each of the tori and can be integrated by quadratures.  
The case $k=n$ corresponds to complete integrability, while the case $k=1$ is known as super-integrability. The same definition applies in the case when $M^{2n}$ is a holomorphic symplectic manifold. See \cite{J08} for further details and references. 

Let us consider the Calogero--Moser space $\Cnqdt=\MMt /\!/\Gl_n$; this is a smooth Poisson variety of dimension $2nd$. 
We identify functions on $\Cnqdt$ with $\Gl_n$-invariant functions on $\MMt$. The functions $h_i:=\tr Z^i$, $i=1, \dots n$ are independent and Poisson commute, by Proposition \ref{comp} and Lemma \ref{invyz}. Introduce $t_{\alpha\beta}^k:=\tr(W_\alpha V_\beta Z^k)=V_\beta Z^k W_\alpha$ and the subalgebra $\mathcal Q\subset \CC[\Cnqdt]$ generated by all $t_{\alpha\beta}^k$ with $k \in \N$, $\alpha,\beta=1,\ldots,d$. The following result is proved in \ref{Ann:A1}. 
\begin{lem} \label{Lem:ttMQV}  (1) We have $\br{h_i, t_{\beta \alpha}^k}=0$ for any $\alpha, \beta$ and $k\ge 0$. 

(2) For any $\alpha,\beta,\gamma,\epsilon$ and $k,l\geq 1$ we have 
\begin{equation} \label{Eqtt}
 \begin{aligned} 
    \br{t_{ \gamma\epsilon}^k,t_{ \alpha \beta}^l}=& \frac{1}{2}
\left[o(\gamma,\beta)+o(\epsilon,\alpha)-o(\epsilon,\beta)-o(\gamma,\alpha)\right] 
t_{ \gamma \epsilon}^k t_{ \alpha \beta}^l  \\
&+\frac{1}{2}o(\gamma,\beta)\,\,  t_{ \alpha \epsilon}^{k+l} t_{\gamma \beta}^0 
+\frac{1}{2}o(\epsilon,\alpha)\,\, t_{ \alpha \epsilon}^0 t_{\gamma \beta}^{k+l} 
-\frac{1}{2}o(\epsilon,\beta)\,\, t_{ \alpha \epsilon}^l t_{\gamma \beta}^{k}
-\frac{1}{2}o(\gamma,\alpha)\,\, t_{\alpha \epsilon}^k t_{\gamma \beta}^{l} \\
&-\delta_{\gamma\beta}\left[
t_{\alpha\epsilon}^{k+l}+\frac{1}{2}  t_{\alpha \epsilon}^{k+l} t_{\gamma \beta}^{0}
+\frac{1}{2}t_{\gamma\epsilon }^k t_{\alpha\beta}^{l} \right] 
+\delta_{\alpha\epsilon}\left[t_{\gamma\beta}^{k+l}
+\frac{1}{2}  t_{\alpha \epsilon}^0 t_{\gamma \beta}^{k+l}
+\frac{1}{2}t_{\gamma \epsilon}^k t_{\alpha\beta}^{l} \right] \\
&+\frac{1}{2}  \,\left[  \sum_{\tau=1}^k  t_{\gamma\beta}^{k-\tau} t_{\alpha\epsilon}^{l+\tau}-\sum_{\tau=1}^{k-1}  t_{\gamma\beta}^{k+l-\tau} t_{\alpha\epsilon}^{\tau}
\right] 
- \frac{1}{2}  \,\left[ \sum_{\sigma=1}^l t_{\gamma\beta}^{k+\sigma} t_{\alpha\epsilon}^{l-\sigma}  
-\sum_{\sigma=1}^{l-1} t_{\gamma\beta}^{\sigma} t_{\alpha\epsilon}^{k+l-\sigma}
\right]\,. 
\end{aligned}
\end{equation}
This formula remains valid when $k$ or $l$ (or both) are equal to zero, provided that we omit the final four sums. 
\end{lem}

\medskip

\begin{rem}\label{sc}
 Here are some special cases of the relations \eqref{Eqtt}. For $k=l=0$ we have 
 \begin{equation*}
 \begin{aligned} 
    \br{t_{\gamma\epsilon }^0,t_{\alpha\beta}^0}=& 
\delta_{\alpha\epsilon} t_{\gamma\beta}^0-\delta_{\gamma\beta}t_{\alpha \epsilon}^0
    +\frac{1}{2}\Big[\delta_{\alpha\epsilon}- \delta_{\gamma\beta}+o(\gamma,\beta)+o(\epsilon,\alpha)-o(\epsilon,\beta)-o(\gamma,\alpha)\Big] 
(t_{\gamma \epsilon}^0 t_{\alpha \beta}^0  + t_{\alpha \epsilon}^{0} t_{\gamma \beta}^0)  \,.
\end{aligned}
\end{equation*}
In particular, for $\alpha=\beta$ and $\gamma=\epsilon$ we obtain
$\br{t_{\gamma \gamma}^0,t_{\alpha \alpha}^0}=0$\,. More generally, for $k,l\ge 1$  
\begin{equation*}
 \begin{aligned}  \label{ttsame}
    \br{t^k_{\gamma\gamma},t^l_{\alpha\alpha}}=
&\frac{1}{2}o(\gamma,\alpha)\,\,\left[ t^0_{\gamma \alpha}\,t^{l+k}_{\alpha\gamma}+ t^{k+l}_{\gamma \alpha}\,t^0_{\alpha\gamma} 
- t^k_{\gamma \alpha}\,t^{l}_{\alpha\gamma} -t^l_{\gamma \alpha}\,t^k_{\alpha\gamma} \right] \\
& +\frac12 [t_{\gamma\alpha}^0t_{\alpha \gamma}^{k+l} - t_{\gamma\alpha}^{k+l}t_{\alpha \gamma}^0]
+\frac{1}{2}  \,\left[\sum_{\sigma=1}^{l-1}+\sum_{\sigma=1}^{k-1}\right] 
\left(  t_{\gamma \alpha}^{\sigma} t_{\alpha \gamma}^{k+l-\sigma} - t_{\gamma \alpha}^{k+l-\sigma} t_{\alpha \gamma}^{\sigma}\right)  \,. 
\end{aligned}
\end{equation*}
If $l=0$ and $k\ge 0$, this becomes
\begin{equation*}
 \label{ttsame1}
    \br{t^k_{\gamma\gamma},t^0_{\alpha\alpha}}=
\frac{1}{2}o(\gamma,\alpha)\,\,\left[ t^0_{\gamma \alpha}\,t^{k}_{\alpha\gamma}+ t^{k}_{\gamma \alpha}\,t^0_{\alpha\gamma} 
- t^k_{\gamma \alpha}\,t^{0}_{\alpha\gamma} -t^0_{\gamma \alpha}\,t^k_{\alpha\gamma} \right] =0\,.
\end{equation*}
Note that  we also have $\br{t^k_{\alpha\alpha},t^l_{\alpha\alpha}}=0$ for any $k,l\ge 0$.
\end{rem}

\begin{rem} \label{Lodayzac} It is possible to write an analogue of \eqref{Eqtt} for  the functions $s_{\beta \alpha}^k:=\Bs_\beta Z^k \As_\alpha$. It is still true that $\br{h_i, s_{\beta \alpha}^k}=0$, but the expressions for $\br{s_{\epsilon \gamma}^k, s_{\beta \alpha}^l}$ are more complicated than \eqref{Eqtt}. Note that $s_{\beta \alpha}^k$ are trigonometric analogues of the functions $J_k^{\alpha\beta}$ considered in \cite[(3.45)]{AF}.
\end{rem}

\medskip

Lemma \ref{Lem:ttMQV} shows that the algebra $\Qal$ is Poisson. The degenerate integrability of $h_i$ follows from the following result.

\begin{prop} \label{Pr:DISu} We have $h_i\in\Qal$ for all $i$. The algebra $\Qal$ is finitely generated and $\dim \Qal=2nd-n$.  
\end{prop} 
\begin{proof} The equation \eqref{lqt} can be written as  
\begin{equation}\label{lqt1}
q^{-1}XZX^{-1}=(\Id_n+W_d V_d)\ldots (\Id_n+W_1 V_1)Z\,.
\end{equation}
Raising it to power $k$ gives
\begin{equation}\label{lqt2}
q^{-k}XZ^kX^{-1}=Z^k+\dots \,,
\end{equation}
where the dots represent terms of the form $Z^aW_\alpha V_\alpha Z^bW_\beta V_\beta\dots W_\gamma V_\gamma Z^c$, with $a, b, \dots, c\ge 0$. The trace of every such term is easily expressed in terms of $t_{\alpha\beta}^i$:
\begin{equation*}
\tr (Z^aW_\alpha V_\alpha Z^bW_\beta \dots V_\gamma Z^c)=\tr (V_\alpha Z^bW_\beta \dots V_\gamma Z^{c+a}W_\alpha)=(V_\alpha Z^bW_\beta) \dots (V_\gamma Z^{a+c}W_\alpha)=t_{\beta\alpha}^b\dots t_{\alpha\gamma}^{a+c}\,.
\end{equation*}
Thus, by taking traces in \eqref{lqt2} we obtain $(q^{-k}-1)\tr Z^k\in \Qal$, i.e. $h_k\in\Qal$. 
It is now clear that $\Qal$ is generated by $t_{\alpha\beta}^i$ with $0\le i\le n$, since $t_{\alpha\beta}^k=V_\beta Z^k W_\alpha$ for $k>n$ can be expressed through those generators by the Cayley--Hamilton theorem and the fact that $\tr Z^j\in\Qal$ for all $j$.   
It remains to calculate the (Krull) dimension of $\Qal$. Note that it coincides with the maximal number of algebraically independent elements of $\Qal$; this also equals the dimension of the span of $df$, $f\in\Qal$ at a generic point of $\Cnqdt$.  
Since the Poisson bracket on $\Cnqdt$ is non-degenerate and the functions $h_i$ are independent, the equations $\br {h_i, f}=0$ for $f\in\Qal$ imply that $\dim\Qal\le 2nd-n$. Hence, it is sufficient to show the opposite inequality, $\dim\Qal\ge 2nd-n$.

To this end, consider the component $\mathcal M^\times\subset\MMt$ as in Proposition \ref{comp}. Recall that $\Gl_n$ acts freely on $\mathcal M^\times$, so $\dim\mathcal M^\times=n^2+\dim\mathcal C^\times=n^2+2nd$. Consider the following $\Gl_n$-equivariant map: 
\begin{equation}\label{mmap}
\pi\,:\ \mathcal M^\times\to \CC^{n^2+2nd}\,,\qquad (X, Z, V_\alpha, W_\alpha)\mapsto (Z, V_\alpha, W_\alpha)\,.
\end{equation}
We claim that generic fibers of $\pi$ have dimension $n$. Indeed, take a generic point $(X, Z, V_\alpha, W_\alpha)$ in $\mathcal M^\times$, then by Proposition \ref{comp} $Z$ has simple spectrum, so we may assume it is in diagonal form. Then \eqref{lqt1} tells us that $X$ puts $\widetilde Z:=q(\Id_n+W_d V_d)\ldots (\Id_n+W_1 V_1)Z$ into a diagonal form.  Therefore, for a given $(Z, V_\alpha, W_\alpha)$, $X$ is determined by choosing an eigenbasis for $\widetilde Z$. Hence, $\pi^{-1}(Z, V_\alpha, W_\alpha)$ is $n$-dimensional. As a result, $\pi(\mathcal M^\times)$ has dimension $\ge n^2+2nd-n$. 

If we view elements of $\Qal$ as functions of $Z, V_\alpha, W_\alpha$, then it is straightforward to check that $\dim\Qal\ge 2nd$ at any point in $\CC^{n^2+2nd}$ where $Z$ has simple spectrum and, say, $W_{1,i}\ne 0$. The dimension of $\Qal$ may drop after restriction onto $\pi(\mathcal M^\times)$. However, $\dim\pi(\mathcal M^\times)\ge n^2+2nd-n$, 
and so $\pi(\mathcal M^\times)\subset \CC^{n^2+2nd}$ is of codimension at most $n$. Thus, the dimension of $\Qal$ reduces by at most $n$, that is, $\dim\Qal\ge 2nd-n$ on $\pi(\mathcal M^\times)$. As a corollary, $\dim\Qal\ge 2nd-n$ when viewed on $\mathcal M^\times$. Since the functions in $\Qal$ are constant along $\Gl_n$-orbits in $\mathcal M^\times$, the dimension of $\Qal$ is the same whether viewed on $\mathcal M^\times$ or on $\mathcal C^\times=\mathcal M^\times/\!/\Gl_n$. We conclude that the span of $\mathrm{d}f$, $f\in\Qal$ has dimension $\ge 2nd-n$ generically on $\mathcal C^\times$, as needed.  
\end{proof}

\begin{rem}
Similar results are true for the modified spin RS system given by the Hamiltonians $h_k=\tr Y^k$, $k=1, \dots, n$ on the variety $\Cnqd$. Namely, we can set 
$t_{\alpha\beta}^k:=\tr(W_\alpha V_\beta Y^k)$ and consider the algebra $\Qal$ generated by all $t_{\alpha\beta}^k$. Then the formulas \eqref{Eqtt} remain true for that case as well.  We also have an analogue of Proposition \ref{Pr:DISu} proved in the same manner, and so the Hamiltonians $h_k=\tr Y^k$ define a degenerately integrable system on the space $\Cnqd$.
\end{rem}



\subsection{Algebra of first integrals and Liouville integrability} \label{Sd2}

By Proposition \ref{Pr:DISu}, the Hamiltonians $h_i=\tr Z^i$, $i=1, \dots, n$ define a degenerately integrable system. More precisely, this is true on a connected component of the space $\Cnqdt$ on which we have the local coordinates $\xi\,:\, \hreg\to \Cnqdt$. In the real smooth case, any degenerately integrable system can be extended (in a non-canonical way) to a completely integrable system, see \cite{BJ, J08}. Therefore, it is natural to expect that there exists a complete set of algebraic first integrals in our case, as well as in the case of the Hamiltonians $\tr Y^i$. Since $h_i$ Poisson commute with any $t_{\beta \alpha}^k=\tr (W_\beta V_\alpha Z^k)$, we may look for complementary Hamiltonians inside the algebra $\Qal$ generated by all $t_{\beta \alpha}^k$. The algebra $\Qal$ can be regarded as the algebra of joint first integrals for the Hamiltonians $h_1, \dots, h_n$.

Before discussing the general case, let us remark on some cases where the complete integrability is easy to establish. 
The case $d=1$ is trivial, since the functions $h_k$, $k=1, \dots, n$ are enough for integrability. 
Another case is $d=2$, where we can complement the functions $h_k$, $k=1, \dots, n$ by $t_{11}^k$ with $k=1\dots, n$. The latter functions Poisson commute with $h_k$ and between themselves by Remark \ref{sc}.  

Let us now introduce an infinite-dimensional version of the algebra of first integrals.  Namely, given $d\ge 1$ we define $\Qd$ to be the commutative algebra freely generated by the symbols $T_{\beta \alpha}^k$ with $\alpha, \beta\in\{1, \dots, d\}$ and $k\in\N$. These algebras form an increasing chain $\Qal_1\subset\Qal_2\subset \dots$. 

\begin{prop}  \label{Pr:Qd}
The formulas \eqref{Eqtt} define a Poisson bracket on $\Qd$. As a result, we have an increasing chain of Poisson algebras $\Qal_1\subset\Qal_2\subset \dots$.
\end{prop}

\begin{proof} Recall the representation spaces $\Rep(\CC\bar{Q}, \bar{\alpha})\cong \CC^{2n^2+2nd}$ whose points are represented by the matrix data \eqref{pts}. 
For any $n\ge 1$ we have a homomorphism $\varphi_n\,:\, \Qd\to \CC[\Rep(\CC\bar{Q}, \bar{\alpha})]^{\Gl_n}$ defined by $T_{\beta \alpha}^k\mapsto t_{\beta \alpha}^k$. 
We claim that 
\begin{equation}\label{krn}
\bigcap_n \ker\varphi_n=0\,.
\end{equation} 
To see that, it is enough to check that $\{t_{\beta\alpha}^i\,:\, \alpha, \beta=1,\dots, d,\ i=1,\dots k\}$ are functionally independent as elements of $\CC[\Rep(\CC\bar{Q}, \bar{\alpha})]$ if $n$ is sufficiently large. 
This can be checked on the subspace where $Z=\diag(z_1, \dots , z_n)$, by considering the matrix $(\partial t_{\beta\alpha}^i/\partial z_j)$ of size $kd^2\times n$ and showing that its rank is $kd^2$ for large $n$. 
This is a straightforward exercise left to the reader.

It is clear now that if we define an antisymmetric bracket on $\Qd$ by \eqref{Eqtt}, then the map $\varphi_n$ intertwines it with the quasi-Poisson bracket on $\CC[\Rep(\CC\bar{Q}, \bar{\alpha})]$. Since Jacobi identity holds on $\CC[\Rep(\CC\bar{Q}, \bar{\alpha})]^{\Gl_n}$, it must  then hold on $\Qd$ due to \eqref{krn}.  The fact that the resulting Poisson bracket is compatible with the inclusions $\Qal_1\subset\Qal_2\subset \dots$ is clear from \eqref{Eqtt}.  
\end{proof}

We can now construct an infinite family of central elements in $\Qd$. For this, let us consider 
\begin{equation} \label{Sspin1}
 S=(\Id_n+W_d V_d)\ldots (\Id_n+W_1V_1)Z\,.
\end{equation}
It is easy to see that for any $k$, $\tr S^k-\tr Z^k$ can be written as a polynomial in $t_{\beta \alpha}^i$ (see the proof of Proposition \ref{Pr:DISu}). Let us denote this polynomial as $h_{k,d}$. For example, we have   
\begin{equation*}
h_{1,1}=t_{11}^1\,,\ h_{2,1}=2t_{11}^2+(t_{11}^1)^2\,,\ h_{1,2}=t_{22}^1+t_{11}^1+t_{12}^0t_{21}^1\,.
\end{equation*}
Note that $h_{k,d}$ does not depend on $n$. It is shown in \ref{Ann:eta} that $\{\tr S^k, t_{\alpha\beta}^i\}=0$ on $\Rep(\CC\bar{Q}, \bar{\alpha})$ and therefore
\begin{equation}\label{central}
\{h_{k,d}, t_{\alpha\beta}^i\}=0\,.
\end{equation} 
Let us introduce $H_{k,d}:=h_{k,d}(T_{\beta \alpha}^i)$ by formally replacing $t_{\beta \alpha}^i$ with $T_{\beta \alpha}^i$. For example, $H_{1,1}=T_{11}^1$, $H_{2,1}=2T_{11}^2+(T_{11}^1)^2$.

\begin{prop} \label{Pr:Hd} The elements $H_{k,d}$ are central in $\Qd$. The subalgebra $\mathcal H_d\subset \Qd$ generated by all $H_{k,\alpha}$ with $1\le \alpha\le d$ and $k\in\N$ is Poisson commutative. 
\end{prop}
\begin{proof} 
According to \eqref{central}, for any fixed $n$ the functions $h_{k,d}$ Poisson commute with all $t_{\beta \alpha}^i$ on 
$\Rep(\CC\bar{Q}, \bar{\alpha})$. We then use \eqref{krn} to conclude that $H_{k,d}\in\mathcal Z(\Qd)$. 
The commutativity of $H_{k,\alpha}$ for all $k\in\N$ and $\alpha=1,\dots, d$ is obvious from the 
inclusions $\Qal_\beta\subset\Qal_\alpha$ for $\beta<\alpha$. (Alternatively, this also follows from \eqref{br:sz}, Lemma \ref{Lem:sBis} and \eqref{relInv}.) 
\end{proof}

The algebra generated by all $H_{k,\alpha}$ can be viewed as a subalgebra of Gelfand--Tsetlin type in $\Qd$. As a corollary, we obtain a completely integrable system on each $\Cnqdt$.

\begin{thm}
The functions $h_{k,\alpha}$ with $\alpha=1, \dots d$ and $k=1, \dots n$ define a completely integrable system on $\Cnqdt$, thus extending the degenerately integrable system defined by the Hamiltonians $h_k=\tr Z^k$.    
\end{thm} 

\begin{rem}
Strictly speaking, the above result is valid on the connected component $\mathcal C^\times\subset \Cnqdt$, see \ref{ss:TadCoord}. We will ignore this subtlety in the proof below.  
\end{rem}

\begin{proof} Let us introduce
\begin{equation} \label{Sspin2}
 S_\alpha=(\Id_n+W_\alpha V_\alpha)\ldots (\Id_n+W_1V_1)Z\,, \qquad \alpha=1,\dots, d\,. 
\end{equation}
By definition, we have $\tr S_\alpha^k=\tr Z^k+h_{k,\alpha}$, and from the moment map equation we have $\tr S_d^k=q^{-k}\tr Z^k$. Thus, it is enough to prove that the functions $\tr Z^k$ and $\tr S_{\alpha}^k$ with $k=1, \dots n$, $1\le \alpha< d$ are functionally independent. We will use the following lemma.  
\begin{lem} Near a generic point of $\Cnqdt$, the $2nd-n$ local functions $z_i$, $v_{\alpha, i}:=V_{\alpha,i}$,  $w_{\alpha, i}:=W_{\alpha,i}$ with $i=1, \dots, n$, $1\le \alpha<d$ are functionally independent. 
 \end{lem}
\begin{proof}[Proof (of the lemma).]
It is sufficient to show that for any pairwise distinct $z_i$ and generic $V_\alpha$, $W_\alpha$ with $\alpha<d$, one can find $V_d\in \CC^n$ and $X\in \Gl_n$ such that the moment map equation \eqref{lqt1} is satisfied with $Z=\diag(z_1, \dots, z_n)$ and with $W_{d,i}=1$. In its turn, it is enough to find $V_d$ so that the matrix
$S_d=(\Id_n+W_d V_d)\ldots (\Id_n+W_1 V_1)Z$
has the eigenvalues $q^{-1}z_1, \dots, q^{-1}z_n$. For fixed generic $V_\alpha, W_\alpha$, $\alpha=1, \dots, d-1$, we can view  $S_d$ as a rank-one perturbation of $S_{d-1}$. It is then an elementary fact that the eigenvalues of a regular semisimple matrix can be independently perturbed by a small rank-one perturbation, so we are done.  
 \end{proof}
As a consequence of the lemma, we can use the above $(z_i,v_{\alpha,i},w_{\alpha,i})$ as part of a local coordinate system of $\Cnqdt$. We have
 \begin{equation} \label{loc}
 \tr Z^k=\sum_i z_i^k, \quad t_{\alpha \beta}^k=\sum_i w_{\alpha,i}v_{\beta,i}z_i^k\,.
 \end{equation}
We therefore may simply view $\tr Z^k$ and $\tr S_{\alpha}^k$ with $\alpha<d$ and $1\le k\le n$ as polynomials of $(z_i,v_{\alpha,i},w_{\alpha,i})$, and we need to show that these polynomials are functionally independent. We will show that for any $\alpha=1, \dots, d-1$, the polynomials $\tr Z^k$ and $\tr S_{\beta}^k$ with $k=1, \dots, n$ and $\beta\le\alpha$ are independent. The proof is inductive. For $\alpha=1$, we want to prove that $\tr Z^k$ and $\tr S_{1}^k$, $k=1,\dots,n$, are functionally independent. We have $\tr S_1^k=\tr((\Id_n+W_1 V_1)Z)^k=\tr Z^k+k\, t_{11}^k+\ldots$, where the dots represent a polynomial in $t_{11}^l$ with $l<k$. Hence it is sufficient to show the functional independence of $\tr Z^k$ and $t_{11}^k$ with $k=1,\dots,n$. We can do this by looking at $2n\times 2n$ Jacobian matrix $J$ of derivatives of these functions with respect to $(z_1, \dots, z_n)$ and $(v_{1,1}, \dots, v_{1,n})$. This has a block structure 
$\begin{pmatrix} 
  J_0 & \ast \\ 
  0 & J_1 
\end{pmatrix}$  
where $J_0$ is the Vandermonde matrix for $z_1, \dots, z_n$, and $J_1$ has entries $\frac{\partial t_{11}^k}{\partial v_{1,i}}=w_{1,i}z_i^k$. Both $J_0, J_1$ are obviously nondegenerate for generic $z_i$ and $w_{1,i}$. This proves the $\alpha=1$ case.

The general case is similar: we form a Jacobian matrix of derivatives of $\tr Z^k$ and $\tr S_{\beta}^k$ with respect to the variables $z_i$ and $v_{\gamma,i}$. It similarly has an upper-triangular block structure, with the $n\times n$ blocks $J_0, \dots, J_{\alpha}$ along the diagonal. By induction, we only need to check that the last block $J_{\alpha}$ is non-degenerate. Its entries are $\frac{\partial \tr (S_{\alpha})^k}{\partial v_{\alpha,i}}$. To show that it is (generically) nondegenerate, we may choose $V_\beta=W_\beta=0$ for all $\beta<\alpha$, in which case $S_{\alpha}=(1+V_\alpha W_\alpha)Z$ and so this case can be analysed in the same way  
as for $\alpha=1$. This finishes the proof of the theorem.  
\end{proof}

\subsection{Explicit integration} \label{ssFlowDIS}
We begin by integrating the flows for the functions $h_k=\tr Z^k$ (and for their analogues, $\tr Y^k$). For $h_k$ the formulas are essentially the same as in \cite{RaS}.  
The main difference is that we work on a completed phase space, and that our flows are intrinsically Hamiltonian.

\begin{prop}\label{dynam}
Let $t$ denote the time flow associated to $\frac1k \tr Z^k$ for any $k \in \N^\times$. Given an initial position $(X, Z, V_\alpha, W_\alpha)$ in $\Cnqdt$, the solution at time $t$  is given by
\begin{equation}  \label{FlowRS}
X(t)=Xe^{-tZ^k}\,,\quad Z(t)=Z\,,\quad V_\alpha(t)=V_\alpha\,,\quad W_\alpha(t)=W_\alpha\,.
\end{equation}
Similarly, if $\tau$ denotes the time flow associated to $\frac1k \tr Y^k$, then the solution at time $\tau$ defined by an initial position $(X, Y, V_\alpha, W_\alpha)$ in $\Cnqd$ is given by
\begin{equation} \label{FlowH}
 X(\tau)=X e^{-\tau Y^k}+Y^{-1}(e^{-\tau Y^k}-1)\,,\quad Y(\tau)=Y\,,\quad V_\alpha(\tau)=V_\alpha\,,\quad W_\alpha(\tau)=W_\alpha\,.
 \end{equation}
These flows are complete when viewed on the corresponding Calogero--Moser spaces.  
\end{prop}
\begin{proof}
Let us write the flow corresponding to $\frac{1}{k}\tr Z^k$. Using \eqref{relBrDyn} together with the relations \eqref{tadidaZ}--\eqref{tadideZ}, one obtains the following equations:
\begin{equation*}
\dot{X}=-XZ^k\,,\quad \dot{Z}=0\,,\quad\dot{V}_\alpha=0\,,\quad \dot{W}_\alpha=0\,,
\end{equation*}
which imply \eqref{FlowRS}. 
For the flow corresponding to $\frac{1}{k}\tr Y^k$, the equations can be obtained by the same method leading to
 \begin{equation*}
\dot{X}=-XY^k-Y^{k-1}\,,\quad \dot{Y}=0\,,\quad\dot{V}_\alpha=0\,,\quad \dot{W}_\alpha=0\,,
\end{equation*}
which are integrated by \eqref{FlowH}. Note that the expression for $X(\tau)$ is well-defined even when the matrix $Y$ is singular.

The completeness of the flows is now clear, since the evolution described by \eqref{FlowRS} and \eqref{FlowH} preserves the invertibility of the factors appearing in the moment map equations \eqref{lqt} and \eqref{lq}, respectively.
\end{proof}

We can also integrate all the flows corresponding to the Hamiltonians $\tr S_\alpha^k$. Note that the matrices $S_\alpha$ represent the elements $s_\alpha$ \eqref{salpha0}. The function $\tr S_\alpha^k$ defines a vector field on the representation space of $\mathcal A^\times$ by the formula \eqref{relBrDyn}. This vector field is given explicitly as follows. 
\begin{prop} 
The vector field associated to the function $\tr S_{\alpha}^k$ is given by
\begin{equation*}
 \begin{aligned}
  & \dot{X}=-kXS_\alpha^k\,,\quad \dot{Z}=k(S_\alpha^k Z - ZS_\alpha^k)\,,\\ 
&\dot{V_\beta}= - kV_\beta S_\alpha^k\,,\quad  
\dot{W_\beta}=kS_\alpha^k W_\beta\,\quad\text{for }\beta \leq \alpha,  \\
&\dot{V_\beta}=0,\,\, 
\dot{W_\beta}=0\,\quad \text{for }\beta > \alpha.
 \end{aligned}
\end{equation*}
We also have $\dot S_\alpha=0$.
\end{prop}
\begin{proof}
The first group of relations is obtained by using \eqref{relBrDyn} together with the relations \eqref{br:sz}--\eqref{br:svw2}. The fact that $\dot S_\alpha=0$ follows from  \eqref{relBrDyn} and Lemma \ref{Lem:sBis}. 
\end{proof}

The following theorem is an immediate corollary.

\begin{thm}\label{dynamS}
Let $t$ denote the time flow associated to $\frac1k \tr S_\alpha^k$. Given an initial position $(X, Z, V_\alpha, W_\alpha)$ in $\Cnqdt$, the solution at time $t$  is given by
\begin{align*} 
X(t)&=Xe^{-tS_\alpha^k}\,,\quad Z(t)=e^{tS_\alpha^k}Ze^{-tS_\alpha^k}\,,\\
V_\beta(t)&=V_\beta e^{-tS_\alpha^k}\,,\quad W_\beta(t)=e^{tS_\alpha^k} W_\beta\quad\text{for}\ \beta\le\alpha\,,\\
V_\beta(t)&=V_\beta\,,\quad W_\beta(t)=W_\beta\quad\text{for}\ \beta>\alpha\,.
\end{align*} 
The flow is complete on $\Cnqdt$. 
\end{thm}

\begin{rem}
  A result similar to Theorem \ref{dynamS} can be obtained for $Y$ instead of $Z$ if we consider the analogue of Lemmas \ref{Lem:s} and \ref{Lem:sBis} in that case.  
\end{rem}

\begin{rem}
One can enlarge the Gelfand--Tsetlin subalgebra $\mathcal H_d$ by adding the elements $T_{\alpha\alpha}^0$, $\alpha=1, \dots, d$. We already know that $\{T_{\alpha\alpha}^0, T_{\beta\beta}^0\}=0$, see Remark \ref{sc}. To see that each of $T_{\beta\beta}^0$ Poisson commutes with $\mathcal H_d$, we check that $\{t_{\beta\beta}^0, \tr S_\alpha^k\}=0$ for all $\alpha$ and $k$ and then use \eqref{krn}. Since $t_{\beta\beta}^0=V_\beta W_\beta$, the fact that $\{t_{\beta\beta}^0, \tr S_\alpha^k\}=0$ is immediate from Theorem \ref{dynamS}.     
\end{rem}

\begin{rem}
The phase space of the (real) trigonometric RS system can be obtained from the moduli space of  flat $SU(n)$-connections on a torus with one puncture \cite{GN}, see also \cite{FockRosly}. Building on this relation, it is remarked in \cite{FGNR} that the self-duality of this system could be seen as a manifestation of a natural action of the mapping class group of the punctured torus.  
A proof of this statement (in the framework of {\it finite-dimensional} quasi-Hamiltonian reduction) can be found in \cite{FK12}. Similarly, the spin system can be linked to the moduli space of flat connections on a torus with several punctures. Indeed, by fixing the values of the first integrals $t_{\beta\beta}^0=V_\beta W_\beta$, one fixes the conjugacy classes of the matrices $\Id_n+W_\beta V_\beta$. The corresponding subvariety of $\Cnqdt$ can then be interpreted as a character variety of the torus with $d$ punctures. Therefore, the quasi-Hamiltonian reduction that leads to $\Cnqdt$ should be compatible with a natural action of the mapping class group of a torus with $d$ punctures. We will return to this question elsewhere.     
\end{rem}

\subsection{Lax matrix with spectral parameter} \label{ss:spec} 
Another approach to the integrability of the spin RS system uses a Lax matrix with spectral parameter \cite{KrZ}.  In our context, such a Lax matrix is given by 
\begin{equation*}
{Z}_\eta=Z+\eta S\,,\quad \text{where}\ S=(\Id_n+W_d V_d)\ldots (\Id_n+W_1 V_1)Z\,.
\end{equation*}
Here $\eta\in\CC$ is the spectral parameter. Note that $S=q^{-1}XZX^{-1}$ due to the moment map equation, thus $Z_\eta$ can be written entirely in terms of $X,Z$.
To see the connection with \cite{KrZ}, we use \eqref{salpharep} to rewrite $S$ as $S=Z+\sum_\alpha A_\alpha B_\alpha$. Then $Z_\eta$ takes the form
\begin{equation*}
Z_\eta=(1+\eta)Z+\eta\sum_{\alpha} A_\alpha B_\alpha\,.
\end{equation*}
If $Z$ has the form as in \eqref{Tadiffeo}, $Z_\eta$ can be easily identified with the trigonometric Lax matrix from \cite{KrZ}. 

The following result is proved in \ref{Ann:Spectral}. 
\begin{thm} \label{Thm:TadInv}
For any $\mu,\eta \in \CC$ and $k,l \in \N$, we have that $\br{\tr Z_\mu^k, \tr Z_\eta^l}=0$. 
\end{thm}
This implies that if we expand $\tr Z_\eta^k$ into a series in $\eta$, $\tr Z_\eta^k=\sum_{i=0}^k\eta^i r_{k,i}$, then $\{r_{k,i}, r_{l,j}\}=0$ for all $k,l,i,j$. In this way we recover the recipe for constructing first integrals from \cite{KrZ}. We remark that $r_{k,0}=\tr Z^k$, while each $r_{k,i}$ for $i>0$ can be rewritten as a product of $t_{\alpha\beta}^k$, i.e. they belong to the algebra $\mathcal Q$ of the first integrals considered above. 

Note that the integrals $r_{k,i}$ are not sufficient to construct a completely integrable system on $\Cnqdt$. Indeed, they all are functions of $X,Z$ and so do not distinguish points of $\Cnqdt$ that have the same $X,Z$ but different $V_\alpha, W_\alpha$. Assuming $d\le n$, it follows from the results of \cite{KrZ, Kr} that the maximum number of independent Poisson commuting Hamiltonians that can be obtained from the $r_{k,i}$ is $nd-d(d-1)/2$, which is strictly less than $nd$.  

We can use $r_{k,i}$ to construct a different commutative subalgebra in $\Qd$ compared to the subalgebra $\mathcal H_d$ of Gelfand--Tsetlin type constructed above. Namely, just replace all $t_{\alpha\beta}^k$ by $T_{\alpha\beta}^k$ in the expression for $r_{k,i}$. Denote the resulting subalgebra as $\mathcal R_d$. We do not know whether it can be enlarged to a bigger commutative subalgebra of $\Qd$ which would produce a completely integrable system on $\Cnqdt$. We only note that even if this is possible, the resulting integrable system will be different from the one constructed from the subalgebra $\mathcal H_d$. To see this, it suffices to check that there are elements in $\mathcal H_d$ and $\mathcal R_d$ that do not commute. One can check, for instance, that $\{t_{11}^1, r_{k,1}\}\ne0$ in general.


\appendix 


\section{ }\label{Ann:tadpole}

\subsection{Computations with double brackets} \label{Ann:Brack} 

We gather some results that we need when performing computations with double brackets in the other appendices. 

Firstly, we have noted that if $(\Aal,\dgal{-,-})$ is a double quasi-Poisson algebra,  its double quasi-Poisson bracket satisfies the cyclic antisymmetry rule  $\dgal{b,a}=-\dgal{a,b}^\circ$ and the derivation property $\dgal{a,bc}=b\dgal{a,c}+\dgal{a,b}c$, i.e. $\dgal{a,bc}=b\dgal{a,c}' \otimes \dgal{a,c}''+\dgal{a,b}' \otimes \dgal{a,b}'' c$ using Sweedler's notation. (This is true for the less restrictive assumption that $\Aal$ has a double bracket.)  There is a similar derivation property in the first argument for the \emph{inner} bimodule structure $\ast$, see \cite[(2.4)]{VdB1}, which gives $\dgal{bc,a}=\dgal{b,a}\ast c+b \ast \dgal{c,a}$ or more explicitly, $\dgal{bc,a}=\dgal{b,a}' c \otimes \dgal{b,a}'' +\dgal{c,a}' \otimes b \dgal{c,a}''$.   

Secondly, note that from the above properties we get that for any $a,b \in \Aal$ where $a$ has an inverse $a^{-1}$ 
\begin{equation} \label{BrInv}
  \dgal{b,a^{-1}}=-a^{-1}\dgal{b,a}a^{-1}\,,\quad 
\dgal{a^{-1},b}=- a^{-1} \ast \dgal{a,b} \ast a^{-1}. 
\end{equation}
 
Thirdly, if $a=a_1 \ldots a_k$ and $a'=a_1' \ldots a_l'$ are two elements of $\Aal$ written in terms of generators, then the double bracket between them is given by 
\begin{equation*}
\begin{aligned}
    \dgal{a, a'}=&\sum_{s=1}^k\sum_{t=1}^l (a_1 \ldots a_{s-1}) \ast (a_1' \ldots a_{t-1}') 
\dgal{a_s,a_t'} (a_{t+1} \ldots a_{l}) \ast (a_{s+1} \ldots a_{k}) \\
=&\sum_{s=1}^k\sum_{t=1}^l  (a_1' \ldots a_{t-1}') \dgal{a_s,a_t'}' (a_{s+1} \ldots a_{k}) \otimes (a_1 \ldots a_{s-1}) \dgal{a_s,a_t'}'' (a_{t+1} \ldots a_{l}) \,,
\end{aligned}
\end{equation*}
using the derivation properties as above.

\subsection{Proof of Lemma \ref{Lem:s}} \label{Ann:br1}

We show the claim by descending induction, starting from $\alpha=d$. So, the first step is to show that 
\begin{subequations}
  \begin{align}
\dgal{s_d,z}=&\frac12 (s_d \otimes z - zs_d \otimes e_0 + e_0 \otimes s_dz - z \otimes s_d) \label{br:szd}\\
\dgal{s_d,x}=&\frac12 (s_d\otimes x - xs_d \otimes e_0 - e_0 \otimes s_d x - x \otimes s_d) \label{br:sxd} \\
\dgal{s_d,v_\beta}=& -\frac12 (v_\beta s_d \otimes e_0 + v_\beta \otimes s_d) \quad 
\dgal{s_d,w_\beta}=\frac12 (s_d \otimes w_\beta + e_0 \otimes s_d w_\beta)\,. \label{br:svwd}
  \end{align}
\end{subequations}
To compute such double brackets, we use the relation $s_d=(\Phi_0)^{-1}\phi z$ and obtain  
\begin{equation}\label{terms}
  \dgal{s_d,a}=\dgal{\Phi_0^{-1},a}\ast \phi z+ \Phi_0^{-1}\ast \dgal{\phi,a}\ast z + \Phi_0^{-1}\phi \ast \dgal{z,a}\,.
\end{equation}
The first term can be calculated with the help of \eqref{Phim}:
\begin{equation*}
  \dgal{\Phi_0^{-1},a}=-\Phi_0^{-1}\ast  \dgal{\Phi_0,a} \ast  \Phi_0^{-1}=
-\frac12 (a \Phi_0^{-1} \otimes e_0- \Phi_0^{-1} \otimes e_0 a + a e_0 \otimes \Phi_0^{-1} - e_0 \otimes \Phi_0^{-1} a)\,.
\end{equation*}
(Note that in the case when $a=v_\beta$ or $a=w_\beta$, some of the terms in this expression vanish due to $e_0 v_\beta=w_\beta e_0=0$.) 
The second term in \eqref{terms} is calculated using \eqref{Phim2}--\eqref{Phim3}, while for the third term we use \eqref{tadidaZ}--\eqref{tadideZ}. Doing this for each of the cases $a=x, z, v_\beta, w_\beta$ verifies \eqref{br:szd}--\eqref{br:svwd}. We leave the details to the reader.

\medskip

For the induction step, recall that  $s_\alpha=u_{\alpha+1} s_{\alpha+1}$ for $u_{\alpha}=(1+w_{\alpha}v_{\alpha})^{-1}$. Therefore,
\begin{equation*}
  \dgal{s_\alpha,a}=\dgal{u_{\alpha+1},a} \ast s_{\alpha+1}+ u_{\alpha+1}\ast \dgal{s_{\alpha+1},a}\,.
\end{equation*}
The second term is given by the induction hypothesis, while we can find the first term using the easily verified formulas
\begin{equation*}
  \begin{aligned}
\dgal{u_{\alpha+1},z}=&
\frac12 (u_{\alpha+1} \otimes z - z u_{\alpha+1} \otimes e_0 - e_0 \otimes u_{\alpha+1} z + z \otimes u_{\alpha+1})\,,\\
\dgal{u_{\alpha+1},x}=&
\frac12 (u_{\alpha+1} \otimes x - x u_{\alpha+1} \otimes e_0 - e_0 \otimes u_{\alpha+1} x + x \otimes u_{\alpha+1})\,, \\
\dgal{u_{\alpha+1},v_\beta}=&\frac12 \delta_{(\alpha+1, \beta)}(v_\beta u_{\alpha+1} \otimes e_0 + v_\beta \otimes u_{\alpha+1}) + \frac12 o(\alpha+1,\beta) (v_\beta u_{\alpha+1} \otimes e_0 - v_\beta \otimes u_{\alpha+1})\,, \\
\dgal{u_{\alpha+1},w_\beta}=&-\frac12 \delta_{(\alpha+1, \beta)}(e_0 \otimes _{\alpha+1} w_\beta + u_{\alpha+1} \otimes w_\beta) + \frac12 o(\alpha+1,\beta) (e_0 \otimes u_{\alpha+1} w_\beta - u_{\alpha+1} \otimes w_\beta)\,. 
  \end{aligned}
\end{equation*}

In order to prove \eqref{br:sz}--\eqref{br:svw2}, we need to consider the cases $a=x, z, v_\beta, w_\beta$. 
We will do the case $a=w_\beta$, leaving the other cases to the reader. 

First, if $\beta \leq \alpha$, we can use \eqref{br:svw} and since $o(\alpha+1,\beta)=-1$ we get  
\begin{equation*}
\begin{aligned}
  \dgal{s_\alpha,w_\beta}=&-\frac12 (s_{\alpha+1} \otimes u_{\alpha+1} w_\beta - u_{\alpha+1}s_{\alpha+1} \otimes w_\beta) 
+\frac12 (s_{\alpha+1} \otimes u_{\alpha+1} w_\beta + e_0 \otimes u_{\alpha+1}s_{\alpha+1} w_\beta) \\
=&\frac12 (e_0 \otimes s_{\alpha} w_\beta +s_\alpha \otimes w_\beta)\,.
\end{aligned}
\end{equation*}
Next, if $\beta=\alpha+1$, we still use \eqref{br:svw} and find 
\begin{equation*}
\begin{aligned}
  \dgal{s_\alpha,w_{\alpha+1}}=&-\frac12 (s_{\alpha+1} \otimes u_{\alpha+1} w_\beta + u_{\alpha+1}s_{\alpha+1} \otimes w_\beta)
+\frac12 (s_{\alpha+1} \otimes u_{\alpha+1} w_{\alpha+1} + e_0 \otimes u_{\alpha+1}s_{\alpha+1} w_{\alpha+1}) \\
=&\frac12 (e_0 \otimes s_{\alpha} w_{\alpha+1} - s_\alpha \otimes w_{\alpha+1})\,.
\end{aligned}
\end{equation*}
Finally, if $\beta > \alpha+1$, we need \eqref{br:svw2} and since $o(\alpha+1,\beta)=+1$ we get  
\begin{equation*}
\begin{aligned}
  \dgal{s_\alpha,w_{\alpha+1}}=&\frac12  (s_{\alpha+1} \otimes u_{\alpha+1} w_\beta - u_{\alpha+1}s_{\alpha+1}  \otimes w_\beta)
+ \frac12 (e_0 \otimes u_{\alpha+1} s_{\alpha+1} w_\beta - s_{\alpha+1} \otimes u_{\alpha+1} w_\beta ) \\
=& \frac12  (e_0 \otimes s_{\alpha} w_\beta - s_{\alpha}  \otimes w_\beta)\,. &\quad \qed
\end{aligned}
\end{equation*}

\subsection{Proof of Lemma \ref{Lem:sBis}} \label{Ann:brBis}

Using the cyclic antisymmetry of the double bracket, we only need to show that for any $\alpha \geq \beta$,  
 \begin{equation} \label{br:ssAll}
   \dgal{s_\alpha,s_\beta}= \frac12 (e_0 \otimes s_\alpha s_\beta - s_\beta s_\alpha \otimes e_0 
+ s_\alpha \otimes s_\beta -  s_\beta \otimes s_\alpha) \,.
 \end{equation}
Using the elements $r_\gamma=1+w_\gamma v_\gamma$ and the definition of $s_\beta$ as \eqref{salpha0}, we can write $s_\beta=r_\beta \ldots r_1 z$. Therefore, we find that 
 \begin{equation} \label{ss:ab}
   \dgal{s_\alpha,s_\beta}=r_\beta\dots r_1 \dgal{s_\alpha,z} + \sum_{\gamma=1}^\beta r_\beta\dots r_{\gamma+1}
 \dgal{s_\alpha,r_\gamma}r_{\gamma-1} \dots r_1 z\,.
 \end{equation}
To find the double brackets $\dgal{s_\alpha,r_\gamma}$, we use \eqref{br:svw} and we obtain  
 \begin{equation*} 
   \dgal{s_\alpha,r_\gamma}=\frac12 (e_0 \otimes s_\alpha r_\gamma - r_\gamma \otimes s_\alpha 
s_\alpha \otimes r_\gamma  - r_\gamma s_\alpha \otimes e_0 )\,, \quad 1\leq \gamma \leq \alpha\,.
 \end{equation*}
For any $\beta \leq \alpha$, it remains to  substitute these double brackets  together with \eqref{br:sz} back in \eqref{ss:ab}, and we obtain \eqref{br:ssAll} after simplification. \qed

\subsection{Proof of Proposition \ref{Pr:TadPoi}} \label{Ann:C1} 

Recall that $f_k:= \tr(X^k)$ and $g^k_{\gamma \epsilon}=\tr(\As^\gamma \Bs^\epsilon  X^k)$. 
First, we need to compute the Poisson brackets between those functions. 
We have remarked in \ref{ss:Repr} that the Poisson bracket $\br{-,-}$ on $\Cnqd$ is (globally) defined from the corresponding Lie bracket $\br{-,-}$ on $\Lambda^q/[\Lambda^q,\Lambda^q]$ by \eqref{relInv}. In fact, it is sufficient to compute that bracket in $\Aal/[\Aal,\Aal]$, then projects into $\Lambda^q/[\Lambda^q,\Lambda^q]$. Assuming that $x$ is invertible, the same holds in $\Aalt$.  Therefore, we need the following lemma.  
\begin{lem} \label{Lodayxac}
 For any $k,l\geq 1$, the following identities hold in $\Aalt/[\Aalt,\Aalt]$
\begin{subequations}
 \begin{align}
  \br{x^k,x^l}=&0\,, \quad 
\br{x^k,a_\alpha b_\beta x^l}=k\, a_\alpha b_\beta x^{k+l}\,, \label{xxacx}\\
\br{a_\gamma b_\epsilon x^k,a_\alpha b_\beta x^l}=&
\frac12 \left(\sum_{r=1}^k-\sum_{r=1}^l \right) 
\left(a_\alpha b_\beta x^r a_\gamma b_\epsilon x^{k+l-r}
+a_\alpha b_\beta x^{k+l-r} a_\gamma b_\epsilon x^{r} \right) \nonumber \\
&+\frac12 o(\alpha,\gamma) (a_\gamma b_\epsilon x^k a_\alpha b_\beta x^l 
+a_\alpha b_\epsilon x^k a_\gamma b_\beta x^l) \nonumber \\
&+\frac12 o(\epsilon,\beta) (a_\alpha b_\beta x^k a_\gamma b_\epsilon x^l 
-a_\alpha b_\epsilon x^k a_\gamma b_\beta x^l) \nonumber \\
&+\frac12 [o(\epsilon,\alpha)+\delta_{\alpha \epsilon}]\,a_\alpha b_\epsilon x^k a_\gamma b_\beta x^l 
-\frac12 [o(\beta,\gamma)+\delta_{\beta \gamma}]\,a_\alpha b_\epsilon x^k a_\gamma b_\beta x^l \nonumber \\
&+\delta_{\alpha \epsilon} \left(zx^k+\sum_{\lambda=1}^{\epsilon-1} a_\lambda b_\lambda x^k \right)
a_\gamma b_\beta x^l - \delta_{\beta \gamma}\,\, a_\alpha b_\epsilon x^k
\left(zx^l + \sum_{\mu=1}^{\beta-1} a_\mu b_\mu x^l \right)\,. \label{acxacx}
 \end{align}
\end{subequations}
\end{lem}
The proof can be seen as a special case of \cite[Lemma 3.2]{F}. By taking the traces and using the identity \eqref{relInv}, 
we obtain the Poisson brackets between the functions $(f_k, g^{k}_{\alpha \beta})$. 
To write them in terms of $f^k, g^k_{\alpha\beta}$ and $h_{\gamma \epsilon}^{k,l}=\tr(\As_\gamma \Bs_\epsilon X^k Z X^l)$, we use
\begin{equation*}
  \tr (A_\alpha B_\beta X^{k} A_\gamma B_\epsilon X^{l})=
(\Bs_\beta X^{k} \As_\gamma) (\Bs_\epsilon X^{l}\As_\alpha)=g_{\gamma \beta}^k g_{\alpha \epsilon}^l\,,
\end{equation*}
and similar variants. 
\begin{lem} \label{LemPoisson}
 For any $\alpha,\beta=1,\ldots,d$ and $k,l\geq 1$,  
\begin{subequations}
 \begin{align}
\br{f_k,f_l}
=\,&0\,, \\
\br{f_k,g_{\alpha \beta}^l}
=\,&k\, g_{\alpha \beta}^{k+l}\,, \\
\br{g_{\gamma \epsilon}^k,g_{\alpha \beta}^l}
=\,&
\frac12 \left(\sum_{r=1}^k-\sum_{r=1}^l \right) 
\left(g_{\gamma \beta}^r g_{\alpha \epsilon}^{k+l-r} + g_{\gamma \beta}^{k+l-r} g_{\alpha \epsilon}^{r}\right) \nonumber \\
&+\frac12 o(\alpha,\gamma) \left(g_{\gamma \beta}^l g_{\alpha \epsilon}^{k} + g_{\gamma \epsilon}^{k} g_{\alpha \beta}^{l}\right) 
+\frac12 o(\epsilon,\beta) \left(g_{\gamma \beta}^k g_{\alpha \epsilon}^l - g_{\gamma \epsilon}^{k} g_{\alpha \beta}^{l}\right) \nonumber \\
&+\frac12 [o(\epsilon,\alpha)+\delta_{\alpha \epsilon}-o(\beta,\gamma)-\delta_{\beta \gamma}]\,
g_{\gamma \epsilon}^k g_{\alpha \beta}^{l} \nonumber \\
&+\delta_{\alpha \epsilon}  h_{\gamma \beta}^{l,k} 
+\delta_{\alpha \epsilon} \sum_{\lambda=1}^{\epsilon-1} g_{\gamma \lambda}^k g_{\lambda \beta}^l 
-\delta_{\beta \gamma} h_{\alpha \epsilon}^{k,l} - \delta_{\beta \gamma}\sum_{\mu=1}^{\beta-1} g_{\alpha \mu}^l g_{\mu \epsilon}^k\,.
 \end{align}
\end{subequations}
\end{lem}

Our goal is to show that for the functions 
$(f_k,g^k_{\gamma \epsilon})$ generating the ring of functions at a generic point, the following equalities hold
\begin{equation*}
 \xi^\ast\br{f_k,f_l}=\br{\xi^* f_k , \xi^* f_l}\,, \quad 
\xi^\ast\br{f_k,g^l_{\alpha \beta}}=\br{\xi^* f_k , \xi^* g^l_{\alpha \beta}}\,, \quad 
\xi^\ast\br{g^k_{\gamma \epsilon},g^l_{\alpha \beta}}=\br{\xi^* g^k_{\gamma \epsilon}, \xi^* g^l_{\alpha \beta}}\,,
\end{equation*}
where we compose with $\xi$ the identities from Lemma \ref{LemPoisson} in the left hand sides, and use the expressions \eqref{xifg} in the right hand sides. 
In fact, we will be quite pedantic and prove these identities also after summing over $\alpha$ and/or $\gamma$ ranging from $1$ to $d$. This allows us to show that the Poisson brackets given in Proposition \ref{Pr:TadPoi} are correct one at a time. Note that in local coordinates, we use $\xi^\ast X_{ij}=\delta_{ij} x_i$, $\xi^\ast (\As_\alpha  \Bs_\beta)_{ij}= \aaa_i^\alpha \ccc_j^\beta$ while we simply write $\xi^\ast Z_{ij}=Z_{ij}$.

To show that the brackets in \eqref{Eqh1} are correct, first notice that $\br{x_i,x_j}=0$ implies $\xi^\ast\br{f_k,f_l}=\br{\xi^* f_k , \xi^* f_l}$ as both expressions vanish. Second, recall that by assumption $\sum_\alpha \aaa_i^\alpha=1$ for all $i$. Thus, from $\br{x_i, \ccc_j^\beta}=\delta_{ij} x_i \ccc_j^\beta$, 
\begin{equation*}
 \begin{aligned}
 \sum_{\alpha} \br{\xi^* f_k , \xi^* g^l_{\alpha \beta}}=&\sum_{i,j=1}^n\br{x_i^k, \ccc_j^\beta x_j^l}
=\sum_{i=1}^n k\, x_i^{k+l} \ccc_i^\beta\,, \\
\sum_{\alpha} \xi^\ast\br{f_k,g^l_{\alpha \beta}}=& 
k\, \sum_{\alpha} \xi^*\tr(\As_\alpha \Bs_\beta X^{k+l})
=k\, \sum_{i=1}^n \ccc_i^\beta x_i^{k+l}\,,
 \end{aligned}
\end{equation*}
and we get  $\xi^\ast\br{f_k, \sum_{\alpha} g^l_{\alpha \beta}}=\br{\xi^* f_k , \xi^* \sum_{\alpha} g^l_{\alpha \beta}}$. 
Third, without summing, we get again that $\xi^\ast\br{f_k,g^l_{\alpha \beta}}=\br{\xi^* f_k , \xi^* g^l_{\alpha \beta}}$ using $\br{\aaa_i^\alpha,x_j}=0$.

To see that we need \eqref{Eqh4}, \eqref{Eqh3} and \eqref{Eqh2}, we will respectively sum over all values of $\alpha$ and $\gamma$, all values of $\gamma$ and finally not sum at all the functions 
$\xi^\ast\br{g^k_{\gamma \epsilon},g^l_{\alpha \beta}}$ and $\br{\xi^* g^k_{\gamma \epsilon}, \xi^* g^l_{\alpha \beta}}$, to show that they agree. We get from Lemma \ref{LemPoisson} that we can write 
\begin{equation}
 \begin{aligned} \label{xigg}
 \xi^\ast \br{g^k_{\gamma \epsilon},g^l_{\alpha \beta}}=\,&
\frac12 \left(\sum_{r=1}^k-\sum_{r=1}^l \right) \sum_{i,j=1}^n
\left(  \ccc_i^\beta x_i^r \aaa_i^\gamma \ccc_j^\epsilon x_j^{k+l-r} \aaa_j^\alpha
+\ccc_i^\beta x_i^{k+l-r} \aaa_i^\gamma \ccc_j^\epsilon x_j^r \aaa_j^\alpha
\right) \\
&+\frac12 o(\alpha,\gamma) \sum_{i,j=1}^n
\left( \ccc_j^\epsilon x_j^k \aaa_j^\alpha \ccc_i^\beta x_i^l \aaa_i^\gamma
+  \ccc_j^\epsilon x_j^k \aaa_j^\gamma \ccc_i^\beta x_i^l \aaa_i^\alpha \right) \\
&+\frac12 o(\epsilon,\beta)  \sum_{i,j=1}^n
\left(  \ccc_i^\beta x_i^k \aaa_i^\gamma \ccc_j^\epsilon x_j^l \aaa_j^\alpha
-  \ccc_j^\epsilon x_j^k \aaa_j^\gamma \ccc_i^\beta x_i^l \aaa_i^\alpha \right)  \\
&+\frac12 [o(\epsilon,\alpha)+\delta_{\alpha \epsilon}]\,
\sum_{i,j=1}^n  \ccc_j^\epsilon x_j^k \aaa_j^\gamma \ccc_i^\beta x_i^l\aaa_i^\alpha
-\frac12 [o(\beta,\gamma)+\delta_{\beta \gamma}]\,
\sum_{i,j=1}^n  \ccc_j^\epsilon x_j^k \aaa_j^\gamma \ccc_i^\beta x_i^l \aaa_i^\alpha \\
&+\delta_{\alpha \epsilon} \sum_{i,j=1}^n 
\left(Z_{ij}  + \sum_{\lambda=1}^{\epsilon-1} \aaa_i^\lambda \ccc_j^\lambda \right) x_j^k \aaa^\gamma_j \ccc_i^\beta x_i^l  -\delta_{\beta \gamma} \sum_{i,j=1}^n  
\left(Z_{ji}+\sum_{\mu=1}^{\beta-1} \aaa_j^\mu \ccc_i^\mu\right)x_i^l \aaa_i^\alpha \ccc_j^\epsilon x_j^k\,,
 \end{aligned}
\end{equation}. 

 In the first case, we have to prove 
\begin{equation} \label{Poi1}
\sum_{\gamma,\alpha=1}^d \xi^\ast\br{g^k_{\gamma \epsilon},g^l_{\alpha \beta}}=\sum_{i,j=1}^n\br{ \ccc_j^\epsilon x_j^k, \ccc_i^\beta x_i^l}\,.
\end{equation}
The right-hand side of \eqref{Poi1} can be read as 
\begin{equation*}
 \begin{aligned}
  \eqref{Poi1}_{RHS}=&
\sum_{i,j=1}^n\left( \br{ \ccc_j^\epsilon ,  x_i^l}  x_j^k \ccc_i^\beta
+\br{ x_j^k, \ccc_i^\beta }   \ccc_j^\epsilon x_i^l +\br{ \ccc_j^\epsilon , \ccc_i^\beta }  x_j^k x_i^l \right) \\
=&(k-l) \sum_{i=1}^n \ccc_i^\epsilon , \ccc_i^\beta  x_i^{k+l}
+\frac12 \sum_{\substack{i,j=1\\ i\neq j}}^n  x_j^k x_i^l \frac{x_j+x_i}{x_j-x_i} (\ccc_j^\epsilon \ccc_i^\beta + \ccc_i^\epsilon\ccc_j^\beta) \\
&+\sum_{i,j=1}^n  x_j^k x_i^l(\ccc_i^\beta Z_{ij} - \ccc_j^\epsilon Z_{ji}) +\frac12 o(\epsilon,\beta)
\sum_{i,j=1}^n  x_j^k x_i^l(\ccc_i^\epsilon\ccc_j^\beta-\ccc_j^\epsilon \ccc_i^\beta)  \\
&+\sum_{i,j=1}^n  x_j^k x_i^l \ccc_i^\beta \sum_{\lambda=1}^{\epsilon-1}\aaa_i^\lambda (\ccc_j^\lambda-\ccc_j^\epsilon)
-\sum_{i,j=1}^n  x_j^k x_i^l \ccc_j^\epsilon \sum_{\mu=1}^{\beta-1}\aaa_j^\mu (\ccc_i^\mu-\ccc_i^\beta).
 \end{aligned}
\end{equation*}
Now,  the left-hand side of \eqref{Poi1} can be written from \eqref{xigg}, after summing over $\alpha,\gamma$ and using the condition $\sum_\gamma \aaa_i^\gamma=1$ when possible. We get 
\begin{subequations}
 \begin{align}
  \eqref{Poi1}_{LHS}=&
\frac12  \sum_{i,j=1}^n  \ccc_i^\beta \ccc_j^\epsilon
\left(\sum_{r=1}^k-\sum_{r=1}^l \right) \left(  x_i^r   x_j^{k+l-r}+ x_i^{k+l-r} x_j^r\right) \label{Poi1a} \\
&+\frac12   \sum_{i,j=1}^n \ccc_j^\epsilon \ccc_i^\beta x_j^k x_i^l
\sum_{\alpha,\gamma=1}^d o(\alpha,\gamma) \left(\aaa_i^\gamma \aaa_j^\alpha  
+ \aaa_i^\alpha \aaa_j^\gamma  \right) 
+\frac12 o(\epsilon,\beta)  \sum_{i,j=1}^n x_j^k x_i^l
\left( \ccc_j^\beta  \ccc_i^\epsilon - \ccc_j^\epsilon  \ccc_i^\beta  \right) \label{Poi1b} \\
&+\frac12 \sum_{\alpha=1}^d[o(\epsilon,\alpha)+\delta_{\alpha \epsilon}]\,
\sum_{i,j=1}^n \aaa_i^\alpha \ccc_j^\epsilon x_j^k  \ccc_i^\beta x_i^l
-\frac12 \sum_{\gamma=1}^d [o(\beta,\gamma)+\delta_{\beta \gamma}]\,
\sum_{i,j=1}^n  \ccc_j^\epsilon x_j^k \aaa_j^\gamma \ccc_i^\beta x_i^l  \label{Poi1c}\\
&+\sum_{i,j=1}^n x_j^k x_i^l (Z_{ij} \ccc_i^\beta - Z_{ji} \ccc_j^\epsilon)
+\sum_{i,j=1}^n \sum_{\lambda=1}^{\epsilon-1} \aaa_i^\lambda \ccc_j^\lambda  x_j^k  \ccc_i^\beta x_i^l  - \sum_{i,j=1}^n  \sum_{\mu=1}^{\beta-1} \aaa_j^\mu \ccc_i^\mu x_i^l  \ccc_j^\epsilon x_j^k \,, \label{Poi1d}
 \end{align}
\end{subequations} 
To reduce this expression further, remark that by definition of the ordering function $o(-,-)$
\begin{equation*}
 \sum_{\alpha,\gamma=1}^d o(\alpha,\gamma) \left(\aaa_i^\gamma \aaa_j^\alpha  
+ \aaa_i^\alpha \aaa_j^\gamma  \right) = 
\sum_{\alpha<\gamma} \left(\aaa_i^\gamma \aaa_j^\alpha  + \aaa_i^\alpha \aaa_j^\gamma  \right) 
- \sum_{\alpha>\gamma}  \left(\aaa_i^\gamma \aaa_j^\alpha  
+ \aaa_i^\alpha \aaa_j^\gamma  \right)=0\,,
\end{equation*}
after relabelling the indices in the second sum, so that the first term of \eqref{Poi1b} disappears. Then, write \eqref{Poi1a} as 
\begin{equation*}
 \begin{aligned}
\eqref{Poi1a}=& (k-l) \sum_{i=1}^n  \ccc_i^\beta \ccc_i^\epsilon x_i^{k+l}
+ \frac12  \sum_{\substack{i,j=1\\ i\neq j}}^n  \ccc_i^\beta \ccc_j^\epsilon
  \left(\sum_{r=1}^k-\sum_{r=1}^l \right) \left(  x_i^r   x_j^{k+l-r}+ x_i^{k+l-r} x_j^r\right)\,,
 \end{aligned}
\end{equation*}
so that the sum for $i\neq j$ can be written as (here we assume $k>l$, the case $k<l$ is exactly the same) 
\begin{equation}
 \begin{aligned} \label{Sumrkl}
&\, \frac12  \sum_{\substack{i,j=1\\ i\neq j}}^n  \ccc_i^\beta \ccc_j^\epsilon
\sum_{r=l+1}^k \frac{x_i-x_j}{x_i-x_j}\left(  x_i^r   x_j^{k+l-r}+ x_i^{k+l-r} x_j^r\right) \\
=&\frac12  \sum_{\substack{i,j=1\\ i\neq j}}^n  \ccc_i^\beta \ccc_j^\epsilon\frac{x_i+x_j}{x_i-x_j}
\left(x_i^{k} x_j^{l}-x_i^{l} x_j^{k}\right)\,
=\,-\frac12  \sum_{\substack{i,j=1\\ i\neq j}}^n x_i^{l} x_j^{k} \frac{x_i+x_j}{x_i-x_j}
\left(\ccc_j^\beta \ccc_i^\epsilon+\ccc_i^\beta \ccc_j^\epsilon \right) \,,
 \end{aligned}
\end{equation}
after relabelling indices to obtain last equality. 
Finally, let's look at the terms in \eqref{Poi1c}-\eqref{Poi1d} with no factor $Z_{ij}$.  They can be written as 
\begin{equation*}
\frac12 \sum_{i,j=1}^n x_j^k x_i^l 
\left( 
\left[\sum_{\alpha\geq \epsilon}-\sum_{\alpha=1}^{\epsilon-1}\right]\aaa_i^\alpha \ccc_j^\epsilon \ccc_i^\beta -\left[\sum_{\gamma\geq \beta}-\sum_{\gamma=1}^{\beta-1}\right] \ccc_j^\epsilon  \aaa_j^\gamma \ccc_i^\beta + 2\sum_{\lambda=1}^{\epsilon-1} \aaa_i^\lambda \ccc_j^\lambda    \ccc_i^\beta   
- 2 \sum_{\mu=1}^{\beta-1} \aaa_j^\mu \ccc_i^\mu   \ccc_j^\epsilon
\right) \,,
\end{equation*} 
and if we split the sum  $\sum_{\alpha\geq \epsilon}$ as $\sum_{\alpha=1}^d-\sum_{\alpha=1}^{\epsilon-1}$ and do the same with the sum over $\gamma \geq \beta$, we get after using the conditions $\sum_\alpha \aaa_i^\alpha=1$ (and the same for $\gamma$)
\begin{equation*}
\begin{aligned}
 &\sum_{i,j=1}^n x_j^k x_i^l 
\left( 
-\sum_{\alpha=1}^{\epsilon-1}\aaa_i^\alpha \ccc_j^\epsilon \ccc_i^\beta +\sum_{\gamma=1}^{\beta-1} \ccc_j^\epsilon  \aaa_j^\gamma \ccc_i^\beta + \sum_{\lambda=1}^{\epsilon-1} \aaa_i^\lambda \ccc_j^\lambda    \ccc_i^\beta - \sum_{\mu=1}^{\beta-1} \aaa_j^\mu \ccc_i^\mu   \ccc_j^\epsilon
\right)\\
=&\sum_{i,j=1}^n x_j^k x_i^l 
\left( \sum_{\lambda=1}^{\epsilon-1} \aaa_i^\lambda (\ccc_j^\lambda-\ccc_j^\epsilon)    \ccc_i^\beta - \sum_{\mu=1}^{\beta-1} \aaa_j^\mu (\ccc_i^\mu - \ccc_i^\beta)   \ccc_j^\epsilon \right)\,.
\end{aligned}
\end{equation*} 
Summing together all the terms, we have reduced the left-hand side of \eqref{Poi1} to the form 
\begin{equation*}
 \begin{aligned}
  \eqref{Poi1}_{LHS}=&
(k-l) \sum_{i=1}^n  \ccc_i^\beta \ccc_i^\epsilon x_i^{k+l}
-\frac12  \sum_{\substack{i,j=1\\ i\neq j}}^n x_i^{l} x_j^{k} \frac{x_i+x_j}{x_i-x_j}
 \left(\ccc_j^\beta \ccc_i^\epsilon+\ccc_i^\beta \ccc_j^\epsilon \right) +\frac12 o(\epsilon,\beta)  \sum_{i,j=1}^n x_j^k x_i^l
\left( \ccc_j^\beta  \ccc_i^\epsilon - \ccc_j^\epsilon  \ccc_i^\beta  \right)  \\
&+\sum_{i,j=1}^n x_j^k x_i^l (Z_{ij} \ccc_i^\beta - Z_{ji} \ccc_j^\epsilon)
+\sum_{i,j=1}^n x_j^k x_i^l 
\left( \sum_{\lambda=1}^{\epsilon-1} \aaa_i^\lambda (\ccc_j^\lambda-\ccc_j^\epsilon)    \ccc_i^\beta - \sum_{\mu=1}^{\beta-1} \aaa_j^\mu (\ccc_i^\mu - \ccc_i^\beta)   \ccc_j^\epsilon \right)\,.
 \end{aligned}
\end{equation*} 
This is precisely the right-hand side of $\eqref{Poi1}$. 
In the second case, we show
\begin{equation} \label{Poi2}
\sum_{\gamma=1}^d \xi^\ast\br{g^k_{\gamma \epsilon},g^l_{\alpha \beta}}=\sum_{i,j=1}^n\br{ \ccc_j^\epsilon x_j^k, \aaa_i^\alpha \ccc_i^\beta x_i^l}\,.
\end{equation}

The right-hand side of \eqref{Poi2} can be read as 
\begin{equation*}
 \begin{aligned}
  \eqref{Poi2}_{RHS}=&
\sum_{i,j=1}^n\left( \br{ \ccc_j^\epsilon ,  x_i^l}  x_j^k  \aaa_i^\alpha \ccc_i^\beta
+\br{ x_j^k, \ccc_i^\beta }   \ccc_j^\epsilon \aaa_i^\alpha x_i^l +\br{ \ccc_j^\epsilon , \ccc_i^\beta }  x_j^k x_i^l \aaa_i^\alpha +  \br{ \ccc_j^\epsilon ,\aaa_i^\alpha}   x_j^k  x_i^l \ccc_i^\beta \right) \\
 &=(k-l) \sum_{i=1}^n \ccc_i^\epsilon \aaa_i^\alpha \ccc_i^\beta  x_i^{k+l}
 +\frac12 \sum_{\substack{i,j=1\\ i\neq j}}^n  x_j^k x_i^l \frac{x_j+x_i}{x_j-x_i} (\ccc_i^\epsilon\ccc_j^\beta \aaa_i^\alpha + \ccc_j^\epsilon\ccc_i^\beta \aaa_j^\alpha) \\
 &+\frac12 \sum_{\kappa=1}^d o(\alpha,\kappa)  \sum_{i,j=1}^n x_j^k  x_i^l \ccc_i^\beta \ccc_j^\epsilon 
(\aaa_j^\kappa \aaa_i^\alpha+\aaa_i^\kappa \aaa_j^\alpha)
+\frac12 o(\epsilon,\beta) \sum_{i,j=1}^n  x_j^k x_i^l(\ccc_i^\epsilon\ccc_j^\beta-\ccc_j^\epsilon \ccc_i^\beta) \aaa_i^\alpha \\
 & -\sum_{i,j=1}^n  x_j^k x_i^l \aaa_i^\alpha \ccc_j^\epsilon \sum_{\mu=1}^{\beta-1}\aaa_j^\mu (\ccc_i^\mu-\ccc_i^\beta)-\delta_{(\alpha<\epsilon)} \sum_{i,j=1}^n x_j^k  x_i^l \ccc_i^\beta \aaa_i^\alpha \ccc_j^\epsilon \\
& +\delta_{\epsilon \alpha}\, \sum_{i,j=1}^n  x_j^k x_i^l \ccc_i^\beta \left(Z_{ij} + \sum_{\lambda=1}^{\epsilon-1} \aaa_i^\lambda \ccc_j^\lambda  \right) 
-\sum_{i,j=1}^n  x_j^k x_i^l \ccc_j^\epsilon Z_{ji} \aaa_i^\alpha \,,
 \end{aligned}
\end{equation*} 
after some easy simplifications. To get the left-hand side, we  sum \eqref{xigg} over $\gamma$ and we write 
\begin{subequations}
 \begin{align}
  \eqref{Poi2}_{LHS}=& (k-l)\sum_{i=1}^n \aaa_i^\alpha \ccc_i^\beta \ccc_j^\epsilon x_i^{k+l} - 
\frac12  \sum_{\substack{i,j=1\\ i\neq j}}^n x_j^k x_i^l \frac{x_i+x_j}{x_i-x_j}\left(\aaa_i^\alpha \ccc_j^\beta \ccc_i^\epsilon + \aaa_j^\alpha \ccc_i^\beta \ccc_j^\epsilon \right) \label{Poi2a} \\
&+\frac12 \sum_{\gamma=1}^d o(\alpha,\gamma) \sum_{i,j=1}^n x_j^k x_i^l \ccc_j^\epsilon \ccc_i^\beta
\left(\aaa_i^\gamma   \aaa_j^\alpha  + \aaa_i^\alpha  \aaa_j^\gamma \right) 
+\frac12 o(\epsilon,\beta)  \sum_{i,j=1}^n x_j^k x_i^l \aaa_i^\alpha 
\left( \ccc_j^\beta  \ccc_i^\epsilon -  \ccc_j^\epsilon  \ccc_i^\beta \right) \label{Poi2b} \\
&+\frac12 [o(\epsilon,\alpha)+\delta_{\alpha \epsilon}]\,
\sum_{i,j=1}^n x_j^k x_i^l  \ccc_j^\epsilon \aaa_i^\alpha \ccc_i^\beta 
-\frac12 \sum_{\gamma=1}^d [o(\beta,\gamma)+\delta_{\beta \gamma}]\,
\sum_{i,j=1}^n x_j^k x_i^l  \aaa_i^\alpha \ccc_j^\epsilon  \aaa_j^\gamma \ccc_i^\beta \label{Poi2c}\\
&+\delta_{\alpha \epsilon} \sum_{i,j=1}^n x_j^k x_i^l \ccc_i^\beta
\left(Z_{ij}  + \sum_{\lambda=1}^{\epsilon-1} \aaa_i^\lambda \ccc_j^\lambda \right)     
- \sum_{i,j=1}^n  x_j^k x_i^l \left(Z_{ji}+\sum_{\mu=1}^{\beta-1} \aaa_j^\mu \ccc_i^\mu\right) \aaa_i^\alpha \ccc_j^\epsilon \,, \label{Poi2d}
 \end{align}
\end{subequations} 
where we used an argument similar to \eqref{Sumrkl} to rewrite the first line of \eqref{xigg} in order to obtain \eqref{Poi2a}. Next,  we can write after rearranging terms
\begin{equation*}
 \begin{aligned}
  &\eqref{Poi2c}+\eqref{Poi2d} \\
=& +\frac12 [1-2\delta_{(\alpha<\epsilon)}]\,
\sum_{i,j=1}^n x_j^k x_i^l  \ccc_j^\epsilon \aaa_i^\alpha \ccc_i^\beta 
-\frac12 \sum_{\gamma=1}^d [1-2\delta_{(\beta>\gamma)}]\,
\sum_{i,j=1}^n x_j^k x_i^l  \aaa_i^\alpha \ccc_j^\epsilon  \aaa_j^\gamma \ccc_i^\beta \\
&+\delta_{\alpha \epsilon} \sum_{i,j=1}^n x_j^k x_i^l \ccc_i^\beta
\left(Z_{ij}  + \sum_{\lambda=1}^{\epsilon-1} \aaa_i^\lambda \ccc_j^\lambda \right)     
- \sum_{i,j=1}^n  x_j^k x_i^l Z_{ji} \aaa_i^\alpha \ccc_j^\epsilon 
- \sum_{i,j=1}^n  x_j^k x_i^l \aaa_i^\alpha \ccc_j^\epsilon \sum_{\mu=1}^{\beta-1} \aaa_j^\mu \ccc_i^\mu \\
=& -\delta_{(\alpha<\epsilon)}\,
\sum_{i,j=1}^n x_j^k x_i^l  \ccc_j^\epsilon \aaa_i^\alpha \ccc_i^\beta 
+  \sum_{i,j=1}^n x_j^k x_i^l  \aaa_i^\alpha \ccc_j^\epsilon \sum_{\mu=1}^{\beta-1} \aaa_j^\mu (\ccc_i^\beta-\ccc_i^\mu) \\
&+\delta_{\alpha \epsilon} \sum_{i,j=1}^n x_j^k x_i^l \ccc_i^\beta
\left(Z_{ij}  + \sum_{\lambda=1}^{\epsilon-1} \aaa_i^\lambda \ccc_j^\lambda \right)     
- \sum_{i,j=1}^n  x_j^k x_i^l Z_{ji} \aaa_i^\alpha \ccc_j^\epsilon \,,
 \end{aligned}
\end{equation*}
where we used again the condition $\sum_{\gamma=1}^d \aaa_j^\gamma=1$. It is not hard to see that replacing the terms in  \eqref{Poi2c} and \eqref{Poi2d} by this last expression gives that $\eqref{Poi2}_{LHS}$ and 
$\eqref{Poi2}_{RHS}$ coincide.  
In the third case, we need to prove that 
\begin{equation} \label{Poi3}
 \xi^\ast\br{g^k_{\gamma \epsilon},g^l_{\alpha \beta}}=\sum_{i,j=1}^n\br{ \aaa_j^\gamma \ccc_j^\epsilon x_j^k, \aaa_i^\alpha \ccc_i^\beta x_i^l}\,.
\end{equation}
By antisymmetry in \eqref{Eqh3}, we can write 
\begin{equation}
\begin{aligned} \label{Eqh3bis}
  \br{\aaa_j^\gamma,\ccc_i^\beta}=&-\delta_{\beta \gamma}Z_{ji}+\aaa_j^\gamma Z_{ji}-
\frac12 \delta_{(j \neq i)}\frac{x_i+x_j}{x_i-x_j}\ccc_i^\beta (\aaa_i^\gamma-\aaa_j^\gamma)
+\delta_{(\gamma<\beta)}\aaa_j^\gamma \ccc_i^\beta \nonumber \\
&+\aaa_j^\gamma \sum_{\mu=1}^{\beta-1}\aaa_j^\mu (\ccc_i^\mu-\ccc_i^\beta) 
-\delta_{\beta \gamma} \sum_{\mu=1}^{\beta-1} \aaa_j^\mu \ccc_i^\mu 
-\frac12 \sum_{\sigma=1}^d o(\gamma,\sigma)\ccc_i^\beta 
(\aaa_i^\sigma \aaa_j^\gamma+\aaa_j^\sigma \aaa_i^\gamma) \,,
\end{aligned}
\end{equation}
so that the right-hand side  yields 
\begin{equation*}
 \begin{aligned}
  \eqref{Poi3}_{RHS}=&(k-l) \sum_{i=1}^n \aaa_i^\gamma \ccc_i^\epsilon \aaa_i^\alpha \ccc_i^\beta  x_i^{k+l}+
\sum_{i,j=1}^nx_j^k x_i^l \left(\br{\aaa_j^\gamma,\aaa_i^\alpha}\ccc_j^\epsilon \ccc_i^\beta + \br{\aaa_j^\gamma, \ccc_i^\beta}\ccc_j^\epsilon  \aaa_i^\alpha +  \br{ \ccc_j^\epsilon ,\aaa_i^\alpha}   \aaa_j^\gamma  \ccc_i^\beta  +\br{ \ccc_j^\epsilon , \ccc_i^\beta }   \aaa_j^\gamma  \aaa_i^\alpha  \right) \\
 =& 
 (k-l) \sum_{i=1}^n \aaa_i^\gamma \ccc_i^\epsilon \aaa_i^\alpha \ccc_i^\beta  x_i^{k+l} 
+\frac12 \sum_{\substack{i,j=1\\ i\neq j}}^n  x_j^k x_i^l \frac{x_j+x_i}{x_j-x_i} (\ccc_i^\epsilon \ccc_j^\beta \aaa_j^\gamma \aaa_i^\alpha+\ccc_j^\epsilon \ccc_i^\beta \aaa_i^\gamma \aaa_j^\alpha)   \\
&+\frac12 o(\alpha,\gamma) \sum_{i,j=1}^n  x_j^k x_i^l \ccc_j^\epsilon \ccc_i^\beta
(\aaa_j^\gamma \aaa_i^\alpha+\aaa_i^\gamma \aaa_j^\alpha)    
+\frac12 o(\epsilon,\beta) \sum_{i,j=1}^n  x_j^k x_i^l  \aaa_i^\alpha \aaa_j^\gamma 
(\ccc_i^\epsilon\ccc_j^\beta-\ccc_j^\epsilon \ccc_i^\beta)\\
&+[\delta_{(\gamma<\beta)} -\delta_{(\alpha<\epsilon)}] \sum_{i,j=1}^n  x_j^k x_i^l \aaa_j^\gamma \ccc_j^\epsilon \ccc_i^\beta \aaa_i^\alpha  \\
&+\delta_{\epsilon \alpha}\sum_{i,j=1}^n  x_j^k x_i^l  \ccc_i^\beta \aaa_j^\gamma \left( Z_{ij}+\sum_{\lambda=1}^{\epsilon-1} \aaa_i^\lambda \ccc_j^\lambda \right)  -\delta_{\beta \gamma}\sum_{i,j=1}^n  x_j^k x_i^l \ccc_j^\epsilon  \aaa_i^\alpha \left(Z_{ji}+\sum_{\mu=1}^{\beta-1} \aaa_j^\mu \ccc_i^\mu \right) \,.
 \end{aligned}
\end{equation*} 
This is obtained by simplifying terms without any non obvious manipulation. Now, remark that we can write 
$o(\epsilon,\alpha)=\delta_{(\epsilon<\alpha)}-\delta_{(\epsilon>\alpha)}=1-\delta_{\epsilon \alpha}-2\delta_{(\epsilon>\alpha)}$ 
so that 
\begin{equation*}
 \frac12 [o(\epsilon,\alpha)+\delta_{\epsilon\alpha}- o(\beta,\gamma)- \delta_{\beta\gamma}] 
=[\delta_{(\gamma<\beta)} -\delta_{(\alpha<\epsilon)}]\,.
\end{equation*}
We can also repeat the argument in \eqref{Sumrkl}, to get
\begin{equation*}
 \begin{aligned}
  &\frac12 \left(\sum_{r=1}^k-\sum_{r=1}^l \right) \sum_{i,j=1}^n \aaa_j^\alpha \ccc_i^\beta \aaa_i^\gamma \ccc_j^\epsilon
\left(  x_i^r  x_j^{k+l-r}+ x_i^{k+l-r} x_j^r \right)\\
=&(k-l) \sum_{i=1}^n \aaa_i^\gamma \ccc_i^\epsilon \aaa_i^\alpha \ccc_i^\beta  x_i^{k+l}  
-\frac12  \sum_{\substack{i,j=1\\ i\neq j}}^n x_i^{l} x_j^{k} \frac{x_i+x_j}{x_i-x_j}
\left(\aaa_j^\gamma\ccc_j^\beta \aaa_i^\alpha\ccc_i^\epsilon+\aaa_i^\gamma\ccc_i^\beta \aaa_j^\alpha\ccc_j^\epsilon \right)\,.
 \end{aligned}
\end{equation*}
Incorporating these two facts in $\eqref{Poi3}_{RHS}$ gives us 
\begin{equation*}
 \begin{aligned}
  \eqref{Poi3}_{RHS}=&
 \frac12 \left(\sum_{r=1}^k-\sum_{r=1}^l \right) \sum_{i,j=1}^n \aaa_j^\alpha \ccc_i^\beta \aaa_i^\gamma \ccc_j^\epsilon
\left(  x_i^r  x_j^{k+l-r}+ x_i^{k+l-r} x_j^r \right)   \\
&+\frac12 o(\alpha,\gamma) \sum_{i,j=1}^n  x_j^k x_i^l \ccc_j^\epsilon \ccc_i^\beta
(\aaa_j^\gamma \aaa_i^\alpha+\aaa_i^\gamma \aaa_j^\alpha)    
+\frac12 o(\epsilon,\beta) \sum_{i,j=1}^n  x_j^k x_i^l  \aaa_i^\alpha \aaa_j^\gamma 
(\ccc_i^\epsilon\ccc_j^\beta-\ccc_j^\epsilon \ccc_i^\beta)\\
&+ \frac12 [o(\epsilon,\alpha)+\delta_{\epsilon\alpha}- o(\beta,\gamma)- \delta_{\beta\gamma}] \sum_{i,j=1}^n  x_j^k x_i^l \aaa_j^\gamma \ccc_j^\epsilon \ccc_i^\beta \aaa_i^\alpha  \\
&+\delta_{\epsilon \alpha}\sum_{i,j=1}^n  x_j^k x_i^l  \ccc_i^\beta \aaa_j^\gamma \left( Z_{ij}+\sum_{\lambda=1}^{\epsilon-1} \aaa_i^\lambda \ccc_j^\lambda \right)  -\delta_{\beta \gamma}\sum_{i,j=1}^n  x_j^k x_i^l \ccc_j^\epsilon  \aaa_i^\alpha \left(Z_{ji}+\sum_{\mu=1}^{\beta-1} \aaa_j^\mu \ccc_i^\mu \right) \,.
 \end{aligned}
\end{equation*} 
This is nothing else than \eqref{xigg}, which is $\eqref{Poi3}_{LHS}$ as desired.  \qed

\subsection{Proof of Proposition \ref{ggFinal}}\label{Ann:C2}

We need the following lemma. 
\begin{lem} \label{ggMedium}
 For any $\epsilon,\gamma=1,\ldots,d$ and $j,k,l=1,\ldots,n$, 
\begin{subequations}
 \begin{align}
  \br{\ccc_j^\epsilon, f_{kl}}=&(Z_{kj}\ccc^\epsilon_l-Z_{jl}\ccc_j^\epsilon) 
+(Z_{lj}-Z_{kj}) f_{kl}+\frac12 \delta_{(j\neq k)}\frac{x_j+x_k}{x_j-x_k}\ccc_j^\epsilon 
( f_{jl}- f_{kl}) 
\nonumber \\  &+\frac12
\delta_{(j\neq l)}\frac{x_j+x_l}{x_j-x_l}(\ccc_j^\epsilon  f_{kl}+\ccc_l^\epsilon f_{kj})
+\frac12 \ccc_l^\epsilon  f_{kj}-\frac12 \ccc_j^\epsilon  f_{jl} \nonumber \\
&+ f_{kl} \sum_{\lambda=1}^{\epsilon-1}(\ccc_j^\lambda-\ccc_j^\epsilon)(\aaa_l^\lambda-\aaa_k^\lambda) \\
\br{\aaa_i^\gamma, f_{kl}}=& \aaa_i^\gamma Z_{il}-\aaa_k^\gamma Z_{il}
+\frac12 \delta_{(i\neq k)}\frac{x_i+x_k}{x_i-x_k}(\aaa_k^\gamma-\aaa_i^\gamma)( f_{il}- f_{kl})
\nonumber \\ &+\frac12
\delta_{(i\neq l)}\frac{x_i+x_l}{x_i-x_l} f_{kl}(\aaa_{l}^\gamma-\aaa_i^\gamma)
+\frac12 \aaa_i^\gamma  f_{il}-\frac12 \aaa_k^\gamma  f_{il} \nonumber \\
&+\frac12 \sum_{\sigma=1}^d o(\gamma,\sigma)  f_{kl}[\aaa_i^\gamma(\aaa_k^\sigma-\aaa_l^\sigma)
+\aaa_i^\sigma(\aaa_k^\gamma-\aaa_l^\gamma)]
 \end{align}
\end{subequations}
\end{lem}
\begin{proof}
As usual, we use the normalisation $\sum_\alpha \aaa_k^\alpha=1$, and we compute from \eqref{Eqh3}--\eqref{Eqh4} that 
\begin{equation}
 \begin{aligned} \label{cAC}
  \br{\ccc_j^\epsilon, f_{kl}}=&\sum_{\alpha=1}^d \left(\br{\ccc_j^\epsilon,\aaa^\alpha_k}\ccc^\alpha_l + \aaa^\alpha_k  \br{\ccc_j^\epsilon,\ccc^\alpha_l} \right)\\
=& 
\ccc^\epsilon_l Z_{kj} - f_{kl}  Z_{kj} +\frac12 \delta_{(j\neq k)}\frac{x_j+x_k}{x_j-x_k}  \ccc_j^\epsilon ( f_{jl}- f_{kl}) - \sum_{\alpha=1}^{\epsilon-1}  \ccc^\alpha_l \aaa_k^\alpha \ccc_j^\epsilon  \\
&- f_{kl}  \sum_{\lambda=1}^{\epsilon-1}\aaa_k^\lambda (\ccc_j^\lambda-\ccc_j^\epsilon) 
 + \ccc^\epsilon_l  \sum_{\lambda=1}^{\epsilon-1} \aaa_k^\lambda \ccc_j^\lambda 
 +\frac12 \sum_{\alpha=1}^d  \sum_{\kappa=1}^d o(\alpha,\kappa) \ccc^\alpha_l \ccc_j^\epsilon (\aaa_j^\kappa \aaa_k^\alpha+\aaa_k^\kappa \aaa_j^\alpha)  \\
&+\frac12 \delta_{(j\neq l)}\frac{x_j+x_l}{x_j-x_l} (\ccc_j^\epsilon  f_{kl} + \ccc_l^\epsilon  f_{kj})  
+ f_{kl}  Z_{lj}  - \ccc_j^\epsilon Z_{jl} \\
&
+\frac12 \sum_{\alpha=1}^d  o(\epsilon,\alpha)  \aaa^\alpha_k(\ccc_l^\epsilon\ccc_j^\alpha-\ccc_j^\epsilon \ccc_l^\alpha)  +  f_{kl} \sum_{\lambda=1}^{\epsilon-1} \aaa_l^\lambda (\ccc_j^\lambda-\ccc_j^\epsilon) 
-\sum_{\alpha=1}^d  \sum_{\mu=1}^{\alpha-1} \aaa^\alpha_k \ccc_j^\epsilon \aaa_j^\mu (\ccc_l^\mu-\ccc_l^\alpha) 
 \end{aligned}
\end{equation}
Our aim is to reduce some of these thirteen terms, mostly using properties of the ordering function $o(-,-)$. Summing the fourth, sixth and eleventh terms of \eqref{cAC} together yields  
\begin{equation*}
 \begin{aligned} 
 &- \sum_{\lambda=1}^{\epsilon-1}  \aaa_k^\lambda \ccc^\lambda_l  \ccc_j^\epsilon  
+  \sum_{\lambda=1}^{\epsilon-1} \aaa_k^\lambda \ccc_j^\lambda \ccc^\epsilon_l
+\frac12 \left[ \sum_{\lambda=\epsilon+1}^d - \sum_{\lambda=1}^{\epsilon-1} \right]   (\aaa^\lambda_k\ccc_j^\lambda\ccc_l^\epsilon - \aaa^\lambda_k \ccc_l^\lambda \ccc_j^\epsilon) \\
&=\frac12 \sum_{\substack{\lambda=1\\ \lambda\neq \epsilon}}^d    (\aaa^\lambda_k\ccc_j^\lambda\ccc_l^\epsilon - \aaa^\lambda_k \ccc_l^\lambda \ccc_j^\epsilon)
=\frac12 ( f_{kj}\ccc_l^\epsilon -  f_{kl} \ccc_j^\epsilon)\,.
 \end{aligned}
\end{equation*}
The fifth and twelfth terms of \eqref{cAC} give
\begin{equation*}
- f_{kl}  \sum_{\lambda=1}^{\epsilon-1}\aaa_k^\lambda (\ccc_j^\lambda-\ccc_j^\epsilon) 
  +  f_{kl} \sum_{\lambda=1}^{\epsilon-1} \aaa_l^\lambda (\ccc_j^\lambda-\ccc_j^\epsilon) 
=  f_{kl} \sum_{\lambda=1}^{\epsilon-1} (\aaa_l^\lambda-\aaa_k^\lambda) (\ccc_j^\lambda-\ccc_j^\epsilon) 
\end{equation*}
Relabelling indices, we transform the seventh terms from \eqref{cAC} as
\begin{equation*}
 \begin{aligned} 
\frac12 \left[ \sum_{\alpha=1}^d  \sum_{\kappa=\alpha+1}^d - \sum_{\alpha=1}^d  \sum_{\kappa=1}^{\alpha-1} \right]   
\ccc^\alpha_l \ccc_j^\epsilon (\aaa_j^\kappa \aaa_k^\alpha+\aaa_k^\kappa \aaa_j^\alpha)  
&=\frac12 \sum_{\alpha=1}^d  \sum_{\mu=1}^{\alpha-1}  \ccc_j^\epsilon
 (\ccc^\mu_l  \aaa_j^\alpha \aaa_k^\mu+\ccc^\mu_l  \aaa_k^\alpha \aaa_j^\mu
-\ccc^\alpha_l  \aaa_j^\mu \aaa_k^\alpha-\ccc^\alpha_l  \aaa_k^\mu \aaa_j^\alpha) \\
&=\frac12 \sum_{\alpha=1}^d  \sum_{\mu=1}^{\alpha-1}  \ccc_j^\epsilon (\ccc^\mu_l-\ccc^\alpha_l)
 (  \aaa_j^\alpha \aaa_k^\mu+  \aaa_k^\alpha \aaa_j^\mu)\,,
 \end{aligned}
\end{equation*}
which can be summed with the thirteen terms in \eqref{cAC} to yield 
\begin{equation*}
 \begin{aligned} 
&\frac12  \sum_{\alpha=1}^d  \sum_{\mu=1}^{\alpha-1}   \ccc_j^\epsilon (\ccc_l^\alpha-\ccc_l^\mu) 
\left(2 \aaa^\alpha_k \aaa_j^\mu-\aaa_j^\alpha \aaa_k^\mu- \aaa_k^\alpha \aaa_j^\mu\right)
=\frac12  \sum_{\alpha=1}^d  \sum_{\mu=1}^{\alpha-1}   \ccc_j^\epsilon (\ccc_l^\alpha-\ccc_l^\mu) 
\left(\aaa^\alpha_k \aaa_j^\mu-\aaa_j^\alpha \aaa_k^\mu\right) \\
=&\frac12  \sum_{\alpha=1}^d  \sum_{\mu=1}^{\alpha-1} \ccc_j^\epsilon (\ccc_l^\alpha-\ccc_l^\mu) \aaa^\alpha_k \aaa_j^\mu
-\frac12  \sum_{\alpha=1}^d  \sum_{\mu=\alpha+1}^{d}   \ccc_j^\epsilon (\ccc_l^\mu-\ccc_l^\alpha) \aaa_j^\mu \aaa_k^\alpha \\
=&\frac12  \sum_{\alpha=1}^d  \sum_{\substack{\mu=1\\ \mu \neq \alpha}}^{d} \ccc_j^\epsilon (\ccc_l^\alpha-\ccc_l^\mu) \aaa^\alpha_k \aaa_j^\mu 
=\frac12  \sum_{\alpha=1}^d  \sum_{\mu=1}^{d} \ccc_j^\epsilon (\aaa^\alpha_k \ccc_l^\alpha \aaa_j^\mu-\aaa_j^\mu \ccc_l^\mu \aaa^\alpha_k)  
= \frac12  \ccc_j^\epsilon \left( f_{kl} -  f_{jl} \right)\,.
 \end{aligned}
\end{equation*}
Introducing the different terms back in \eqref{cAC}, we find 
\begin{equation*}
 \begin{aligned} 
  \br{\ccc_j^\epsilon, f_{kl}}=&
(\ccc^\epsilon_l Z_{kj}  - \ccc_j^\epsilon Z_{jl}) + (Z_{lj} -  Z_{kj})  f_{kl} 
  +\frac12 \delta_{(j\neq k)}\frac{x_j+x_k}{x_j-x_k}  \ccc_j^\epsilon ( f_{jl}- f_{kl})  \\
&+\frac12 \delta_{(j\neq l)}\frac{x_j+x_l}{x_j-x_l} (\ccc_j^\epsilon  f_{kl} + \ccc_l^\epsilon  f_{kj})  
+\frac12 (\ccc_l^\epsilon  f_{kj} - \ccc_j^\epsilon   f_{jl} ) \\
&+  f_{kl} \sum_{\lambda=1}^{\epsilon-1}(\ccc_j^\lambda-\ccc_j^\epsilon)  (\aaa_l^\lambda-\aaa_k^\lambda) \,,
 \end{aligned}
\end{equation*}
as desired. 
For the second identity,  we need \eqref{Eqh2} and\eqref{Eqh3bis}, then the same kind of manipulations allow to find 
$\br{\aaa_i^\gamma, f_{kl}}$.  
\end{proof}


To establish Proposition \ref{ggFinal}, we have from Lemma \ref{ggMedium} and the identity  $\sum_{\gamma=1}^d \aaa_i^\gamma=1$ that 
\begin{equation}
 \begin{aligned}  \label{acAC}
 \br{f_{ij},f_{kl}}=& \sum_{\gamma=1}^d\left(\br{\aaa_i^\gamma,  f_{kl}}\ccc_j^\gamma +\aaa_i^\gamma \br{\ccc_j^\gamma,  f_{kl}} \right) \\
=&( f_{ij}-  f_{kj}) Z_{il} +
 \frac12 \delta_{(i\neq k)}\frac{x_i+x_k}{x_i-x_k}
 ( f_{kj}- f_{ij}) ( f_{il}- f_{kl})   \\
& +\frac12 \delta_{(i \neq l)}\frac{x_i+x_l}{x_i-x_l} ( f_{lj}- f_{ij})   f_{kl} 
+\frac12  f_{ij}   f_{il}  - \frac12  f_{kj}  f_{il} \\
&+\frac12 \sum_{\gamma=1}^d \sum_{\sigma=1}^d o(\gamma,\sigma)  f_{kl}  \big(\aaa_i^\gamma \ccc_j^\gamma (\aaa_k^\sigma-\aaa_l^\sigma)+\aaa_i^\sigma  (\aaa_k^\gamma  - \aaa_l^\gamma) \ccc_j^\gamma\big) \\
&+ Z_{kj}  f_{il}-Z_{jl}  f_{ij}
 +(Z_{lj}-Z_{kj}) f_{kl} \\
&+\frac12 \delta_{(j\neq k)}\frac{x_j+x_k}{x_j-x_k}  f_{ij}  ( f_{jl}- f_{kl}) 
+\frac12 \delta_{(j\neq l)}\frac{x_j+x_l}{x_j-x_l}( f_{ij}  f_{kl}+ f_{il} f_{kj})\\
 &+\frac12  f_{il}  f_{kj}  -\frac12  f_{ij}  f_{jl} + f_{kl} \sum_{\gamma=1}^d \sum_{\lambda=1}^{\gamma-1}\aaa_i^\gamma(\ccc_j^\lambda-\ccc_j^\gamma)(\aaa_l^\lambda-\aaa_k^\lambda)\,.
 \end{aligned}
\end{equation}
The sums in the third line of \eqref{acAC} can be re-expressed as follows :  
\begin{equation*}
 \begin{aligned} 
&\frac12 f_{kl}\, \sum_{\gamma=1}^d \left[\sum_{\sigma=\gamma+1}^d- \sum_{\sigma=1}^{\gamma-1}\right]    \big(\aaa_i^\gamma \ccc_j^\gamma (\aaa_k^\sigma-\aaa_l^\sigma)+\aaa_i^\sigma  (\aaa_k^\gamma  - \aaa_l^\gamma) \ccc_j^\gamma\big) \\
=& 
\frac12  f_{kl}\, \sum_{\gamma=1}^d \left[\sum_{\sigma=\gamma+1}^d- \sum_{\sigma=1}^{\gamma-1}\right]   \big(\aaa_i^\gamma \ccc_j^\gamma (\aaa_k^\sigma-\aaa_l^\sigma) - \aaa_i^\gamma (\aaa_k^\sigma  - \aaa_l^\sigma) \ccc_j^\sigma\big) \,,
 \end{aligned}
\end{equation*}
after swapping the labels $\sigma \leftrightarrow \gamma$ for the second term in the sums. This is nothing else that 
\begin{equation*}
 \frac12  f_{kl}\, \sum_{\gamma=1}^d \left[\sum_{\sigma=\gamma+1}^d- \sum_{\sigma=1}^{\gamma-1}\right]  \aaa_i^\gamma  ( \ccc_j^\gamma - \ccc_j^\sigma) (\aaa_k^\sigma-\aaa_l^\sigma)\,.
\end{equation*}
Summing with the last term of \eqref{acAC}, we get 
\begin{equation*}
\begin{aligned}
  &\frac12  f_{kl}\, \sum_{\gamma=1}^d \left[\sum_{\lambda=\gamma+1}^d- \sum_{\lambda=1}^{\gamma-1}\right]  \aaa_i^\gamma  (\ccc_j^\gamma - \ccc_j^\lambda) (\aaa_k^\lambda-\aaa_l^\lambda)
+ f_{kl} \sum_{\gamma=1}^d \sum_{\lambda=1}^{\gamma-1}\aaa_i^\gamma(\ccc_j^\lambda-\ccc_j^\gamma)(\aaa_l^\lambda-\aaa_k^\lambda)\\
=&\frac12  f_{kl}\, \sum_{\gamma=1}^d \sum_{\substack{\lambda=1\\\lambda\neq\gamma}}^d  \aaa_i^\gamma  (\ccc_j^\gamma - \ccc_j^\lambda) (\aaa_k^\lambda-\aaa_l^\lambda)
=\frac12  f_{kl}\, \sum_{\gamma=1}^d \sum_{\lambda=1}^d  
\aaa_i^\gamma\left(\ccc_j^\gamma \aaa_k^\lambda - \ccc_j^\gamma \aaa_l^\lambda - \aaa_k^\lambda \ccc_j^\lambda+\aaa_l^\lambda \ccc_j^\lambda \right)  \\
&=\frac12  f_{kl} \big( f_{ij} -  f_{ij} -  f_{kj} +  f_{lj}    \big) 
= \frac12  f_{kl} f_{lj}  - \frac12  f_{kl}  f_{kj}\,. 
\end{aligned}
\end{equation*}
We can thus rewrite \eqref{acAC} as 
\begin{equation*}
 \begin{aligned}  
 \br{f_{ij},f_{kl}}
=&f_{ij} Z_{il} - f_{kj} Z_{il}  +
 \frac12 \delta_{(i\neq k)}\frac{x_i+x_k}{x_i-x_k}
(f_{kj}f_{il}- f_{kj}f_{kl} -f_{ij} f_{il} +f_{ij} f_{kl}) \\
& +\frac12 \delta_{(i \neq l)}\frac{x_i+x_l}{x_i-x_l} (f_{lj}f_{kl}-f_{ij}f_{kl})   
+\frac12 f_{ij}  f_{il}  - \frac12 f_{kj} f_{il} +\frac12 f_{kl}f_{lj}  - \frac12 f_{kl} f_{kj}  \\
&+ Z_{kj} f_{il}-Z_{jl} f_{ij}+ Z_{lj}f_{kl}  - Z_{kj}f_{kl}
+\frac12 \delta_{(j\neq k)}\frac{x_j+x_k}{x_j-x_k}  (f_{ij} f_{jl} - f_{ij} f_{kl})  \\  
&+\frac12 \delta_{(j\neq l)}\frac{x_j+x_l}{x_j-x_l}(f_{ij} f_{kl}+f_{il}f_{kj})
 +\frac12 f_{il} f_{kj}-\frac12 f_{ij} f_{jl} \,.
 \end{aligned}
\end{equation*}
Now, remark that we can rearrange terms to obtain
\begin{equation*}
 \begin{aligned}  
 \br{f_{ij},f_{kl}}
=& f_{ij} \left(Z_{il}+\frac12 f_{il}\right) - f_{kj}  \left(Z_{il} +\frac12 f_{il} \right)
+ f_{il} \left(Z_{kj} + \frac12 f_{kj} \right) - f_{ij} \left(Z_{jl} + \frac12 f_{jl} \right)  \\
& + f_{kl} \left(Z_{lj} + \frac12 f_{lj}\right) -  f_{kl}\left(Z_{kj} + \frac12 f_{kj} \right) 
+\frac12 \delta_{(i \neq l)}\frac{x_i+x_l}{x_i-x_l} (f_{lj}f_{kl}-f_{ij}f_{kl})    \\
&+ \frac12 \delta_{(i\neq k)}\frac{x_i+x_k}{x_i-x_k} (f_{kj}f_{il}- f_{kj}f_{kl} -f_{ij} f_{il} +f_{ij} f_{kl}) \\
&+\frac12 \delta_{(j\neq k)}\frac{x_j+x_k}{x_j-x_k}  (f_{ij} f_{jl} 
- f_{ij} f_{kl}) +\frac12 \delta_{(j\neq l)}\frac{x_j+x_l}{x_j-x_l}(f_{ij} f_{kl}+f_{il}f_{kj})\,.
 \end{aligned}
\end{equation*}
After that, a simple rearrangement using 
\begin{equation}\label{Eq:Z12f} 
Z_{ij}+\frac12 f_{ij}=\frac12 \frac{x_i+qx_j}{x_i-q x_j} f_{ij}
\end{equation}
leads to \eqref{Eq:ffAF}. This proves the first claim of Proposition \ref{ggFinal}.
To prove the second claim, we use Lemma \ref{ggMedium} to get  
\begin{equation*}
 \begin{aligned}
     \br{\aaa_i^\gamma,f_{kk}}=& \frac12 \delta_{(i\neq k)}\frac{x_i+x_k}{x_i-x_k}(\aaa_k^\gamma-\aaa_i^\gamma)f_{ik} 
+ (\aaa_i^\gamma-\aaa_k^\gamma) \left(Z_{ik}+\frac12f_{ik}\right)\,, \\
\br{\ccc_j^\epsilon,f_{kk}}=&\frac12\delta_{(j\neq k)}\frac{x_j+x_k}{x_j-x_k}(\ccc_j^\epsilon f_{jk}+\ccc_k^\epsilon f_{kj}) 
+ \ccc_k^\epsilon \left(Z_{kj} + \frac12 f_{kj}\right)  - \ccc_j^\epsilon \left(Z_{jk} + \frac12 f_{jk}\right)\,.
 \end{aligned}
\end{equation*} 
Using \eqref{Eq:Z12f} we get, after summation over $k$, the relations \eqref{trz2}--\eqref{trz3}. The remaining relation \eqref{trz1} is obvious. 
\qed

\subsection{Proof of Lemma \ref{Lem:ttMQV}} \label{Ann:A1}

Let  $u\in \{y,z\}$, and remark that we can write $\dgal{u,u}=-\frac12 [u^2 \otimes e_0 - e_0 \otimes u^2]$, together with
\begin{equation} \label{Equvw}
  \dgal{u, w_\alpha}= \frac12 e_{0}\otimes uw_\alpha-\frac12 u\otimes w_\alpha\,,\quad 
\dgal{u, v_\alpha}= \frac12 v_\alpha u\otimes e_0-\frac12 v_\alpha\otimes u\,.
\end{equation}  
Now,  consider the elements $w_\alpha  v_\beta u^l\in \Aal$ for any $l\in \N$ and $\alpha,\beta=1,\ldots,d$. The following statement holds in $\Aal/[\Aal,\Aal]$ if $u=y$, and $\Aalt/[\Aalt,\Aalt]$ if $u=z$. 
\begin{lem} \label{Lem:tt} For any  $k,l\geq 1$ and $\alpha,\beta,\gamma,\epsilon=1,\ldots,d$, we have  that  
\begin{equation}
 \begin{aligned} 
\br{u^k,w_\alpha v_\beta u^l}=&0\,, \\
    \br{w_\gamma v_\epsilon u^k,w_\alpha v_\beta u^l}=& \frac{1}{2}
\left[o(\gamma,\beta)+o(\epsilon,\alpha)-o(\epsilon,\beta)-o(\gamma,\alpha)\right]
w_\alpha v_\epsilon u^k\,w_\gamma v_\beta u^l \\
&+\frac{1}{2}o(\gamma,\beta)\,\, w_\alpha v_\beta\, w_\gamma v_\epsilon u^{k+l}
+\frac{1}{2}o(\epsilon,\alpha)\,\, w_\alpha v_\beta u^{k+l}\,w_\gamma v_\epsilon \\
&-\frac{1}{2}o(\epsilon,\beta)\,\, w_\alpha v_\beta u^{k}\,w_\gamma v_\epsilon u^l
-\frac{1}{2}o(\gamma,\alpha)\,\, w_\alpha v_\beta u^{l}\,w_\gamma v_\epsilon u^{k} \\
&-\delta_{\gamma\beta}\left[w_\alpha v_\epsilon u^{k+l}
+\frac{1}{2}  w_\alpha v_\beta\,w_\gamma v_\epsilon u^{k+l}
+\frac{1}{2}w_\alpha v_\epsilon u^k\,w_\gamma v_\beta u^l \right]  \\
&+\delta_{\alpha\epsilon}\left[w_\gamma v_\beta u^{k+l}
+\frac{1}{2}  w_\alpha v_\beta u^{k+l}\,w_\gamma v_\epsilon 
+\frac{1}{2}w_\alpha v_\epsilon u^k\,w_\gamma v_\beta u^l \right] \\
&-\frac{1}{2}  \,\left[  \sum_{\tau=1}^{k-1}  w_\alpha v_\beta u^{k+l-\tau} w_\gamma v_\epsilon u^{\tau}
+ \sum_{\sigma=1}^l w_\alpha v_\beta u^{k+\sigma} w_\gamma v_\epsilon u^{l-\sigma}  \right] \\
&+\frac{1}{2}  \,\left[
\sum_{\sigma=1}^{l-1}  w_\alpha v_\beta u^{\sigma} w_\gamma v_\epsilon u^{k+l-\sigma}
+ \sum_{\tau=1}^k  w_\alpha v_\beta u^{k-\tau} w_\gamma v_\epsilon u^{l+\tau}\right]   \label{Eq:tt} 
\end{aligned}
\end{equation}
The same holds without the sums if $k=0$, $l=0$ or both $k=l=0$. 
\end{lem} 
We delay the proof of this lemma until \ref{Ann:Lemtt} to explain how we can conclude from this result. 
Denote by $U$ the matrix representing $u\in \{z,y\}$.  
Then $t_{\alpha \epsilon}^l=\tr( W_\alpha V_\epsilon U^l)$, and Lemma \ref{Lem:ttMQV} is deduced from Lemma \ref{Lem:tt} and \eqref{relInv} by using that 
$\tr(W_\alpha V_\beta U^k W_\gamma  V_\epsilon U^l)=(V_\epsilon U^l W_\alpha) ( V_\beta U^k W_\gamma)=t_{\alpha\epsilon}^l t_{\gamma\beta}^k$. \qed

\begin{rem} \label{Rem:tt}
A similar result also holds for $u\in \{x,e_0+xy\}$, and with $u=x+y^{-1}$ if we decide to localise at $y$. We have in those cases $\dgal{u,u}=+\frac12 [u^2 \otimes e_0 - e_0 \otimes u^2]$, and \eqref{Equvw} also holds. Then, Lemma \ref{Lem:tt} can also be proved, except that we need to change the signs in front of the last two lines in \eqref{Eq:tt}.  We do not discuss these cases any further. 
\end{rem}

\subsubsection{Proof of Lemma \ref{Lem:tt}} \label{Ann:Lemtt}

First, we note from the discussion at the beginning of \ref{Ann:Brack} that 
\begin{equation*}
    \dgal{u^k, u^l}
=-\frac{1}{2}\sum_{\tau=1}^k\sum_{\sigma=1}^l 
\left( u^{k-\tau+\sigma+1} \otimes u^{l-\sigma+\tau-1} - u^{k-\tau+\sigma-1} \otimes u^{l-\sigma+\tau+1}\right) \,.
\end{equation*}
Next, using \eqref{Equvw}, 
we remark  that 
\begin{equation} \label{Lod1u}
 \begin{aligned} 
  \dgal{u, w_\alpha v_\beta}=& w_\alpha \dgal{u,  v_\beta} + \dgal{u, w_\alpha } v_\beta \\
=& \frac12 \left(w_\alpha v_\beta u\otimes e_0 - w_\alpha v_\beta\otimes u 
+e_{0}\otimes uw_\alpha v_\beta - u\otimes w_\alpha v_\beta \right)\,.
 \end{aligned}
\end{equation}
Therefore
 \begin{equation}
  \begin{aligned} 
   \dgal{u^k,w_\alpha v_\beta u^l} 
=& \frac12 \sum_{\tau=1}^k   
\Big(w_\alpha v_\beta u^{k-\tau+1} \otimes u^{l+\tau-1} - w_\alpha v_\beta u^{k-\tau}\otimes u^{l+\tau}  \\
& \qquad \quad +u^{k-\tau}\otimes u^{\tau}w_\alpha v_\beta u^{l} - u^{k-\tau+1}\otimes u^{\tau-1}w_\alpha v_\beta u^{l} \Big)\\
&-\frac12  \sum_{\tau=1}^k \sum_{\sigma=1}^l  
 \left(w_\alpha v_\beta u^{k-\tau+\sigma+1}  \otimes  u^{l-\sigma+\tau -1}
-w_\alpha v_\beta u^{k-\tau+\sigma-1}  \otimes u^{l-\sigma+\tau+1}  \right)\,. \label{Lod2u}
  \end{aligned}
 \end{equation}
After application of the multiplication map, we get $0$ and the first equality follows. To prove that \eqref{Eq:tt} holds, write 
\begin{equation} 
 \dgal{w_\gamma v_\epsilon u^k, w_\alpha v_\beta u^l}=
w_\gamma v_\epsilon \ast  \dgal{u^k,w_\alpha v_\beta u^l} 
+ w_\alpha v_\beta \dgal{w_\gamma v_\epsilon, u^l} \ast u^k 
+\dgal{w_\gamma v_\epsilon,w_\alpha v_\beta } u^l \ast u^k \,.
\end{equation}
  Again, we begin by reducing the two first terms. Using \eqref{Lod1u} and \eqref{Lod2u} gives 
\begin{equation*}
 \begin{aligned}
T:=&
w_\gamma v_\epsilon \ast  \dgal{u^k,w_\alpha v_\beta u^l} 
- \sum_{\sigma=1}^l w_\alpha v_\beta u^{\sigma-1} \dgal{u,w_\gamma v_\epsilon}^\circ u^{l-\sigma} \ast u^k \\
=&- \frac12 \sum_{\tau=1}^k \sum_{\sigma=1}^l  
 \left(w_\alpha v_\beta u^{k-\tau+\sigma+1}  \otimes w_\gamma v_\epsilon  u^{l+\tau-\sigma-1}
-w_\alpha v_\beta  u^{k-\tau+\sigma-1}  \otimes w_\gamma v_\epsilon u^{l+\tau-\sigma+1} \right) \\
&+ \frac12 \sum_{\tau=1}^k   
\left(w_\alpha v_\beta u^{k-\tau+1} \otimes w_\gamma v_\epsilon u^{l+\tau-1} - w_\alpha v_\beta u^{k-\tau}\otimes w_\gamma v_\epsilon u^{l+\tau}\right) \\
& + \frac12 \sum_{\tau=1}^k   \left( u^{k-\tau}\otimes w_\gamma v_\epsilon u^{\tau}w_\alpha v_\beta u^{l} -  u^{k-\tau+1}\otimes w_\gamma v_\epsilon u^{\tau-1}w_\alpha v_\beta u^{l} \right) \\
&-\frac12 \sum_{\sigma=1}^l \left(w_\alpha v_\beta u^\sigma w_\gamma v_\epsilon u^k \otimes u^{l-\sigma} - w_\alpha v_\beta u^{\sigma-1} w_\gamma v_\epsilon u^k \otimes u^{l-\sigma +1} \right) \\
&-\frac12 \sum_{\sigma=1}^l \left(w_\alpha v_\beta u^{k+\sigma-1}  \otimes w_\gamma v_\epsilon u^{l-\sigma+1} - w_\alpha v_\beta u^{k+\sigma} \otimes w_\gamma v_\epsilon u^{l-\sigma} \right)\,.
 \end{aligned}
\end{equation*}
This gives, after multiplication and modulo commutators 
\begin{equation*}
 \begin{aligned}
 &m \circ T= -\frac12  \sum_{\tau=1}^k \sum_{\sigma=1}^l  
 \left(w_\alpha v_\beta u^{k-\tau+\sigma+1}  w_\gamma v_\epsilon  u^{l+\tau-\sigma-1}
-w_\alpha v_\beta  u^{k-\tau+\sigma-1}  w_\gamma v_\epsilon u^{l+\tau-\sigma+1} \right) \,,
 \end{aligned}
\end{equation*}
because the last four sums cancel out. Relabelling indices, we write 
\begin{equation*}
 \begin{aligned}
 m \circ T=& - \frac12  \left[\sum_{\tau=1}^{k-1} \sum_{\sigma=l}+\sum_{\tau=0} \sum_{\sigma=1}^l - \sum_{\tau=k} \sum_{\sigma=1}^{l-1} - \sum_{\tau=1}^{k} \sum_{\sigma=0}  \right]  
 w_\alpha v_\beta u^{k-\tau+\sigma}  w_\gamma v_\epsilon  u^{l+\tau-\sigma} \\
=&-\frac{1}{2}  \,\left[  \sum_{\tau=1}^{k-1}  w_\alpha v_\beta u^{k+l-\tau} w_\gamma v_\epsilon u^{\tau}
+ \sum_{\sigma=1}^l w_\alpha v_\beta u^{k+\sigma} w_\gamma v_\epsilon u^{l-\sigma}  \right] \\
&+\frac{1}{2}  \,\left[
\sum_{\sigma=1}^{l-1}  w_\alpha v_\beta u^{\sigma} w_\gamma v_\epsilon u^{k+l-\sigma}
+ \sum_{\tau=1}^k  w_\alpha v_\beta u^{k-\tau} w_\gamma v_\epsilon u^{l+\tau}\right] 
 \end{aligned}
\end{equation*}
It remains to compute $\dgal{w_\gamma v_\epsilon,w_\alpha v_\beta }$. We can find from \eqref{tadidv}--\eqref{tadidu}  
\begin{equation*}
 \begin{aligned} 
\dgal{w_\gamma v_\epsilon,w_\alpha v_\beta }= &
-\frac12 o(\gamma,\alpha) \left(w_\alpha v_\epsilon \otimes w_\gamma v_\beta + w_\gamma v_\epsilon \otimes w_\alpha v_\beta \right) -\frac12 o(\beta,\gamma) \left(w_\alpha v_\beta w_\gamma v_\epsilon \otimes e_0 + w_\alpha v_\epsilon \otimes w_\gamma v_\beta \right) \\
&-\delta_{\beta \gamma} \left(w_\alpha v_\epsilon \otimes e_0 + \frac12 w_\alpha v_\epsilon \otimes w_\gamma v_\beta + \frac12 w_\alpha v_\beta w_\gamma v_\epsilon \otimes e_0 \right) \\
&+\delta_{\alpha \epsilon} \left(e_0 \otimes w_\gamma v_\beta + \frac12 w_\alpha v_\epsilon \otimes w_\gamma v_\beta + \frac12 e_0 \otimes w_\gamma v_\epsilon w_\alpha v_\beta \right) \\
&+\frac12 o(\epsilon,\alpha) \left( e_0 \otimes w_\gamma v_\epsilon w_\alpha v_\beta + w_\alpha v_\epsilon \otimes w_\gamma v_\beta \right) - \frac12 o(\epsilon,\beta) \left(w_\alpha v_\beta \otimes w_\gamma v_\epsilon + w_\alpha v_\epsilon \otimes w_\gamma v_\beta \right)\,.
 \end{aligned}
\end{equation*}
By applying the multiplication map $m$ on $\dgal{w_\gamma v_\epsilon,w_\alpha v_\beta } u^l \ast u^k $ 
we get
\begin{equation*}
 \begin{aligned} 
m\circ(\dgal{w_\gamma v_\epsilon,w_\alpha v_\beta } u^l \ast u^k )= &
-\frac12 o(\gamma,\alpha) \left(w_\alpha v_\epsilon u^k   w_\gamma v_\beta u^l + w_\gamma v_\epsilon u^k   w_\alpha v_\beta u^l \right) \\
& -\frac12 o(\beta,\gamma) \left(w_\alpha v_\beta w_\gamma v_\epsilon u^k   u^l + w_\alpha v_\epsilon u^k   w_\gamma v_\beta u^l \right) \\
&-\delta_{\beta \gamma} \left(w_\alpha v_\epsilon  u^k   u^l + \frac12 w_\alpha v_\epsilon u^k   w_\gamma v_\beta u^l + \frac12 w_\alpha v_\beta w_\gamma v_\epsilon u^k   u^l \right) \\
&+\delta_{\alpha \epsilon} \left(u^k   w_\gamma v_\beta u^l + \frac12 w_\alpha v_\epsilon u^k   w_\gamma v_\beta u^l + \frac12 u^k   w_\gamma v_\epsilon w_\alpha v_\beta u^l \right) \\
&+\frac12 o(\epsilon,\alpha) \left( u^k   w_\gamma v_\epsilon w_\alpha v_\beta u^l + w_\alpha v_\epsilon u^k   w_\gamma v_\beta u^l \right) \\
& - \frac12 o(\epsilon,\beta) \left(w_\alpha v_\beta u^k   w_\gamma v_\epsilon u^l + w_\alpha v_\epsilon u^k   w_\gamma v_\beta u^l \right)\,.
 \end{aligned}
\end{equation*}
Adding $m \circ T$ to this last expression finishes the proof. \qed


\subsection{Computations with $S$} \label{Ann:eta}

To simplify notations in this appendix, we denote by $s$ the element $s_d$ given by \eqref{salpha0} with $\alpha=d$.

\subsubsection{Some brackets} 
Recall that \eqref{central} holds if we can show that the relation 
$\{\tr S^k, t_{\alpha\beta}^i\}=0$ on $\Rep(\Aalt, \bar{\alpha})$. Moreover, we have that $t_{\alpha \beta}^i$ and $\tr S^k$ are traces of the matrices representing the elements $w_\alpha v_\beta z^i$ and $s^k$. Using \eqref{relInv}, the desired relation follows from the following lemma.
\begin{lem}
  For any $k,l\in \N$ and $1\leq \alpha,\beta \leq d$, $\br{s^k,w_\alpha v_\beta z^i}=0$ in $\Aalt/[\Aalt,\Aalt]$. 
\end{lem}
\begin{proof}
 We use  \eqref{br:sz} and \eqref{br:svw} to get 
\begin{equation*}
\begin{aligned}
\dgal{s,w_\alpha v_\beta z^i}=& \dgal{s,w_\alpha} v_\beta z^i + w_\alpha \dgal{s,v_\beta} z^i 
+\sum_{\tau=1}^i w_\alpha v_\beta z^{\tau-1} \dgal{s,z} z^{i-\tau}    \\
=&\frac12 (s \otimes w_\alpha v_\beta z^i + e_0 \otimes s w_\alpha v_\beta z^i 
- w_\alpha v_\beta z^i s \otimes e_0 - w_\alpha v_\beta z^i \otimes s)\,.
\end{aligned}
\end{equation*}
Hence, 
\begin{equation*}
  \br{s^k,w_\alpha v_\beta z^i} = 
m \circ \left(\sum_{\tau=1}^k s^{\tau-1} \ast \dgal{s,w_\alpha v_\beta z^i}\ast s^{k-\tau} \right) 
= l\, (s^k w_\alpha v_\beta z^i - w_\alpha v_\beta z^i s^k)\,,
\end{equation*}
which vanishes modulo commutators. 
\end{proof}

\subsubsection{Proof of Theorem \ref{Thm:TadInv}} \label{Ann:Spectral}

Consider the following result.

\begin{lem} \label{PropInvPhi}
Let $z_{\eta}=z+\eta s$ for arbitrary $\eta \in \CC$ playing the role of a spectral parameter. Then, 
\begin{equation*}
\br{z_{\mu}^k,z_{\eta}^l}=0 \quad \text{ mod }[\Aalt,\Aalt], \qquad \text{ for any }\mu,\eta\in \CC\,, k,l \in \N^\times\,.
\end{equation*}
\end{lem}
Using this lemma together with \eqref{relInv} we get the property $\br{\tr Z_{\mu}^k\,,\, \tr Z_{\eta}^l}=0$ since the matrix $Z_\eta$ represents the element $z_\eta\in \Aalt$. Thus, Theorem \ref{Thm:TadInv} follows directly from  this intermediate result.

To prove Lemma \ref{PropInvPhi}, we use the derivation properties of the double bracket to see that 
\begin{equation} \label{br:spec}
\frac{1}{kl}\br{z_{\mu}^k,z_{\eta}^l}= \dgal{z_{\mu},z_{\eta}}' z_\mu^{k-1} \dgal{z_{\mu},z_{\eta}}'' z_\eta^{l-1} \quad \text{ mod }[\Aalt,\Aalt].
\end{equation}
Hence, the first step is to compute the double bracket $\dgal{z_{\mu},z_{\eta}}$. Using that $z_{\eta}=z+\eta s$ and the same with $\mu$, we need $\dgal{z,z},\dgal{s,z}$ given in \eqref{tadidaZ},\eqref{br:sz} together with 
 \begin{equation} \label{br:ss}
   \dgal{s,s}=\frac12 (e_0 \otimes s^2 - s^2 \otimes e_0)\,,
 \end{equation}
 which is a special case of Lemma \ref{Lem:sBis}. 
This yields 
\begin{equation*}
  \begin{aligned}
    \dgal{z_{\mu},z_{\eta}} =&\dgal{z,z}+\mu \dgal{s,z}- \eta \dgal{s,z}^\circ + \mu \eta \dgal{s,s}\\
=&\frac12 (e_0 \otimes z^2- z^2 \otimes e_0) + \frac12 \mu (s \otimes z - z s \otimes e_0 + e_0 \otimes sz - z \otimes s) \\
&-\frac12 \eta (z \otimes s - e_0 \otimes zs + sz \otimes e_0 - s \otimes z) 
+\frac12 \mu \eta (e_0 \otimes s^2 - s^2 \otimes e_0)\,.
  \end{aligned}
\end{equation*}
By grouping terms together, we can write 
\begin{equation*}
  \begin{aligned}
    \dgal{z_{\mu},z_{\eta}} =&\frac12 (e_0 \otimes z z_\eta - z z_\mu \otimes e_0) 
+\frac12 \mu (e_0 \otimes s z_\eta + s \otimes z - z \otimes s)
- \frac12 \eta (s z_\mu \otimes e_0 + z \otimes s - s \otimes z)\,.
  \end{aligned}
\end{equation*}
We can use that $\mu s = z_\mu - z$ for the terms with a factor $\mu$, and do the same with $\eta$. We can write in this way 
\begin{equation*}
  \begin{aligned}
    \dgal{z_{\mu},z_{\eta}} =&\frac12 (e_0 \otimes z_\mu z_\eta - z_\eta z_\mu \otimes e_0) 
+\frac12 (z_\mu \otimes z - z \otimes z_\mu) + \frac12 (z_\eta \otimes z - z \otimes z_\eta)\,.
  \end{aligned}
\end{equation*}
Substituting back in \eqref{br:spec}, we get modulo commutators 
\begin{equation*} 
\frac{1}{kl}\br{z_{\mu}^k,z_{\eta}^l}= 
\frac12( z_\mu^{k}z  z_\eta^{l-1} - z z_\mu^{k}  z_\eta^{l-1})
+ \frac12 (z_\mu^{k-1}  z z_\eta^{l} - z z_\mu^{k-1}  z_\eta^{l}) \,.
\end{equation*}
This is clearly zero when $\mu=\eta$. If $\mu \neq \eta$, we can substitute $z=\frac{1}{\mu-\eta}(\mu z_\eta - \eta z_\mu)$ in the two groups of terms, which then vanish modulo commutators.  \qed





\begin{thebibliography}{AKSM}
\bibitem{Su} Suris, Yu.~B.: {\it Why is the Ruijsenaars-Schneider hierarchy governed by the same $R$-operator as the Calogero-Moser one?}.  Phys. Lett. A 225~(4-6), 253--262 (1997)
\end{thebibliography}
\end{document}